\documentclass[a4paper]{article}
\usepackage{enumitem}
\usepackage{calc}
\usepackage{hyperref}

\usepackage{lmodern} 

\usepackage{graphicx} 

\usepackage{amsmath, accents,mathtools}
\usepackage{amssymb,tikz}
\usepackage{pict2e}
\usepackage{amsthm}
\usepackage{amscd}
\usepackage{mathrsfs}
\usepackage{todonotes}

\usetikzlibrary{calc} 

\usepackage{yfonts}

\usepackage{marvosym}

\newtheorem{theorem}			     {Theorem} [section]

\newtheorem{lemma}	      [theorem]  {Lemma}		
\theoremstyle{definition}

\newtheorem{remark} {Remark}

\newcommand{\C}{\mathbb{C}}

\usepackage{tikz}
\usetikzlibrary{arrows}
\usetikzlibrary{decorations.pathmorphing}
\usetikzlibrary{decorations.markings}
\usetikzlibrary{patterns}
\usetikzlibrary{automata}
\usetikzlibrary{positioning}
\usepackage{tikz-cd}

\tikzset{
	master/.style={
		execute at end picture={
			\coordinate (lower right) at (current bounding box.south east);
			\coordinate (upper left) at (current bounding box.north west);
		}
	},
	slave/.style={
		execute at end picture={
			\pgfresetboundingbox
			\path (upper left) rectangle (lower right);
		}
	}
}

\usepackage{pgfplots}

\numberwithin{equation}{section}

\def\ds{\displaystyle}
\def\bigO{{\cal O}}

\mathtoolsset{showonlyrefs}

\begin{document}
\title{Large gap asymptotics for the generating function \\ of the sine point process}
\author{Christophe Charlier\footnote{Department of Mathematics, KTH Royal Institute of Technology, Lindstedtsv\"{a}gen 25, SE-114 28 Stockholm, Sweden. e-mail: cchar@kth.se}}

\maketitle

\begin{abstract}
We consider the generating function of the sine point process on $m$ consecutive intervals. It can be written as a Fredholm determinant with discontinuities, or equivalently as the convergent series
\begin{equation*}
\sum_{k_{1},...,k_{m} \geq 0} \mathbb{P}\Bigg(\bigcap_{j=1}^{m} \#\{\mbox{points in the j-th interval}\}=k_{j}\Bigg)\prod_{j=1}^{m} s_{j}^{k_{j}},
\end{equation*}
where $s_{1},\ldots,s_{m} \in [0,+\infty)$. In particular, we can deduce from it joint probabilities of the counting function of the process. In this work, we obtain large gap asymptotics for the generating function, which are asymptotics as the size of the intervals grows. Our results are valid for an arbitrary integer $m$, in the cases where all the parameters $s_{1},\ldots,s_{m}$, except possibly one, are positive. This generalizes two known results: 1) a result of Basor and Widom, which corresponds to $m=1$ and $s_{1}>0$, and 2) the case $m=1$ and $s_{1} = 0$ for which many authors have contributed. We also present some applications in the context of thinning and conditioning of the sine process.
\end{abstract}


\section{Introduction}
Let 
\begin{equation*}
m \in \mathbb{N}\setminus \{0\}, \qquad \vec{s} = (s_{1},\ldots,s_{m}) \in [0,+\infty)^{m} \quad \mbox{and} \quad \vec{x} = (x_{0},x_{1},\ldots,x_{m}) \in \mathbb{R}^{m+1}
\end{equation*}
with $\vec{x}$ such that $-\infty < x_{0} < x_{1} < \ldots < x_{m} < +\infty$, and consider the Fredholm determinant
\begin{align}\label{F Fredholm}
F(\vec{x},\vec{s}) & = \det \bigg( 1- \sum_{k=1}^{m}(1-s_{k})\mathcal{K}|_{(x_{k-1},x_{k})} \bigg)  
\end{align}
where, for a given bounded Borel set $A \subset \mathbb{R}$, $\mathcal{K}|_{A}$ is the (trace class) integral operator acting on $L^{2}(A)$ whose kernel is given by
\begin{equation}\label{Sine kernel}
K(x,y) = \frac{\sin(x-y)}{\pi(x-y)}.
\end{equation}
In this paper, we obtain asymptotics for $F(r\vec{x},\vec{s})$ as $r \to + \infty$, up to and including the term of order 1, in the cases where all the parameters $s_{1},\ldots,s_{m}$, except possibly one, are positive. $F(\vec{x},\vec{s})$ is the generating function of the well-known sine point process of random matrices and has attracted considerable attention over the years. We discuss some background on the sine process and give more motivation for the study of $F$ in Subsection \ref{subsection: background and motivation} below. In particular, we show that if $s_{1},\ldots,s_{m}\in [0,1]$, then the asymptotics of $F(r\vec{x},\vec{s})$ as $r \to + \infty$ can be interpreted as large gap asymptotics.

\newpage

\medskip Before stating our main theorems, we first briefly review the known asymptotic results available in the literature.

\medskip Large $r$ asymptotics for $F(r\vec{x},\vec{s})$ when $m=1$ are already completely understood. We need to distinguish two regimes: 1) the case $s_{1} = 0$ and 2) the case $s_{1} \in (0,+\infty)$. For the first case $s_{1} = 0$, the asymptotics are given by
\begin{equation}\label{known result m=1 s1 = 0}
F\big((rx_{0},rx_{1}),0\big) = \exp \left( - \frac{r^{2}(x_{1}-x_{0})^{2}}{8} - \frac{1}{4} \log \big( r(x_{1}-x_{0}) \big) + \frac{1}{3}\log 2 + 3 \zeta^{\prime}(-1) +\bigO(r^{-1}) \right)
\end{equation}
as $r \to + \infty$, where $\zeta$ is Riemann's zeta-function. This result was first conjectured by Dyson in \cite{Dyson}, then proved simultaneously and independently by Ehrhardt and Krasovsky in \cite{Ehr sine, Krasovsky}, and then by Deift et al. in \cite{DIKZ2007}. On the other hand, for the second case where $s_{1} = e^{u_{1}} \in (0,+\infty)$, we have
\begin{equation}\label{known result m=1 s1 neq 0}
F((rx_{0},rx_{1}),e^{u_{1}}) = \exp \bigg(  \frac{ru_{1}(x_{1}-x_{0})}{\pi} + \frac{u_{1}^{2}}{2\pi^{2}} \log\big( 2r(x_{1}-x_{0}) \big) + 2\log G\Big(1+\frac{u_{1}}{2\pi i}\Big)G\Big(1-\frac{u_{1}}{2\pi i}\Big) + \bigO(r^{-1}) \bigg),
\end{equation}
as $r \to + \infty$, where $G$ is Barnes' $G$-function (see e.g. \cite[eq 5.17.2]{NIST} for a definition). This result was first proved by Basor and Widom in \cite{BW1983}, and then independently by Budylin and Buslaev in \cite{BB1995}. Note that the leading term for $\log F$ is of order $r^{2}$ in \eqref{known result m=1 s1 = 0} while it is of order $r$ in \eqref{known result m=1 s1 neq 0}, and that if we naively take $u_{1} \to -\infty$ (or equivalently $s_{1} \to 0$) in \eqref{known result m=1 s1 neq 0}, we do not recover \eqref{known result m=1 s1 = 0}. This explains heuristically why these two cases cannot be treated both at once. In fact, a critical transition takes place as $r \to + \infty$ and simultaneously $s_{1} \to 0$. This transition is quite technical and is described in terms of elliptic $\theta$-function in a series of papers by Bothner, Deift, Its and Krasovsky \cite{BDIK2015, BDIK2017, BDIK2019}.

\medskip Less is known for $m \geq 2$. In \cite{Widom1995}, Widom has tackled the problem of finding large $r$ asymptotics for $F(r\vec{x},\vec{s})$ in the case where $m$ is odd and $\vec{s} = (0,1,0,1,\ldots,0,1,0)$. He obtained
\begin{align}\label{widom union}
\partial_{r}\log F(r\vec{x},(0,1,\ldots,1,0)) = c_{1} \, r + c_{2}(r)+o(1), \qquad \mbox{as } r \to + \infty,
\end{align}
where $c_{1}$ is independent of $r$ and is explicitly computable, and the function $c_{2}(r)$ is a bounded oscillatory function of $r$ that requires the solution of a Jacobi inversion problem. These asymptotics were subsequently refined in \cite{DeiftItsZhou}, where the oscillations are described in terms of elliptic $\theta$-function. Note that \eqref{widom union} is an asymptotic formula for the log derivative of $F$, which leads after integration to an asymptotic formula for $\log F(r\vec{x},(0,1,\ldots,1,0))$. However with this method, the constant of integration (the term of order $1$ in the large $r$ asymptotics) remains unknown. Using a different method, Fahs and Krasovsky in \cite{FahsKrasovsky, FahsKrasovsky2} have recently obtained this constant for the case $m=3$ and $\vec{s}=(0,1,0)$. 

\medskip Until now, no results were available in the literature on large $r$ asymptotics of $F(r\vec{x},\vec{s})$ when $m \geq 2$ and several $s_{j}$'s are in the open intervals $(0,1)\cup (1,+\infty)$.

\medskip The aim of this paper is to contribute to these developments on large $r$ asymptotics of $F(r\vec{x},\vec{s})$. We obtain our results for an arbitrary integer $m$, in the cases where all the parameters $s_{1},\ldots,s_{m}$, except possibly one, are positive. We distinguish two cases: in Theorem \ref{thm:s1 neq 0}, we obtain large $r$ asymptotics for $F(r\vec{x},\vec{s})$ with $s_{1},\ldots,s_{m} \in (0,+\infty)$, and in Theorem \ref{thm:sp=0}, we obtain asymptotics for $F(r\vec{x},\vec{s})$ with $s_{p} = 0$ and $s_{1},\ldots,s_{p-1},s_{p+1},\ldots,s_{m} \in (0,+ \infty)$ (for an arbitrary $p \in \{1,\ldots,m\}$). Theorem \ref{thm:s1 neq 0} generalizes the result \eqref{known result m=1 s1 neq 0}, while Theorem \ref{thm:sp=0} generalizes \eqref{known result m=1 s1 = 0}. We describe several applications of our results in Subsection \ref{subsection: background and motivation} below.

\newpage 

\subsection{Main results}


\begin{theorem}\label{thm:s1 neq 0}
Let 
\begin{align*}
m \in \mathbb{N}_{>0}, \quad \vec{s}=(s_1,\ldots,s_m) \in (0,+ \infty)^{m}, \quad \vec{x} = (x_{0},\ldots,x_{m}) \in \mathbb{R}^{m+1}
\end{align*}
be such that $x_{0} < x_1 < x_2 < \ldots < x_m$. As $r \to + \infty$, we have
\begin{multline}\label{thm product s1 neq 0}
F(r\vec{x},\vec{s}) = \exp \bigg\{  \sum_{j=1}^{m} \frac{u_{j}}{\pi}(x_{j}-x_{0})r + \sum_{j=1}^{m} \frac{u_{j}^{2}}{2\pi^{2}} \log \big( 2r(x_{j}-x_{0}) \big) \\ + \sum_{1 \leq j < k \leq m} \frac{u_{j}u_{k}}{2\pi^{2}}\log\left( \frac{2r(x_{j}-x_{0})(x_{k}-x_{0})}{x_{k}-x_{j}}\right) + \sum_{j=1}^{m} \log \bigg( G\Big(1+\frac{u_{j}}{2\pi i}\Big)G\Big(1-\frac{u_{j}}{2\pi i}\Big) \bigg) \\ + \log \bigg(G\Big(1+ \sum_{j=1}^{m}\frac{u_{j}}{2\pi i}\Big)G\Big(1-\sum_{j=1}^{m}\frac{u_{j}}{2\pi i}\Big)\bigg) + \bigO \Big( \frac{\log r}{r} \Big) \bigg\}
\end{multline}
where $G$ is Barnes' $G$-function, and 
\begin{align}
& u_{j} =  \log \frac{s_{j}}{s_{j+1}} \quad \mbox{ for } \quad j = 1,\ldots,m, \label{def beta thm s1 neq 0}
\end{align}
with $s_{m+1} := 1$. Furthermore, the error term in \eqref{thm product s1 neq 0} is uniform in $s_{1},\ldots,s_{m}$ in compact subsets of $(0,+\infty)$ (or equivalently uniform in $u_{1},\ldots,u_{m}$ in compact subsets of $\mathbb{R}$) and uniform in $x_{0},\ldots,x_{m}$ in compact subsets of $\mathbb{R}$, as long as there exists $\delta > 0$ independent of $r$ such that
\begin{equation}\label{condition on xj in terms of delta}
\min_{0 \leq j < k \leq m} x_{k}-x_{j} \geq \delta.
\end{equation}
Alternatively, one can rewrite \eqref{thm product s1 neq 0} as follows:
\begin{multline}\label{F asymptotics thm s1 neq 0}
F(r\vec{x},\vec{s}) = \exp \bigg\{  \sum_{j=1}^{m}  u_{j} \mu_{j}(r) + \sum_{j=1}^{m} \frac{u_{j}^{2}}{2} \sigma_{j}^{2}(r) +  \sum_{1 \leq j < k \leq m} u_{j} u_{k} \Sigma_{j,k}(r)  \\  + \log \bigg(G\Big(1+\sum_{j=1}^{m}\frac{u_{j}}{2\pi i}\Big)G\Big(1-\sum_{j=1}^{m}\frac{u_{j}}{2\pi i}\Big)\bigg) + \sum_{j=1}^{m} \log \bigg( G\Big(1+\frac{u_{j}}{2\pi i}\Big)G\Big(1-\frac{u_{j}}{2\pi i}\Big) \bigg) + \bigO \Big( \frac{\log r}{r} \Big) \bigg\},
\end{multline}
where $\mu_{j}$, $\sigma_{j}^{2}$ and $\Sigma_{j,k}$ are given by
\begin{align}
& \mu_{j}(r) = \frac{r(x_{j}-x_{0})}{\pi}, \label{mean thm s1 neq 0} \\
& \sigma_{j}^{2}(r) = \frac{\log(2r(x_{j}-x_{0}))}{\pi^{2}}, \label{variance thm s1 neq 0} \\
& \Sigma_{j,k}(r) = \frac{1}{2\pi^{2}} \log \left( \frac{2r(x_{j}-x_{0})(x_{k}-x_{0})}{|x_{k}-x_{j}|} \right). \label{cov thm s1 neq 0}
\end{align}
\end{theorem}
\newpage
\begin{theorem}\label{thm:sp=0}
Let $m \in \mathbb{N}_{>0}$, $p \in \{1,\ldots,m\}$ and
\begin{align*}
s_{p} = 0, \quad (s_1,\ldots,s_{p-1},s_{p+1},\ldots,s_{m}) \in (0,+\infty)^{m-1}, \quad \vec{x} = (x_{0},\ldots,x_{m}) \in \mathbb{R}^{m+1}
\end{align*}
be such that $x_{0} < x_1 < x_2 < \ldots < x_m$, and define $\vec{s}=(s_1,\ldots,s_{m})$. As $r \to + \infty$, we have
\begin{align}
F(r \vec{x},\vec{s}) = & \exp \Bigg\{ - \frac{r^{2}(x_{p}-x_{p-1})^{2}}{8} \nonumber \\
& - \Bigg( \sum_{j=0}^{p-2} \frac{u_{j}}{\pi}\sqrt{x_{p}-x_{j}}\sqrt{x_{p-1}-x_{j}} - \sum_{j=p+1}^{m} \frac{u_{j}}{\pi} \sqrt{x_{j}-x_{p}}\sqrt{x_{j}-x_{p-1}} \Bigg)r \nonumber \\
& + \sum_{\substack{j=0 \\ j \neq p-1,p}}^{m} \frac{u_{j}^{2}}{4\pi^{2}} \log  \left(  \frac{4 \sqrt{|x_{j}-x_{p}| \, |x_{j}-x_{p-1}|}|2x_{j}-x_{p}-x_{p-1}|r}{x_{p}-x_{p-1}} \right) -\frac{1}{4}\log\big( r(x_{p}-x_{p-1}) \big)  \nonumber \\
& + \sum_{\substack{0 \leq j < k \leq m \\ j,k \neq p-1,p}} \frac{u_{j} u_{k}}{2\pi^{2}} \log \left( \frac{\sqrt{|x_{k}-x_{p}|}\sqrt{|x_{j}-x_{p-1}|}+\sqrt{|x_{k}-x_{p-1}|}\sqrt{|x_{j}-x_{p}|}}{\big|\sqrt{|x_{k}-x_{p}|}\sqrt{|x_{j}-x_{p-1}|}-\sqrt{|x_{k}-x_{p-1}|}\sqrt{|x_{j}-x_{p}|}\big|} \right) \nonumber\\
& + \frac{1}{3}\log 2 + 3 \zeta^{\prime}(-1) + \sum_{\substack{j=0 \\ j \neq p-1,p}}^{m} \log \bigg( G\Big(1+\frac{u_{j}}{2\pi i}\Big)G\Big(1-\frac{u_{j}}{2\pi i}\Big) \bigg) + \bigO\Big( \frac{\log r}{r} \Big) \Bigg\},\label{explicit asymp for F in the case where sp=0 in thm}
\end{align}
where $G$ is Barnes' $G$-function, $\zeta$ is Riemann's zeta-function, and $u_{0},\ldots,u_{p-2},u_{p+1},\ldots,u_{m}$ are given by 
\begin{align}
& u_{j} = \log \frac{s_{j}}{s_{j+1}}, \qquad j \in \{0,\ldots,m \}\setminus \{p-1,p\} \label{def of beta in thm sp=0},
\end{align}
where $s_{0} := 1$, $s_{m+1} := 1$. Furthermore, the error term in \eqref{explicit asymp for F in the case where sp=0 in thm} is uniform in $s_{1},\ldots,s_{p-1},s_{p+1},\ldots$, $s_{m}$ in compact subsets of $(0,+\infty)$ (or equivalently uniform in $u_{0},\ldots,u_{p-2},u_{p+1},\ldots,u_{m}$ in compact subsets of $\mathbb{R}$) and uniform in $x_{0},\ldots,x_{m}$ in compact subsets of $\mathbb{R}$, as long as there exists $\delta > 0$ independent of $r$ such that \eqref{condition on xj in terms of delta} holds.

\medskip Alternatively, one can rewrite \eqref{explicit asymp for F in the case where sp=0 in thm} as follows:
\begin{align}
F(r \vec{x},\vec{s}) = & F\big( (rx_{p-1},rx_{p}),0 \big) \exp \Bigg\{ - \Bigg( \sum_{j=0}^{p-2}u_{j}\hat{\mu}_{j}(r) - \sum_{j=p+1}^{m}u_{j}\hat{\mu}_{j}(r) \Bigg) + \sum_{\substack{j=0 \\ j \neq p-1,p}}^{m} \frac{u_{j}^{2}}{2} \hat{\sigma}_{j}^{2}(r) \nonumber \\
&  + \sum_{\substack{0 \leq j < k \leq m \\ j,k \neq p-1,p}}u_{j}u_{k} \hat{\Sigma}_{j,k} + \sum_{\substack{j=0 \\ j \neq p-1,p}}^{m} \log \bigg( G\Big(1+\frac{u_{j}}{2\pi i}\Big)G\Big(1-\frac{u_{j}}{2\pi i}\Big) \bigg) + \bigO\Big( \frac{\log r}{r}\Big) \Bigg\}
\end{align}
where the large $r$ asymptotics of $F\big( (rx_{p-1},rx_{p}),0 \big)$ are given by \eqref{known result m=1 s1 = 0}, and $\hat{\mu}_{j}$, $\hat{\sigma}_{j}^{2}$ and $\hat{\Sigma}_{j,k}$ are given by
\begin{align}
& \hat{\mu}_{j}(r) = \frac{r}{\pi}\sqrt{|\smash{x_{p}-x_{j}}|}\sqrt{|\smash{x_{p-1}-x_{j}}|}, \label{mean sp=0} \\
& \hat{\sigma}_{j}^{2}(r) = \frac{1}{2\pi^{2}}\log \Bigg(  \frac{4 \sqrt{|x_{j}-x_{p}| \, |x_{j}-x_{p-1}|}|2x_{j}-x_{p}-x_{p-1}|r}{x_{p}-x_{p-1}} \Bigg), \label{var sp=0} \\
& \hat{\Sigma}_{j,k} = \frac{1}{2\pi^{2}}\log \Bigg( \frac{\sqrt{|x_{k}-x_{p}|}\sqrt{|x_{j}-x_{p-1}|}+\sqrt{|x_{k}-x_{p-1}|}\sqrt{|x_{j}-x_{p}|}}{\big|\sqrt{|x_{k}-x_{p}|}\sqrt{|x_{j}-x_{p-1}|}-\sqrt{|x_{k}-x_{p-1}|}\sqrt{|x_{j}-x_{p}|}\big|} \Bigg). \label{cov sp=0}
\end{align}
\end{theorem}
\newpage 
\paragraph{Numerical confirmations of Theorems \ref{thm:s1 neq 0} and \ref{thm:sp=0}.} 
\begin{figure}[h]
\begin{tikzpicture}[master]
\node at (0,0) {};
\node at (0,-0.1) {\includegraphics[scale=0.25]{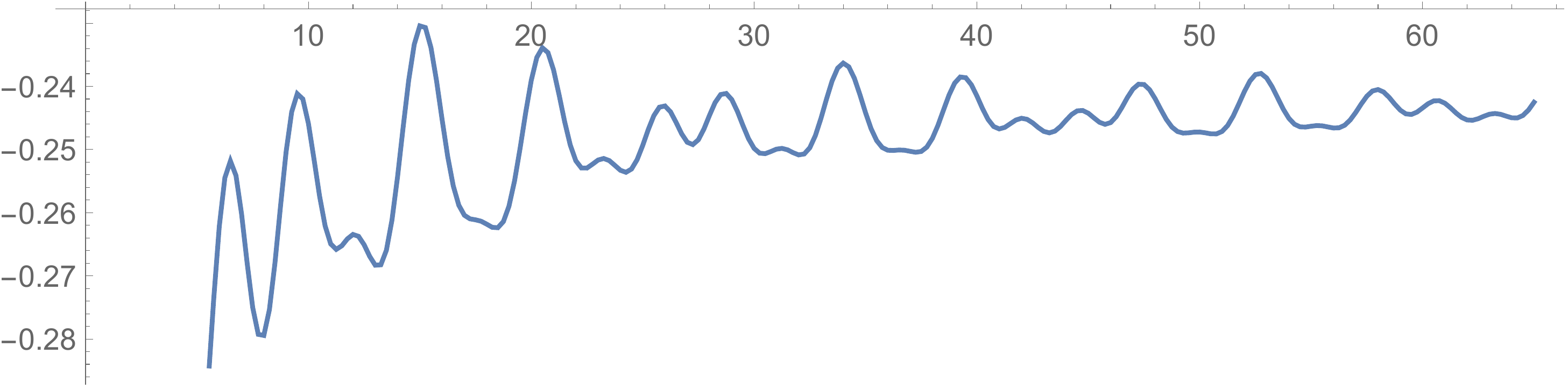}};
\end{tikzpicture}
\begin{tikzpicture}[slave]
\node at (0,0) {\includegraphics[scale=0.25]{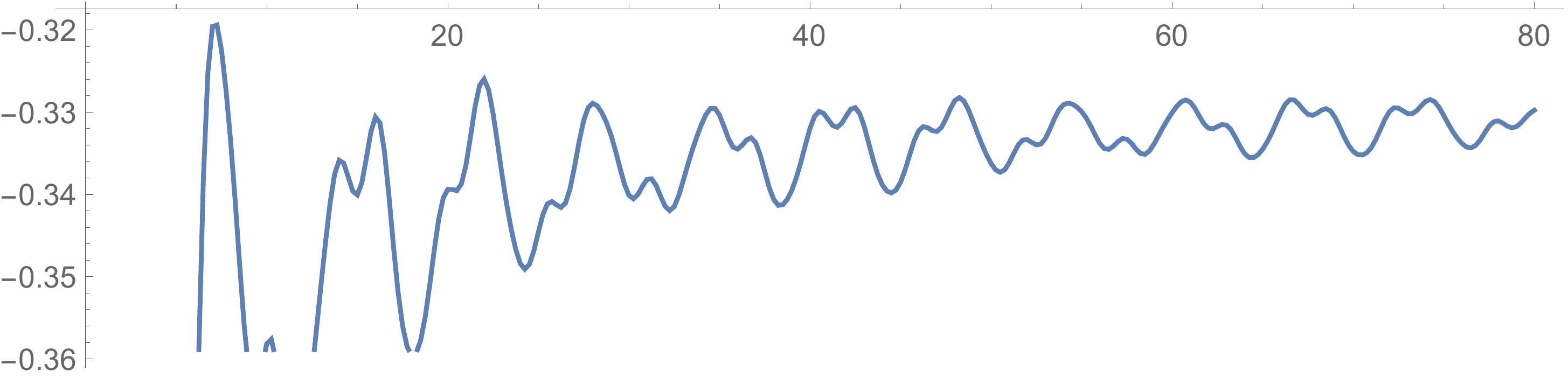}};
\end{tikzpicture}
\caption[]{\label{fig: num check}  Numerical confirmations of Theorem \ref{thm:s1 neq 0}.
}
\end{figure}
Recent progress of Bornemann \cite{Bornemann} on the numerical evaluation of Fredholm determinants have allowed us to verify Theorems \ref{thm:s1 neq 0} and \ref{thm:sp=0} for several choices of the parameters.
Let $\mathcal{F}_{1}(r\vec{x},\vec{s})$ denote the right-hand side of \eqref{thm product s1 neq 0} without the error term. Figure \ref{fig: num check} represents the graph of the function
\begin{align}\label{function in num check}
r \mapsto r\Big(\log F(r\vec{x},\vec{s}) - \log \mathcal{F}_{1}(r\vec{x},\vec{s}) \Big)
\end{align} 
for the following two choices of the parameters:
\begin{align*}
& \mbox{Left: } & & m=2, & & x_{0}=0, \; x_{1} = 0.7, \; x_{2} = 1.2, & & u_{1} = -1.1, \; u_{2} = -2.4, \\
& \mbox{Right: } & & m=3, & & x_{0} = 0, \; x_{1} = 0.5, \; x_{2} = 1.1, \; x_{3} = 1.7, & & u_{1} = -0.8, \; u_{2} = -1.8, \; u_{3} = -1.32.
\end{align*}
Similarly, let $\mathcal{F}_{2}(r\vec{x},\vec{s})$ denote the right-hand side of \eqref{explicit asymp for F in the case where sp=0 in thm} without the error term. Figure \ref{fig: num check 2} represents the graph of the function 
\begin{align}\label{function in num check 2}
r \mapsto r\Big(\log F(r\vec{x},\vec{s}) - \log \mathcal{F}_{2}(r\vec{x},\vec{s}) \Big)
\end{align} 
for the following two cases:
\begin{align*}
& \mbox{Left: } & & m=3, \; p=2, & & x_{0}=0, \; x_{1} = 0.5, \; x_{2} = 1.1, \; x_{3} = 1.7, \\
& && && u_{0} = 0.8, \; u_{3} = -1.32, \\
& \mbox{Right: } & & m=4, \; p=3, & & x_{0} = 0, \; x_{1} = 0.5, \; x_{2} = 1.1, \; x_{3} = 1.7, \; x_{4} = 2.5, \\
& && && u_{0} = 0.8, \; u_{1} = 1.8, \; u_{4} = -1.87.
\end{align*}
We see in Figures \ref{fig: num check} and \ref{fig: num check 2} that the functions \eqref{function in num check} and \eqref{function in num check 2} seem to remain bounded as $r \to + \infty$. These observations are consistent with Theorems \ref{thm:s1 neq 0} and \ref{thm:sp=0}. In fact, Figures \ref{fig: num check} and \ref{fig: num check 2} also suggest that the error terms in Theorems \ref{thm:s1 neq 0} and \ref{thm:sp=0} could be reduced from $\bigO(\frac{\log r}{r})$ to $\bigO(\frac{1}{r})$. 
\begin{figure}[h]
\begin{tikzpicture}[master]
\node at (0,0) {};
\node at (0,-0.1) {\includegraphics[scale=0.25]{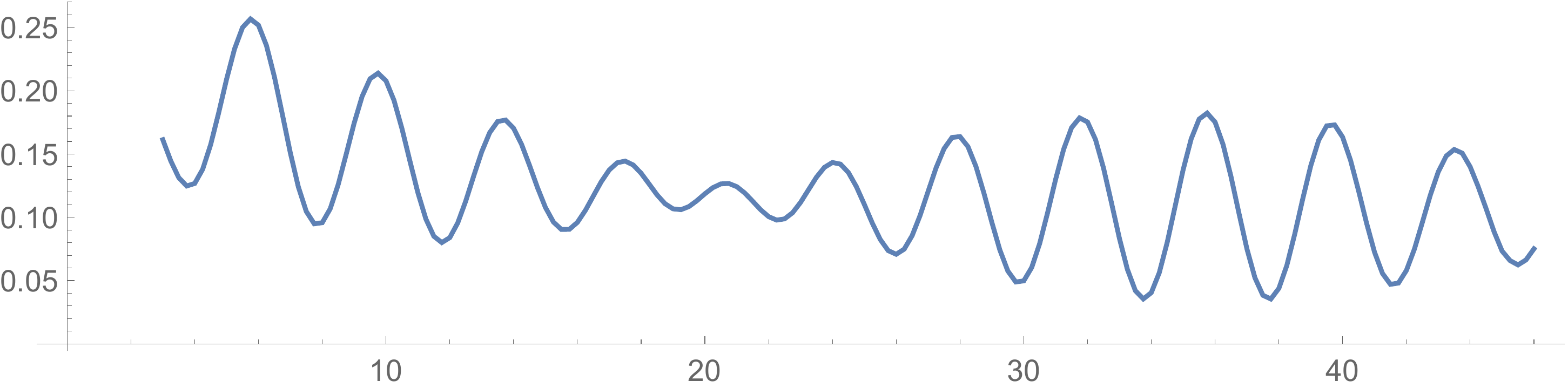}};
\end{tikzpicture}
\begin{tikzpicture}[slave]
\node at (0,0) {\includegraphics[scale=0.25]{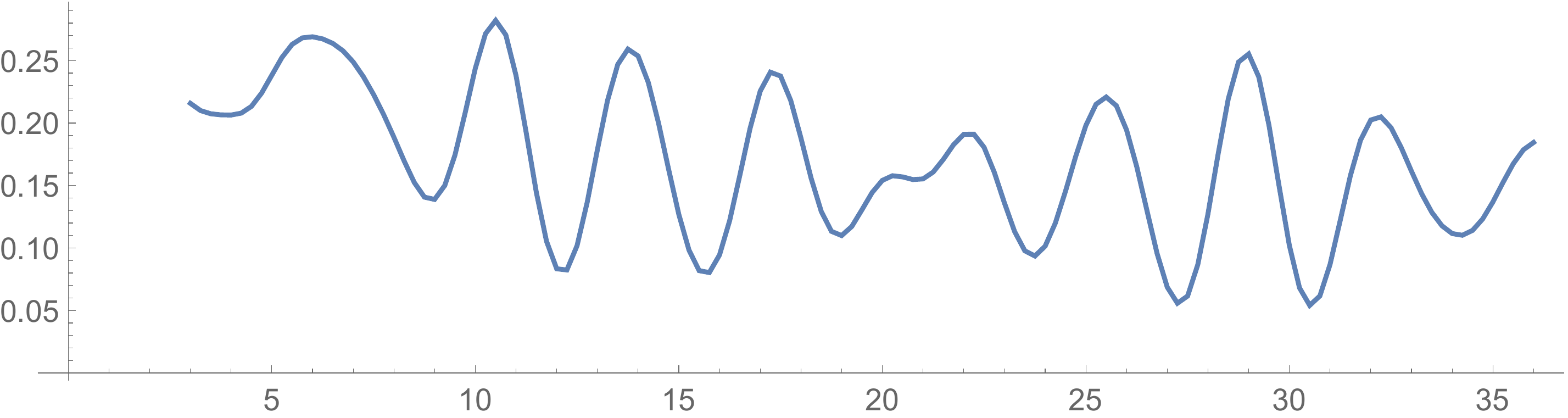}};
\end{tikzpicture}
\caption[]{\label{fig: num check 2} Numerical confirmations of Theorem \ref{thm:sp=0}.
}
\end{figure}

\newpage 

\subsection{Background and applications of Theorems \ref{thm:s1 neq 0} and \ref{thm:sp=0}}\label{subsection: background and motivation}
The sine point process lies at the heart of random matrix theory. It has attracted a lot of attention since the seminal work of Dyson \cite{DysonConjecture}, who first proved that this process describes the local eigenvalue statistics in the bulk of the spectrum of large random Hermitian matrices taken from the Gaussian Unitary Ensemble. Dyson also conjectured that this same process also describes the bulk local eigenvalue statistics for a wide class of large random matrices. There has been much progress on this conjecture, which has now been rigorously proved for many random matrix models, see e.g. \cite{DKMVZ2, PS1997, BleherIts, JohanssonUniversality, PS2008, EPRSY2010, TaoVu}. We refer to \cite{Erdos, TaoVu2, KuijlaarsUniversality, Lub2016} for recent surveys of known appearances of the sine process in random matrix theory.

\medskip The Fredholm determinant $F(\vec{x},\vec{s})$ is a central object in the study of the sine process. For the convenience of the reader, we first briefly recall the definition of a point process, following the classical references \cite{Soshnikov, Borodin, Johansson}.

\medskip A point process on $\mathbb{R}$ is a probability measure over the space $\{X\}$ of all locally finite point configurations on $\mathbb{R}$. In general, the process can be well understood via the study of its $k$-point correlation functions $\{\rho_{k}:\mathbb{R}^{k} \to [0,+\infty)\}_{k \geq 1}$ which are defined such that
\begin{align}\label{def of DPP}
\mathbb{E}\bigg[ \sum_{\substack{\xi_{1},\ldots,\xi_{k} \in X \\ \xi_{i}\neq \xi_{j} \, \mathrm{if} \, i \neq j}} f(\xi_{1},\ldots,\xi_{k}) \bigg] = \int_{\mathbb{R}^{k}} f(u_{1},\ldots,u_{k}) \rho_{k}(u_{1},\ldots,u_{k})du_{1}\ldots du_{k}
\end{align}
holds for any measurable symmetric function $f: \mathbb{R}^{k} \to \mathbb{R}$ with compact support. The sum at the left-hand side of \eqref{def of DPP} is taken over all (ordered) $k$-tuples of distinct points of the random point configuration $X$. 

\medskip A point process on $\mathbb{R}$ is \textit{determinantal} if all its correlation functions $\{\rho_{k}\}_{k \geq 1}$ exist and can be expressed as determinants involving a kernel $\mathbb{K}:\mathbb{R}^{2}\to \mathbb{R}$ as follows
\begin{align}\label{def of correlation functions}
\rho_{k}(u_{1},\ldots,u_{k}) = \det (\mathbb{K}(u_{i},u_{j}))_{i,j=1}^{k}, \qquad \mbox{for all } k \geq 1 \mbox{ and for all } u_{1},\ldots,u_{k} \in \mathbb{R}.
\end{align}
The sine process is determinantal and corresponds to the case $\mathbb{K}=K$, where the sine kernel $K:\mathbb{R}^{2} \to \mathbb{R}$ is defined in \eqref{Sine kernel}.
In a determinantal point process, all quantities of interest  can be expressed in terms of the kernel. For example, it is directly seen from \eqref{def of DPP} and \eqref{def of correlation functions} that the probability that a random point configuration $X$ (distributed according to the sine point process) contains no points on a given bounded Borel set $A \subset \mathbb{R}$ is equal to
\begin{align}\label{gap proba generalities}
\hspace{-0.5cm}\mathbb{P}[X \cap A = \emptyset] = \mathbb{E}\bigg[\prod_{\xi \in X}(1-\chi_{A}(\xi)) \bigg] = 1 + \sum_{k = 1}^{+\infty} \frac{(-1)^{k}}{k !} \int_{A^{k}} \det (K(u_{i},u_{j}))_{i,j=1}^{k}du_{1} \ldots du_{k},
\end{align}
where $\chi_{A}(\xi) = 1$ if $\xi \in A$ and $\chi_{A}(\xi) = 0$ otherwise. Note that the right-hand side of \eqref{gap proba generalities} is, by definition, equal to the Fredholm determinant $\det (1-\mathcal{K}|_{A})$. In the sine process, the expected number of points that fall in $A$ can be computed explicitly using \eqref{Sine kernel}:
\begin{align}\label{explicit formula for the expectation}
\mathbb{E}[\#(X \cap A)] = \mathbb{E}\bigg[ \sum_{\xi \in X} \chi_{A}(\xi) \bigg] = \int_{A} K(x,x)dx = \frac{|A|}{\pi},
\end{align}
where $|A|$ is the Lebesgue measure of $A$. We refer the reader to \cite{Soshnikov, Borodin, Johansson} for more discussions on the algebraic and probabilitic properties of determinantal point processes. 

\medskip Taking $A = (x_{0},x_{1})$ in \eqref{gap proba generalities}, we see that the probability to observe a gap in the sine process on $(x_{0},x_{1})$ can be expressed in terms of $F$ (defined in \eqref{F Fredholm}) by
\begin{align}\label{prob one gap sine}
F((x_{0},x_{1}),0) = \det (1-\mathcal{K}|_{(x_{0},x_{1})}) = \mathbb{P}[X \cap (x_{0},x_{1}) = \emptyset].
\end{align}
The large gap asymptotics on a single interval in the sine process are given by \eqref{known result m=1 s1 = 0}.

\medskip More generally, taking $m$ odd and $A = (x_{0},x_{1})\cup (x_{2},x_{3})\cup \ldots \cup (x_{m-1},x_{m})$, we infer from \eqref{gap proba generalities} and \eqref{F Fredholm} that the probability to find no points on $\frac{m+1}{2}$ disjoint intervals is given by
\begin{equation}\label{gap prob multi interval}
F(\vec{x},(0,1,0,1,0,\ldots,1,0)) = \mathbb{P}\Big[X \cap \Big( (x_{0},x_{1})\cup (x_{2},x_{3})\cup \ldots \cup (x_{m-1},x_{m}) \Big) = \emptyset\Big],
\end{equation}
where $s_{j} = 0$ if $j$ is odd and $s_{j} = 1$ otherwise. The known results on the asymptotics of \eqref{gap prob multi interval} as the size of the intervals gets large have been discussed below \eqref{widom union}.  

\medskip We now discuss the meaning of  $F(\vec{x},\vec{s})$ in terms of probabilities for general values of $s_{1},\ldots,s_{m}$. As before, let $X$ be a random point configuration distributed according to the sine process. Given a Borel set $A$, we define $N_{A} = \#(X \cap A)$. In other words, $N_{A}$ is the random variable that counts the number of points in $X$ that falls in $A$; $N_{A}$ is also called the counting function on $A$. It is known \cite[Theorem 2]{Soshnikov} that $F(\vec{x},\vec{s})$ is an entire function of $s_{1},\ldots,s_{m}$ which can be rewritten as follows
\begin{equation}\label{F generating function}
F(\vec{x},\vec{s}) = \mathbb{E}\bigg[ \prod_{j=1}^{m} s_{j}^{N_{(x_{j-1},x_{j})}} \bigg] = \sum_{k_{1},...,k_{m} \geq 0} \mathbb{P}\Bigg(\bigcap_{j=1}^{m}N_{(x_{j-1},x_{j})}=k_{j}\Bigg)\prod_{j=1}^{m} s_{j}^{k_{j}}.
\end{equation}
The above expression motivates why $F$ is called the generating function of the sine point process; it shows in particular that we can deduce a lot of information from $F$. Any quantity of the form $0^{0}$ in \eqref{F generating function} should be interpreted as being equal to $1$. More precisely,
\begin{align*}
s_{j}^{N_{(x_{j-1},x_{j})}}=1 \mbox{ if } s_{j}=0 \mbox{ and } N_{(x_{j-1},x_{j})}=0, \qquad \mbox{ and } \qquad s_{j}^{k_{j}}=1 \mbox{ if } s_{j}=0 \mbox{ and } k_{j}=0.
\end{align*}
For example, for an odd integer $m$ and for $\vec{s}$ such that $s_{j} = 0$ if $j$ is odd and $s_{j} = 1$ if $j$ is even, we get from \eqref{F generating function} that
\begin{equation}\label{gap prob multi interval 2}
F(\vec{x},(0,1,0,\ldots,1,0)) = \mathbb{P}\Big( N_{(x_{0},x_{1})} = 0 \cap N_{(x_{2},x_{3})} = 0 \cap \ldots \cap N_{(x_{m-1},x_{m})} = 0 \Big),
\end{equation}
which is equivalent to \eqref{gap prob multi interval}, as it must.

\medskip We mention that there is a well-known connection between $F$ and the theory of Painlev\'{e} equations. Using monodromy preserving deformations, Jimbo, Miwa, M\^{o}ri and Sato in \cite[eq (2.27)]{JMMS1980} have established the following remarkable identity 
\begin{align*}
F((x_{0},x_{1}),s_{1}) = \exp \bigg( \int_{0}^{x_{1}-x_{0}} \frac{\sigma(x)}{x}dx \bigg)
\end{align*}
where $s_{1} \in [0,+\infty)$ and $\sigma$ is the solution to the Painlev\'{e} V equation
\begin{align*}
& && \hspace{-2cm} (x\sigma^{\prime \prime})^{2} + 4(x \sigma^{\prime}-\sigma)(x\sigma^{\prime} - \sigma + (\sigma^{\prime})^{2}) = 0
	\\ \nonumber
& \mbox{which satisfies} & & \hspace{-2cm} \sigma(x) = -\frac{s_{1}}{\pi}x - \frac{s_{1}^{2}}{\pi^{2}}x^{2} - \frac{s_{1}^{3}}{\pi^{3}}x^{3} + \bigO(x^{4}), \qquad \mbox{as } x \to 0.
\end{align*}
For general values of the parameters $m \geq 1$, $\vec{s}$ and $\vec{x}$, the determinant $F(\vec{x},\vec{s})$ is related to a more involved system of partial differential equations which generalizes the Painlev\'{e} V equation \cite{JMMS1980} (see also \cite[Theorem 3.6.1 and Subsection 3.6.3]{AGZ2010}). The solution of this system of equations involves transcendental functions.

\vspace{-0.3cm}\paragraph{Thinning.} The operation of thinning consists of randomly removing a fraction of points and was introduced in random matrix theory by Bohigas and Pato in \cite{BohigasPato1, BohigasPato2}. Given a function 
\begin{align*}
s: \mathbb{R} \to [0,1], \qquad x \mapsto s(x),
\end{align*}
we say that a point configuration $\widetilde{X}$ is distributed according to \textit{the thinned sine point process} if 
\begin{align}\label{def of X tilde}
\widetilde{X} = \{x \in X: \mathcal{B}(x)=1 \}
\end{align}
where $X$ is distributed according to the sine process, $\mathcal{B} = \{\mathcal{B}(x):x \in \mathbb{R}\}$ is a random field of independent Bernoulli variables with $\mathbb{P}[\mathcal{B}(x)=1]=1-s(x)$, and furthermore $\mathcal{B}$ is independent of $X$. It is well-known \cite[Proposition A.2]{LMR}, and also easy to see from \eqref{def of DPP}, that the thinned sine process is also determinantal and that its kernel is given by
\begin{align}\label{kernel K tilde}
\widetilde{K}(x,y) = \sqrt{1-s(x)} K(x,y) \sqrt{1-s(y)} = \sqrt{1-s(x)} \frac{\sin(x-y)}{\pi(x-y)} \sqrt{1-s(y)}.
\end{align}
If the function $x \mapsto s(x)=s_{1}$ is constant, then each point is removed with the same probability $s_{1} \in [0,1]$.\footnote{If $s_{1} = 0$, no particle are removed and the thinned process coincides with the initial point process.} The thinned sine point process already presents interesting features in this case (see e.g. \cite{BDIK2015,BIP}), as it describes a crossover between the original process (when $s_{1} = 0$), and an uncorrelated Poisson process (when $s_{1} \to 1$ at a certain speed). It follows directly from \eqref{gap proba generalities} (with $X$ and $K$ replaced by $\widetilde{X}$ and $\widetilde{K}$) and the definition \eqref{F Fredholm} of $F$ that 
\begin{align*}
F((x_{0},x_{1}),s_{1}) = \det (1-(1-s_{1})\mathcal{K}|_{(x_{0},x_{1})}) = \mathbb{P}[\widetilde{X} \cap (x_{0},x_{1}) = \emptyset].
\end{align*}
In particular, \eqref{known result m=1 s1 neq 0} can be interpreted as large gap asymptotics in the (constant) thinned sine process. By comparing the leading terms of \eqref{known result m=1 s1 neq 0} and \eqref{known result m=1 s1 = 0}, we see that it is significantly more likely to observe a large gap in the (constant) thinned sine point process than in the usual sine point process. 

\medskip More generally, let $s$ be a piecewise constant function, say 
\begin{align}\label{piecewise constant s}
s(x) = \begin{cases}
s_{j}, & \mbox{if } x \in (x_{j-1},x_{j}), \; j=1,\ldots,m, \\
0, & \mbox{otherwise},
\end{cases}
\end{align}
where $s_{1},\ldots,s_{m} \in [0,1]$. It follows directly from \eqref{F Fredholm}, \eqref{gap proba generalities} and \eqref{kernel K tilde} that the probability of observing no points on $(x_{0},x_{m})$ in the (piecewise constant) thinned point process is simply given by
\begin{align}\label{generalized prob piecewise thinned}
F(\vec{x},\vec{s}) = \det \bigg(1- \sum_{k=1}^{m}(1-s_{k})\mathcal{K}|_{(x_{k-1},x_{k})} \bigg) = \mathbb{P}[\widetilde{X} \cap (x_{0},x_{m}) = \emptyset].
\end{align}
Therefore, Theorems \ref{thm:s1 neq 0} and \ref{thm:sp=0} give large gap asymptotics in any piecewise thinned sine process, as long as at most one of the parameters $s_{1},\ldots,s_{m}$ is $0$.

\vspace{-0.3cm}\paragraph{Conditioning.}

 Now, following \cite{ChCl2}, we consider a situation where we have information about the thinned process, and we try to deduce from it some information about the initial process. More precisely, let $X$ be a random point configuration distributed according to the sine point process, let $\widetilde{X}$ be as in \eqref{def of X tilde}, and assume that $\#(\widetilde{X} \cap B)=0$, where $B = (x_{0},x_{m})$. We are interested in the conditional random variable
\begin{align}\label{conditional random variable}
\widehat{\mathcal{N}}_{B} := \#(X \cap B) | \Big( \#(\widetilde{X} \cap B)=0 \Big).
\end{align}
If $s$ is piecewise constant and given by \eqref{piecewise constant s}, then using first Bayes' formula and then \eqref{prob one gap sine} and \eqref{generalized prob piecewise thinned}, we obtain
\begin{align} 
\mathbb{P}\big( \widehat{\mathcal{N}}_{B} = 0 \big) & = \mathbb{P}\Big( \#(X \cap B)=0 \Big| \#(\widetilde{X} \cap B)=0  \Big) = \frac{\mathbb{P}(\#(X \cap B)=0)}{\mathbb{P}( \#(\widetilde{X} \cap B)=0 )} = \frac{F((x_{0},x_{m}),0)}{F(\vec{x},\vec{s})}. \nonumber
\end{align}
Therefore, if at most one of the parameters $s_{1},\ldots,s_{m}$ is $0$, we can obtain large $r$ asymptotics for $\mathbb{P}\big( \widehat{\mathcal{N}}_{(rx_{0},rx_{m})} = 0 \big)$ by combining \eqref{known result m=1 s1 = 0} with either Theorem \ref{thm:s1 neq 0} or Theorem \ref{thm:sp=0}. The conditional random variable \eqref{conditional random variable} is relevant in e.g. nuclear physics \cite{BohigasPato1,BohigasPato2}. Indeed, it is now well-known (from the work of Dyson) that the energy levels of heavy atoms feature a similar repulsive structure as the points of the sine point process. However, high quality data is often not available, and in practice one usually observes only a fraction of the energy levels. It is then natural to wonder if one can retrieve some missing energy levels given the available information.

\medskip The random variable \eqref{conditional random variable} is conditioned on $\#(\widetilde{X} \cap B)=0$. We mention that different types of conditioning of the sine process have been studied in great depth in the literature. In particular, it is known \cite[Theorem 4.2]{Ghosh} that for almost all point configurations $X$,\footnote{Here, ``almost all $X$" means ``almost all $X$ with respect to the sine process".} if $A$ is a compact interval, then $X \setminus A$ determines almost surely $\#(X \cap A)$. The conditional measure of the sine process on $\{X|X\setminus A\}$ admits an explicit density, see \cite[Theorems 1.1 and 1.4]{BufetovConditional}. Furthermore, the correlation kernel of this conditional sine process converges to the usual sine kernel as the size of the interval $A$ gets large, see \cite[Theorems 1.3 and 1.4]{KuijlaarsDiaz}.

\paragraph{Asymptotics for the variance and covariance of the sine counting function.} \, \newline 
We first briefly review some known results on the counting function of the sine process.

\medskip The formula \eqref{explicit formula for the expectation} implies in particular that $\mathbb{E}[N_{(rx_{0},rx_{1})}] = \mu_{1}(r)$, where $\mu_{1}$ is given by \eqref{mean thm s1 neq 0}. There is no such explicit expression for $\mbox{Var}[N_{(rx_{0},rx_{1})}]$, but we can compute its large $r$ asymptotics as follows. We know from \eqref{F generating function} with $m=1$ that 
\begin{align}
F((rx_{0},rx_{1}),e^{u}) & = \mathbb{E} \big[ e^{uN_{(rx_{0},rx_{1})}} \big] \nonumber \\
& = 1 + u \, \mathbb{E}[N_{(rx_{0},rx_{1})}] + \frac{u^{2}}{2} \mathbb{E}[N_{(rx_{0},rx_{1})}^{2}] + \bigO(u^{3}) \quad \mbox{as } u \to 0. \label{F as expectation m=1}
\end{align}
Recall that the asymptotics of $F((rx_{0},rx_{1}),e^{u})$ as $r \to + \infty$ are given by \eqref{known result m=1 s1 neq 0} (and were obtained in \cite{BW1983}) and are uniform for $u$ in compact subsets of $\mathbb{R}$. In particular, these asymptotics can be expanded as $u \to 0$. A comparison of this expansion with \eqref{F as expectation m=1} yields
\begin{align}
& \mbox{Var}[N_{(rx_{0},rx_{1})}] = \sigma_{1}^{2}(r) + \frac{1 + \gamma_{\mathrm{E}}}{\pi^{2}} + \bigO (r^{-1}), \label{variance asymp s1 neq 0}
\end{align}
as $r \to + \infty$, where $\sigma_{1}^{2}$ is given by \eqref{variance thm s1 neq 0}. Here $\gamma_{\mathrm{E}} \approx 0.5772$ is Euler's gamma constant and is part of the definition of the Barnes' G function, see \cite[formula 5.17.3]{NIST}. The leading term of \eqref{variance asymp s1 neq 0} was also obtained in \cite{CostinLebowitz} without relying on \cite{BW1983}. In a slightly different direction, Holcomb and Paquette in \cite{HolcombPaquette} have studied the maximum deviation of the sine-$\beta$ process. For $\beta=2$, this process coincides with the determinantal sine point process, and their result states that for any $\epsilon >0$, we have
\begin{align*}
\lim_{r \to + \infty} \mathbb{P} \Bigg( \frac{\sqrt{2}}{\pi}-\epsilon \leq \frac{\max_{0 \leq r' \leq r} (N_{(r'x_{0},r'x_{1})}-\mu_{1}(r'))}{\log r} \leq \frac{\sqrt{2}}{\pi}+\epsilon \Bigg) = 1.
\end{align*}

Theorem \ref{thm:s1 neq 0} allows to obtain precise large $r$ asymptotics for the covariance between $N_{(rx_{0},rx_{1})}$ and $N_{(rx_{0},rx_{2})}$. To see this, we first rewrite the expression \eqref{F generating function} for $F$ as follows
\begin{equation}\label{F as expectation}
F(r\vec{x},\vec{s}) = \mathbb{E}\bigg[ \prod_{j=1}^{m} s_{j}^{N_{(rx_{j-1},rx_{j})}} \bigg] = \mathbb{E} \bigg[ \prod_{j=1}^{m}e^{u_{j}N_{(rx_{0},rx_{j})}} \bigg],
\end{equation}
where $u_{1},\ldots,u_{m}$ are given by \eqref{def beta thm s1 neq 0}. In particular, using \eqref{F as expectation} with $m=1$ and $m=2$, we obtain 
\begin{align}
\frac{F\big(r(x_{0},x_{1},x_{2}),(e^{2u},e^{u})\big)}{F(r(x_{0},x_{1}),e^{u})F(r(x_{0},x_{2}),e^{u})} = & \; \frac{\mathbb{E}\big[ e^{u N_{(rx_{0},rx_{1})}}e^{u N_{(rx_{0},rx_{2})}} \big]}{\mathbb{E}\big[ e^{u N_{(rx_{0},rx_{1})}} \big]\mathbb{E}\big[ e^{u N_{(rx_{0},rx_{2})}} \big]} \label{cov in the expansion} \\
= & \; 1 + \mbox{Cov}(N_{(rx_{0},rx_{1})},N_{(rx_{0},rx_{2})})u^{2} + \bigO(u^{3}), \quad \mbox{as } u \to 0. \nonumber
\end{align}
The large $r$ asymptotics for the left-hand side of \eqref{cov in the expansion} can be deduced from Theorem \ref{thm:s1 neq 0} and are uniform for $u$ in compact subsets of $\mathbb{R}$. By expanding these asymptotics as $u \to 0$, and then comparing with the right-hand side of \eqref{cov in the expansion}, we obtain
\begin{equation}\label{asymp for the cov}
\mbox{Cov}[ N_{(rx_{0},rx_{1})},N_{(rx_{0},rx_{2})} ] = \Sigma_{1,2}(r) + \frac{1+\gamma_{\mathrm{E}}}{2\pi^{2}} + \bigO \Big( \frac{\log r}{r} \Big), \qquad \mbox{ as } r \to + \infty,
\end{equation}
where $\Sigma_{1,2}$ is given by \eqref{cov thm s1 neq 0}. Note that the leading term in \eqref{asymp for the cov} is proportional to $\log r$. Interestingly, this contrasts with the asymptotics of the covariances of the Airy and Bessel counting functions which remain bounded, see \cite[below Remark 1]{ChCl3} and \cite[eq (1.17)]{Ch2018}.

\paragraph{Asymptotics for the mean, variance and covariance of a conditional counting function.} If $s_{p} = 0$ for a certain $p \in \{1,\ldots,m\}$, then we can rewrite \eqref{F generating function} as follows
\begin{align}
F(r\vec{x},\vec{s}) & = \mathbb{P}\big( N_{(rx_{p-1},rx_{p})}=0 \big) \mathbb{E}\bigg[ \prod_{j \neq p} s_{j}^{N_{(rx_{j-1},rx_{j})}} \Big| N_{(rx_{p-1},rx_{p})}=0 \bigg] \nonumber \\
& = F((rx_{p-1},rx_{p}),0) \, \mathbb{E}\Bigg[ \prod_{j=0}^{p-2} e^{-u_{j} \widehat{N}_{r\mathcal{I}_{j}}} \prod_{j=p+1}^{m}e^{u_{j} \widehat{N}_{r\mathcal{I}_{j}}}\Bigg], \label{F in remark sp=0 first eq}
\end{align}
where $\widehat{N}_{r\mathcal{I}_{j}}$ is the conditional random variable defined by
\begin{align*}
\widehat{N}_{r\mathcal{I}_{j}} := N_{r\mathcal{I}_{j}} | \Big( N_{(rx_{p-1},rx_{p})}=0 \Big), \qquad \mathcal{I}_{j} = \left\{ \begin{array}{l l}
(x_{p},x_{j}), & \mbox{if } j \in \{p+1,\ldots,m\}, \\
(x_{j},x_{p-1}), & \mbox{if } j \in \{0,\ldots,p-2\}.
\end{array} \right.
\end{align*}
Then, proceeding as in the derivations of \eqref{variance asymp s1 neq 0} and \eqref{asymp for the cov}, we obtain the following new asymptotic formulas
\begin{align*}
& \mathbb{E}\big[\widehat{N}_{r \mathcal{I}_{j}}\big] = \hat{\mu}_{j}(r) + \bigO \Big( \frac{\log r}{r} \Big), \\
& \mbox{Var}\big[\widehat{N}_{r \mathcal{I}_{j}}\big] = \hat{\sigma}_{j}^{2}(r) + \frac{1 + \gamma_{\mathrm{E}}}{\pi^{2}} + \bigO \Big( \frac{\log r}{r} \Big), \\
& \mbox{Cov}[ \widehat{N}_{r \mathcal{I}_{j}},\widehat{N}_{r \mathcal{I}_{k}} ] = \hat{\Sigma}_{j,k} + \bigO \Big( \frac{\log r}{r} \Big),
\end{align*}
as $r \to + \infty$, for any $j,k \in \{0,\ldots,m \}\setminus \{p-1,p\}$, and where $\hat{\mu}_{j}$, $\hat{\sigma}_{j}^{2}$ and $\hat{\Sigma}_{j,k}$ are given by \eqref{mean sp=0}, \eqref{var sp=0} and \eqref{cov sp=0}, respectively. Note that, in contrast to \eqref{asymp for the cov}, $\mbox{Cov}[ \widehat{N}_{r \mathcal{I}_{j}},\widehat{N}_{r \mathcal{I}_{k}} ]$ remains bounded as $r \to + \infty$.


\paragraph{Outline.} The rest of the paper is organized as follows. In Section \ref{Section: model RH problem}, using the fact the sine kernel is \textit{integrable} in the sense of Its, Izergin, Korepin and Slavnov (IIKS) \cite{IIKS}, we express the kernel $K_{\vec{x},\vec{s}}$ of the resolvant operator associated to $F$ in terms of an RH problem whose solution is denoted $Y$ (following \cite{DeiftItsZhou, BDIK2015, BDIK2017, BDIK2019}). Next, we transform the RH problem for $Y$ into a new RH problem with constant jumps whose solution is denoted $\Phi$. In Section \ref{Section: diff id}, we obtain a differential identity which expresses $\partial_{s_k} \log F(r \vec{x},\vec{s})$ (for an arbitrary $k\in\{1,\ldots,m\}$) in terms of $\Phi$. We obtain large $r$ asymptotics for $\Phi$ with $s_{1}, \ldots,s_{m} \in (0,+\infty)$ in Section \ref{Section: Steepest descent with s1>0} via the Deift/Zhou steepest descent method. In Section \ref{Section: integration s1 >0}, we substitute the asymptotics of $\Phi$ in the differential identity to obtain large $r$ asymptotics for $\partial_{s_k} \log F(r \vec{x},\vec{s})$. Then, we proceed with the successive integrations of these asymptotics in $s_{1},\ldots,s_{m}$, which finishes the proof of Theorem \ref{thm:s1 neq 0}. Sections \ref{Section: Steepest descent with sp=0} and \ref{Section: integration sp=0} are devoted to the proof of Theorem \ref{thm:sp=0} (with $s_{p} = 0$), and are organized in the same way as Sections \ref{Section: Steepest descent with s1>0} and \ref{Section: integration s1 >0}.


\section{Model RH problem}\label{Section: model RH problem}
Let us denote $\mathcal{K}_{\vec{x},\vec{s}}$ for the integral operator that appears in the definition \eqref{F Fredholm} of $F(\vec{x},\vec{s})$, that is, 
\begin{equation}\label{integral operator K vec x}
\mathcal{K}_{\vec{x},\vec{s}} = \sum_{j=1}^{m}(1-s_{j})\mathcal{K}|_{(x_{j-1},x_{j})}.
\end{equation}
In Section \ref{Section: diff id}, we will express $\partial_{s_{k}} \log F(\vec{x},\vec{s})$, $k =1,\ldots,m$, in terms of the resolvent operator 
\begin{equation}\label{def resolvent}
\mathcal{R}_{\vec{x},\vec{s}} = (1-\mathcal{K}_{\vec{x},\vec{s}})^{-1}-1 = (1-\mathcal{K}_{\vec{x},\vec{s}})^{-1}\mathcal{K}_{\vec{x},\vec{s}}.
\end{equation}
The goal of this section is to relate $\mathcal{R}_{\vec{x},\vec{s}}$ to a convenient model RH problem. 

\medskip We will proceed in three steps: 
\begin{enumerate}
\item \vspace{-0.1cm} The kernel $K_{\vec{x},\vec{s}}$ of the operator $\mathcal{K}_{\vec{x},\vec{s}}$ is \textit{integrable} in the sense of IIKS \cite{IIKS}, which means that it can be written in the form
\begin{equation}\label{integrable kernel}
K_{\vec{x},\vec{s}}(u,v) = \frac{f^{T}(u)g(v)}{u-v}, 
\end{equation}
for suitable vector valued functions $f$ and $g$ which are written down in \eqref{def of f and g} below. This fact will allow us to use a result of Deift, Its and Zhou \cite{DeiftItsZhou} to express the resolvent operator in terms of an RH problem whose solution is denoted $Y$.
\item As a preparation for the third step, we will consider another RH problem, whose solution $\Phi_{\sin}$ can be explicitly written in terms of elementary functions.
\item Finally, using the properties of $\Phi_{\sin}$, we will transform the RH problem for $Y$ into a new RH problem with constant jumps. The solution to this RH problem is denoted $\Phi$ and will play a central role in the next sections.
\end{enumerate}
\begin{remark}
The above steps 2 and 3 will allow us to work with $\Phi$ instead of $Y$. In Section \ref{Section: diff id}, we will take advantage of the fact that $\Phi$ has constant jumps to simplify the differential identity using a Lax pair. In the same spirit, other RH problems with constant jumps related to the Airy and Bessel processes have also been used in \cite{ClaeysDoeraene, ChDoe} to simplify the analysis. However, we mention that if $m=1$ our RH problem for $Y$ reduces to the RH problem considered by Bothner et al. in \cite{BDIK2015, BDIK2017}, and that their approach is different from ours and does not rely on the steps 2 and 3 above; instead they have successfully performed a Deift--Zhou steepest descent analysis directly on $Y$ (though in a different regime of the parameters than in this paper). 
\end{remark}

It is directly seen from \eqref{integral operator K vec x} and \eqref{Sine kernel} that $K_{\vec{x},\vec{s}}$ can be written in the form \eqref{integrable kernel} with
\begin{equation}\label{def of f and g}
f(u) = \begin{pmatrix}
\sin(u) \\
-\cos(u)
\end{pmatrix}, \qquad g(v) = \frac{1}{\pi} \begin{pmatrix}
\sum_{j=1}^{m} \chi_{(x_{j-1},x_{j})}(v)(1-s_{j})\cos v  \\[0.2cm]
\sum_{j=1}^{m} \chi_{(x_{j-1},x_{j})}(v) (1-s_{j})\sin v  
\end{pmatrix},
\end{equation}
where we recall that for any Borel set $A \subset \mathbb{R}$, $\chi_{A}(u) = 1$ if $u \in A$ and $\chi_{A}(u) = 0$ otherwise. 

\medskip In the sine point process, for all bounded Borel set $B$ with non-zero Lebesgue measure, we have $\mathbb{P}(N_{B} = 0) > 0$. Therefore, from \eqref{F Fredholm} and \eqref{F generating function}, we have 
\begin{equation}\label{observation}
F(\vec{x},\vec{s}) = \det(1-\mathcal{K}_{\vec{x},\vec{s}}) \geq \mathbb{P}(N_{(x_{0},x_{1})} =0) >0,
\end{equation}
which implies in particular that $1-\mathcal{K}_{\vec{x},\vec{s}}$ is invertible and that $\mathcal{R}_{\vec{x},\vec{s}}$ exists. Let us now define the matrix $Y$ by
\begin{equation}\label{def of Y}
Y(z) = I - \int_{x_{0}}^{x_m} \frac{\widetilde{f}(u)g^T(u)}{u - z} du, \qquad \widetilde{f}(u) = \big((1-\mathcal{K}_{\vec{x},\vec{s}})^{-1}f\big)(u).
\end{equation}
The function $Y$ satisfies the following RH problem \cite[Lemma 2.12]{DeiftItsZhou}.
\subsubsection*{RH problem for Y}
\begin{itemize}
\item[(a)] $Y : \mathbb{C}\setminus [x_{0},x_{m}] \to \mathbb{C}^{2\times 2}$ is analytic
\item[(b)] For $u \in (x_{0},x_{m}) \setminus \{x_1,\ldots,x_{m-1}\}$, the limits $\lim_{\epsilon \to 0_{+}} Y(u\pm i \epsilon)$ exist, are denoted $Y_{+}(u)$ and $Y_{-}(u)$ respectively, are continuous as functions of $u$, and satisfy furthermore the jump relation
\begin{equation}
Y_{+}(u) = Y_{-}(u)J_{Y}(u), \qquad J_{Y}(u) = I - 2\pi i f(u)g^{T}(u).
\end{equation}
\item[(c)] $Y(z) = I + \bigO(z^{-1})$ as $z \to \infty$.
\item[(d)] $Y(z) = \bigO(\log(z-x_j))$ as $z \to x_{j}$, for each $j = 0,\ldots,m$.
\end{itemize}
From \cite{DeiftItsZhou}, the kernel $R_{\vec{x},\vec{s}}$ of the resolvent operator $\mathcal{R}_{\vec{x},\vec{s}}$ can be written as 
\begin{equation}\label{R in terms of Y}
R_{\vec{x},\vec{s}}(u,v) = \frac{\widetilde{f}^{T}(u)\widetilde{g}(v)}{u-v}, \qquad u,v \in (x_{0},x_{m}),
\end{equation}
with $\widetilde{f}$ and $\widetilde{g}$ expressed in terms of $Y$ as follows:
\begin{equation}
\widetilde{f}(u) = Y_{+}(u)f(u)  \qquad \mbox{and} \qquad \widetilde{g}(v) = (Y_{+}^{-1}(v))^{T}g(v).
\end{equation}
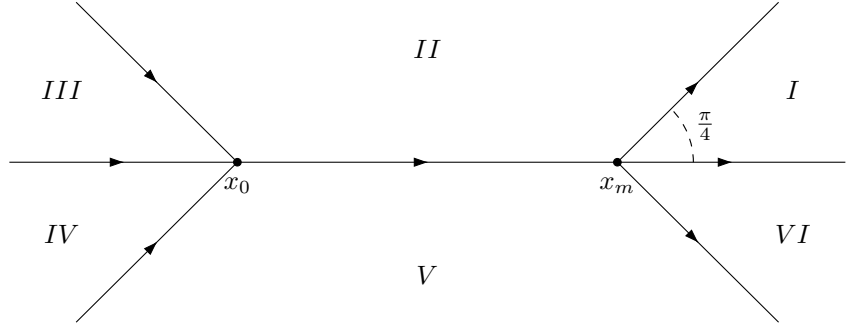
\begin{figure}
\centering
\begin{tikzpicture}
\node at (3,0) {};
\draw (0,0) -- (11,0);
\draw (3,0) -- ($(3,0)+(135:3)$);
\draw (3,0) -- ($(3,0)+(-135:3)$);
\draw (8,0) -- ($(8,0)+(45:3)$);
\draw (8,0) -- ($(8,0)+(-45:3)$);

\draw[fill] (3,0) circle (0.05);
\draw[fill] (8,0) circle (0.05);

\node at (3,-0.3) {$x_{0}$};
\node at (8,-0.3) {$x_{m}$};

\node at ($(8,0)+(22.5:2.5)$) {$I$};
\node at ($(8,0)+(-22.5:2.5)$) {$VI$};
\node at ($(5.5,0)+(90:1.5)$) {$II$};
\node at ($(5.5,0)+(-90:1.5)$) {$V$};
\node at ($(3,0)+(180-22.5:2.5)$) {$III$};
\node at ($(3,0)+(-180+22.5:2.5)$) {$IV$};

\draw[black,arrows={-Triangle[length=0.18cm,width=0.12cm]}]
($(3,0)+(135:1.5)$) --  ++(-45:0.001);
\draw[black,arrows={-Triangle[length=0.18cm,width=0.12cm]}]
($(3,0)+(-135:1.5)$) --  ++(45:0.001);
\draw[black,arrows={-Triangle[length=0.18cm,width=0.12cm]}]
(1.5,0) --  ++(0:0.001);

\draw[black,arrows={-Triangle[length=0.18cm,width=0.12cm]}]
(5.5,0) --  ++(0:0.001);

\draw[black,arrows={-Triangle[length=0.18cm,width=0.12cm]}]
($(8,0)+(45:1.5)$) --  ++(45:0.001);
\draw[black,arrows={-Triangle[length=0.18cm,width=0.12cm]}]
($(8,0)+(-45:1.5)$) --  ++(-45:0.001);
\draw[black,arrows={-Triangle[length=0.18cm,width=0.12cm]}]
(9.5,0) --  ++(0:0.001);

\draw[black,dashed,line width=0.15 mm] ([shift=(0:1cm)]8,0) arc (0:45:1cm);
\node at ($(8,0)+(22.5:1.25)$) {$\frac{\pi}{4}$};
\end{tikzpicture}
\caption{Jump contours $\Sigma_{\sin}$ for the RH problem for $\Phi_{\sin}$.}
\label{fig:contour for Phi Sine}
\end{figure}
Let $I,\ldots,VI$ be the six regions shown in Figure \ref{fig:contour for Phi Sine}. We consider the following RH problem, whose solution is denoted $\Phi_{\sin}$. 
\subsection*{RH problem for $\Phi_{\sin}$}

\begin{itemize}
\item[(a)] $\Phi_{\sin} : \mathbb{C}\setminus \Sigma_{\sin} \to \mathbb{C}^{2\times 2}$ is analytic, where 
\begin{equation}
\Sigma_{\sin} = \mathbb{R} \cup (x_{m} + e^{\pm \frac{\pi i }{4}}\mathbb{R}^{+}) \cup (x_{0} + e^{\pm \frac{3\pi i}{4}}\mathbb{R}^{+})
\end{equation}
is oriented as shown in Figure \ref{fig:contour for Phi Sine} and $\mathbb{R}^{+}:=(0,+\infty)$.
\item[(b)] The jumps are given by
\begin{align*}
& \Phi_{\sin,+}(z) = \Phi_{\sin,-}(z) \begin{pmatrix}
1 & 1 \\
0 & 1
\end{pmatrix}, & & z \in (x_{0},x_{m}), \\
& \Phi_{\sin,+}(z) = \Phi_{\sin,-}(z) \begin{pmatrix}
1 & 0 \\
1 & 1
\end{pmatrix}, & & z \in (x_{m} + e^{\pm \frac{\pi i }{4}}\mathbb{R}^{+}) \cup (x_{0} + e^{\pm \frac{3\pi i}{4}}\mathbb{R}^{+}),  \\
& \Phi_{\sin,+}(z) = \Phi_{\sin,-}(z) \begin{pmatrix}
0 & 1 \\
-1 & 0
\end{pmatrix}, & & z \in (-\infty,x_{0})\cup (x_{m},+\infty). 
\end{align*}
\item[(c)] As $z \to \infty$, we have
\begin{equation}
\Phi_{\sin}(z) = Ne^{-\frac{\pi i}{4}\sigma_{3}} \Big(I + \bigO\big(  e^{-2|\Im z|} \chi_{\substack{ \\[0.1cm] II\cup V}}(z)\big) \Big) e^{-iz \sigma_{3}} \times \left\{  \begin{array}{l l}
I, & \mbox{if } \Im z > 0, \\
\begin{pmatrix}
0 & -1 \\
1 & 0
\end{pmatrix}, & \mbox{if } \Im z < 0,
\end{array} \right.
\end{equation}
where 
\begin{equation}
\sigma_{3} = \begin{pmatrix}
1 & 0 \\ 0 & -1
\end{pmatrix}, \qquad N = \frac{1}{\sqrt{2}}\begin{pmatrix}
1 & i \\ i & 1
\end{pmatrix}, \qquad \chi_{\substack{ \\[0.1cm] II\cup V}}(z) = \left\{ \begin{array}{l l}
1, & \mbox{if } z \in II \cup V, \\
0, & \mbox{otherwise}.
\end{array} \right.
\end{equation}

As $z \to x_{0}$ and as $z \to x_{m}$, we have $\Phi_{\sin}(z) = \bigO(1)$.
\end{itemize}
The unique solution to the above RH problem is explicitly given by
\begin{equation}\label{Phi Sine explicit}
\Phi_{\sin}(z) = N e^{-\frac{\pi i }{4}\sigma_{3}} \times \left\{ \begin{array}{l l}
\begin{pmatrix}
e^{-iz} & 0 \\ 0 & e^{iz}
\end{pmatrix}, & z \in I \cup III, \\[0.3cm]
\begin{pmatrix}
e^{-iz} & 0 \\
e^{iz} & e^{iz}
\end{pmatrix}, & z \in II, \\[0.3cm]
\begin{pmatrix}
0 & -e^{-iz} \\
e^{iz} & 0
\end{pmatrix}, & z \in IV \cup VI, \\[0.3cm]
\begin{pmatrix}
e^{-iz} & -e^{-iz} \\
e^{iz} & 0
\end{pmatrix}, & z \in V.
\end{array} \right.
\end{equation}

\hspace{-0.55cm}Now, we use $\Phi_{\sin}$ to transform the RH problem for $Y$. Let us consider 
\begin{equation}\label{Phi Y Phi Sine relation}
\Phi(z) = Y(z) \Phi_{\sin}(z).
\end{equation}
Since $Y$ is analytic on $\mathbb{C} \setminus [x_{0},x_{m}]$, the jumps $J_{\Phi}:=\Phi_{-}^{-1}\Phi_{+}$ coincide with $\Phi_{\sin,-}^{-1}\Phi_{\sin,+}$ on $\Sigma_{\sin} \setminus [x_{0},x_{m}]$. On $(x_{0},x_{m})$, we have
\begin{equation}\label{lol 12}
J_{\Phi}(u) = \Phi_{\sin,-}^{-1}(u)J_{Y}(u)\Phi_{\sin,+}(u), \qquad u \in (x_{0},x_{m}).
\end{equation}
which, using the explicit expression for $\Phi_{\sin}$ given by \eqref{Phi Sine explicit}, simplifies to
\begin{equation}
J_{\Phi}(u) = \begin{pmatrix}
1 & \sum_{j=1}^{m}s_{j}\chi_{(x_{j-1},x_{j})}(u) \\ 0 & 1
\end{pmatrix}, \qquad u \in (x_{0},x_{m}).
\end{equation}
Now, it directly follows from the properties of $Y$ and $\Phi_{\sin}$ that $\Phi$ satisfies the following RH problem. For convenience, we define $s_{0} := 1$, $s_{m+1} := 1$.
\subsection*{RH problem for $\Phi$}
\begin{itemize}
\item[(a)] $\Phi : \mathbb{C}\setminus \Sigma \to \mathbb{C}^{2\times 2}$ is analytic, where $\Sigma = \Sigma_{\sin}$ is shown in Figure \ref{fig:contour for Phi Sine}.
\item[(b)] The jumps are given by
\begin{align}
& \Phi_{+}(z) = \Phi_{-}(z) \begin{pmatrix}
1 & s_{j} \\
0 & 1
\end{pmatrix}, & & z \in (x_{j-1},x_{j}), \,\, j = 1,\ldots,m, \label{Jumps Phi on xj-1 xj} \\
& \Phi_{+}(z) = \Phi_{-}(z) \begin{pmatrix}
1 & 0 \\
1 & 1
\end{pmatrix}, & & z \in (x_{m} + e^{\pm \frac{\pi i }{4}}\mathbb{R}^{+}) \cup (x_{0} + e^{\pm \frac{3\pi i}{4}}\mathbb{R}^{+}), \\
& \Phi_{+}(z) = \Phi_{-}(z) \begin{pmatrix}
0 & 1 \\
-1 & 0
\end{pmatrix}, & & z \in (-\infty,x_{0})\cup (x_{m},+\infty). \label{Jumps Phi on R setminus support}
\end{align}
\item[(c)] As $z \to \infty$, we have
\begin{equation}\label{Asymp of Phi near infinity}
\Phi(z) = \Big(I + \frac{\Phi_{1}}{z} + \bigO(z^{-2})\Big) Ne^{-\frac{\pi i}{4}\sigma_{3}} e^{-iz \sigma_{3}} \times \left\{ \begin{array}{l l}
I, & \mbox{if } \Im z > 0, \\
\begin{pmatrix}
0 & -1 \\
1 & 0
\end{pmatrix}, & \mbox{if } \Im z < 0,
\end{array} \right.
\end{equation}
for a certain traceless matrix $\Phi_{1} = \Phi_{1}(\vec{x},\vec{s})$ independent of $z$. 

As $z \to x_{j}$, $j \in \{0,1,\ldots,m\}$, we have
\begin{equation}\label{Asymp of Phi near x_j}
\Phi(z) = G_{j}(z;\vec{x},\vec{s}) \begin{pmatrix}
1 & - \frac{s_{j+1}-s_{j}}{2\pi i} \log (z-x_{j}) \\
0 & 1
\end{pmatrix}V_{j}(z) H(z), 
\end{equation}
where $G_{j}$ is analytic in a neighborhood of $x_{j}$,  satisfies $\det G_{j} \equiv 1$ and 
\begin{equation}
V_{j}(z) = \left\{ \begin{array}{l l}
I, & \Im z > 0, \\
\begin{pmatrix}
1 & -s_{j+1} \\
0 & 1
\end{pmatrix}, & \Im z < 0,
\end{array} \right. \qquad H(z) = \left\{ \begin{array}{l l}
I, & z \in II \cup V, \\
\begin{pmatrix}
1 & 0 \\
-1 & 1
\end{pmatrix}, & z \in I \cup III, \\
\begin{pmatrix}
1 & 0 \\
1 & 1 
\end{pmatrix}, & z \in IV \cup VI.
\end{array} \right.
\end{equation}
\end{itemize}
\begin{remark}
The solution to the RH problem for $\Phi$ is unique; this follows from standard arguments and the fact that the jumps for $\Phi$ have determinant $1$ (see e.g. \cite[Theorem 7.18]{Deift}). Proving existence of a given RH problem is in general a more difficult task than proving uniqueness, but in our case this directly follows from the fact that $(1-\mathcal{K}_{\vec{x},\vec{s}})^{-1}$ exists (see \eqref{observation}), and from the explicit formulas \eqref{Phi Y Phi Sine relation} and \eqref{def of Y}. 
\end{remark}
After substituting \eqref{def of f and g}, \eqref{Phi Sine explicit} and \eqref{Phi Y Phi Sine relation} in \eqref{R in terms of Y}, we obtain 
\begin{equation}\label{R in terms of Phi}
R_{\vec{x},\vec{s}}(u,u) = \frac{1-s_{k}}{2\pi i} [\Phi^{-1}(u;\vec{x},\vec{s})\partial_{u}(\Phi(u;\vec{x},\vec{s}))]_{21}, \qquad u \in (x_{k-1},x_{k}), \quad k = 1,\ldots,m.
\end{equation}
\begin{remark}
In \eqref{R in terms of Phi}, we do not indicate whether we use the $+$ or $-$ boundary values of $\Phi$, but this is without ambiguity. Indeed, from the jumps for $\Phi$ on $(x_{0},x_{m})$ given by \eqref{Jumps Phi on xj-1 xj}, we easily verify that
\begin{equation}
[\Phi_{+}^{-1}(u;\vec{x},\vec{s}) \partial_{u} \Phi_{+}(u;\vec{x},\vec{s})]_{21} = [\Phi_{-}^{-1}(u;\vec{x},\vec{s}) \partial_{u} \Phi_{-}(u;\vec{x},\vec{s})]_{21}.
\end{equation}
\end{remark}

\section{Differential identity}\label{Section: diff id}
By standard properties of trace class operators, for any $k \in \{1,\ldots,m\}$, we have
\begin{equation}\label{lol1}
\begin{array}{r c l}
\ds \partial_{s_{k}} \log F(\vec{x},\vec{s}) & = & \ds \partial_{s_{k}} \log \det(1-\mathcal{K}_{\vec{x},\vec{s}}) =  - \mbox{Tr}\Big( (1-\mathcal{K}_{\vec{x},\vec{s}})^{-1}\partial_{s_{k}} \mathcal{K}_{\vec{x},\vec{s}} \Big) \\
& = & \ds \frac{1}{1-s_{k}} \mbox{Tr}\Big( (1-\mathcal{K}_{\vec{x},\vec{s}})^{-1}\mathcal{K}_{\vec{x},\vec{s}}|_{(x_{k-1},x_{k})} \Big) \\
& = & \ds \frac{1}{1-s_{k}} \mbox{Tr} \Big( \mathcal{R}_{\vec{x},\vec{s}}|_{(x_{k-1},x_{k})} \Big) = \frac{1}{1-s_{k}}\int_{x_{k-1}}^{x_{k}}R_{\vec{x},\vec{s}}(u,u)du,
\end{array}
\end{equation}
where we recall that $\mathcal{R}_{\vec{x},\vec{s}}$ is defined by \eqref{def resolvent}. Substituting the expression \eqref{R in terms of Phi} for $R_{\vec{x},\vec{s}}$ in \eqref{lol1}, we obtain the following differential identity
\begin{equation}\label{diff identity integral form}
\partial_{s_{k}} \log F(\vec{x},\vec{s}) = \frac{1}{2\pi i} \int_{x_{k-1}}^{x_{k}}[\Phi^{-1}(u;\vec{x},\vec{s})\partial_{u}\Phi(u;\vec{x},\vec{s})]_{21}du.
\end{equation}
We implicitly assumed $s_{k} \neq 1$ in \eqref{lol1}. However, recall that $F(\vec{x},\vec{s})$ is an entire function of $s_{k}$ (see \cite[Theorem 2]{Soshnikov}) and that $\left. \det (1-\mathcal{K}_{\vec{x},\vec{s}}) \right|_{s_{k} = 1} > 0$ (see \eqref{observation}). Therefore, the left-hand side of \eqref{diff identity integral form} is well-defined at $s_{k} = 1$, and \eqref{diff identity integral form} also holds for $s_{k} = 1$ by continuity.

\vspace{0.2cm}\hspace{-0.55cm}After replacing $\vec{x}$ by $r \vec{x}$ in \eqref{diff identity integral form}, we obtain
\begin{equation}\label{diff identity integral form with r}
\partial_{s_{k}} \log F(r\vec{x},\vec{s}) = \frac{1}{2\pi i} \int_{x_{k-1}}^{x_{k}}[\Phi^{-1}(ru;r\vec{x},\vec{s})\partial_{u}\Phi(ru;r\vec{x},\vec{s})]_{21}du.
\end{equation}
Our goal for the rest of this section is to simplify the integral on the right-hand side of \eqref{diff identity integral form with r}. For this, we will study a Lax pair associated to $\Phi$. We first focus on some properties of $\partial_{z} \Phi(rz;r\vec{x},\vec{s})$. Since the jumps for $\Phi$ are independent of $z$, we have
\begin{equation}\label{A matrix}
\partial_{z} \Phi(rz;r\vec{x},\vec{s}) = A(z) \Phi(rz;r\vec{x},\vec{s}),
\end{equation}
where $z \mapsto A(z)$ is analytic for $z \in \mathbb{C}\setminus \{x_{0},\ldots,x_{m}\}$. $A$ also depends on $r$, $\vec{x}$ and $\vec{s}$, even though this is not indicated in the notation. Furthermore, since $\det \Phi \equiv 1$, $A$ is traceless. From the asymptotics for $\Phi$ at $x_{0},\ldots,x_{m}$ and at $\infty$, we conclude that $A$ takes the form
\begin{equation}\label{A explicit recall}
A(z) = \begin{pmatrix}
0 & -r \\ r & 0
\end{pmatrix} + \sum_{j=0}^{m} \frac{A_{j}}{z-x_{j}},
\end{equation}
where the matrices $A_{j}=A_{j}(r,\vec{x},\vec{s})$ are traceless and given by
\begin{align}
A_{j} & = -\frac{s_{j+1}-s_{j}}{2\pi i}(G_{j} \sigma_{+}G_{j}^{-1})(rx_{j};r\vec{x},\vec{s}) \nonumber  \\
& = -\frac{s_{j+1}-s_{j}}{2\pi i} \begin{pmatrix}
-G_{j,11}G_{j,21} & G_{j,11}^{2} \\ -G_{j,21}^{2} & G_{j,11}G_{j,21}
\end{pmatrix}, \qquad \mbox{where } \sigma_{+} = \begin{pmatrix}
0 & 1 \\ 0 & 0
\end{pmatrix}. \label{A_j}
\end{align}
The integrand on the right-hand side of \eqref{diff identity integral form with r} can now be rewritten using \eqref{A matrix}. Since $A$ is traceless and $\det \Phi \equiv 1$, we obtain
\begin{multline}
[\Phi^{-1}(rz;r\vec{x},\vec{s})\partial_{z} \big( \Phi(rz;r\vec{x},\vec{s}) \big)]_{21} = [\Phi^{-1}(rz;r\vec{x},\vec{s})A(z)\Phi(rz;r\vec{x},\vec{s})]_{21} \\ = \Phi_{11}^{2} A_{21} - \Phi_{21}^{2}A_{12}-2\Phi_{11}\Phi_{21}A_{11}.
\end{multline}
By substituting \eqref{A explicit recall} in the above equation, we infer that
\begin{multline}\label{explicit integrant}
[\Phi^{-1}(rz;r\vec{x},\vec{s})\partial_{z} \big( \Phi(rz;r\vec{x},\vec{s}) \big)]_{21} = (\Phi\sigma_{+}\Phi^{-1})_{12}(rz;r\vec{x},\vec{s}) \Big[ r + \sum_{j=0}^{m} \frac{A_{j,21}}{z-x_{j}} \Big] \\
+ (\Phi\sigma_{+}\Phi^{-1})_{21}(rz;r\vec{x},\vec{s})\Big[ -r+\sum_{j=0}^{m} \frac{A_{j,12}}{z-x_{j}}  \Big] + 2 (\Phi\sigma_{+}\Phi^{-1})_{11}(rz;r\vec{x},\vec{s})\sum_{j=0}^{m} \frac{A_{j,11}}{z-x_{j}}.
\end{multline}
Let us define 
\begin{equation}
B(z) = \partial_{s_{k}}\Phi(rz;r\vec{x},\vec{s}) \Phi(rz;r\vec{x},\vec{s})^{-1}.
\end{equation}
From the RH problem for $\Phi$, we deduce that $B$ satisfies the following RH problem.
\subsubsection*{RH problem for $B$} 
\begin{itemize}
\item[(a)] $B:\mathbb{C}\setminus [x_{k-1},x_{k}] \to \mathbb{C}^{2\times 2}$ is analytic.
\item[(b)] $B$ satisfies the jumps
\begin{equation}\label{boundary values in B why}
B_{+}(z) = B_{-}(z) + (\Phi_{-}\sigma_{+}\Phi_{-}^{-1})(rz;r\vec{x},\vec{s}), \qquad z \in (x_{k-1},x_{k}).
\end{equation}
\item[(c)] $B$ satisfies the following asymptotic behaviors
\begin{align}
& \hspace{-1cm} B(z) = \frac{\partial_{s_{k}}\Phi_{1}(r\vec{x},\vec{s})}{rz} + \bigO(z^{-2}), & & \mbox{as } z \to \infty, \label{F at inf} \\
& \hspace{-1cm} B(z) = \frac{\partial_{s_{k}}(s_{j+1}-s_{j})}{s_{j+1}-s_{j}}A_{j} \log(r(z-x_{j})) + B_{j}+o(1), & & \mbox{as } z \to x_{j}, \, j = 0,\ldots,m, \label{B at xj lol}
\end{align}
where $B_{j} = (\partial_{s_{k}}G_{j}G_{j}^{-1})(rx_{j};r\vec{x},\vec{s})$.
\end{itemize}
Using the jumps \eqref{Jumps Phi on xj-1 xj} for $\Phi$ on $(x_{k-1},x_{k})$, we note that
\begin{equation}
(\Phi_{-}\sigma_{+}\Phi_{-}^{-1})(rz;r\vec{x},\vec{s})=(\Phi_{+}\sigma_{+}\Phi_{+}^{-1})(rz;r\vec{x},\vec{s}), \qquad z \in (x_{k-1},x_{k}),
\end{equation}
which implies that $z \mapsto (\Phi\sigma_{+}\Phi^{-1})(rz;r\vec{x},\vec{s})$ is analytic for $z \in (x_{k-1},x_{k})$. In particular, we can replace $\Phi_{-}\sigma_{+}\Phi_{-}^{-1}$ in \eqref{boundary values in B why} by $\Phi\sigma_{+}\Phi^{-1}$ without ambiguity. By \eqref{boundary values in B why} and Cauchy's formula, we have
\begin{equation}\label{integral representation of B}
B(z) = \frac{1}{2\pi i} \int_{x_{k-1}}^{x_{k}} \frac{(\Phi \sigma_{+} \Phi^{-1})(ru;r\vec{x},\vec{s})}{u-z}du.
\end{equation}
Expanding \eqref{integral representation of B} as $z \to \infty$ and then comparing with \eqref{F at inf}, we obtain
\begin{align}
& -\frac{1}{2\pi i} \int_{x_{k-1}}^{x_{k}}(\Phi\sigma_{+}\Phi^{-1})(ru;r\vec{x},\vec{s})du = \frac{\partial_{s_{k}}\Phi_{1}(r\vec{x},\vec{s})}{r}. 
\end{align}
Now, we substitute \eqref{explicit integrant} in \eqref{diff identity integral form with r}, and then use the expansions of $B$ at $\infty$ and at $x_{j}$, $j=0,1,\ldots,m$ (given by \eqref{F at inf}) to simplify the integral. Since $\det A_{j} \equiv 0$ for $j = 0,\ldots,m$, we infer that the logarithmic terms in the expansions \eqref{B at xj lol} of $B(z)$ as $z \to x_{j}$ for $j = 0,\ldots,m$ do not contribute in \eqref{diff identity integral form with r}, and after a computation we obtain
\begin{multline*}
\partial_{s_{k}} \log F(r\vec{x},\vec{s})= \partial_{s_{k}} \Phi_{1,21}(r\vec{x},\vec{s})-\partial_{s_{k}} \Phi_{1,12}(r\vec{x},\vec{s})+\sum_{j=0}^{m}\Big(A_{j,21}B_{j,12}+A_{j,12}B_{j,21}+2A_{j,11}B_{j,11} \Big).
\end{multline*}
The above formula can be further simplified by using the explicit expressions for the $A_{j}$'s and $B_{j}$'s given by \eqref{A_j} and below \eqref{B at xj lol}. After some simplifications, which use $\det G_{j} \equiv 1$, we get
\begin{equation}\label{DIFF identity final form general case}
\partial_{s_{k}} \log F(r\vec{x},\vec{s})= K_{\infty} + \sum_{j=0}^{m}K_{x_{j}},
\end{equation}
where
\begin{align}
& K_{\infty} = \partial_{s_{k}} \Phi_{1,21}(r\vec{x},\vec{s})-\partial_{s_{k}} \Phi_{1,12}(r\vec{x},\vec{s}), \label{K inf} \\
& K_{x_{j}} = -\frac{s_{j+1}-s_{j}}{2\pi i} \Big( G_{j,11}\partial_{s_{k}}G_{j,21} - G_{j,21}\partial_{s_{k}} G_{j,11} \Big)(rx_{j};r\vec{x},\vec{s}). \label{K xj} 
\end{align}
\section{Riemann-Hilbert analysis for $s_{1},\ldots,s_{m} \in (0,+\infty)$}\label{Section: Steepest descent with s1>0}
In this section, we employ the Deift/Zhou steepest descent method to obtain large $r$ asymptotics for $\Phi(rz;r\vec{x},\vec{s})$ uniformly for $z$ in different regions of the complex $z$-plane.

\medskip \noindent At the level of the parameters, we assume that
\begin{itemize}
\item \vspace{-0.1cm} $s_{1},\ldots,s_{m}$ are in a compact subset of $(0,+\infty)$,
\item \vspace{-0.1cm} $x_{0},\ldots,x_{m}$ are in a compact subset of $\mathbb{R}$, 
\item \vspace{-0.1cm} there exists $\delta > 0$ independent of $r$ such that
\begin{equation}\label{assumption delta}
\min_{0 \leq j < k \leq m} x_{k}-x_{j} \geq \delta.
\end{equation}
\end{itemize}

\subsection{Normalization of the RH problem with $g$-function}\label{subsection: g function s1 neq 0}
The main purpose of the first transformation is to obtain a new RH problem whose solution $T$ remains bounded at $\infty$. Let us define
\begin{equation}\label{def of T s1 neq 0}
T(z)= \Phi(rz;r\vec{x},\vec{s})e^{-rg(z)\sigma_{3}},
\end{equation}
where the $g$-function is given by
\begin{equation}\label{def of g s1 neq 0}
g(z)=\left\{ \begin{array}{l l}
-iz, & \mbox{if }\Im z > 0, \\
iz, & \mbox{if }\Im z < 0.
\end{array} \right.
\end{equation}
One easily verifies from \eqref{Asymp of Phi near infinity}, \eqref{def of T s1 neq 0} and \eqref{def of g s1 neq 0} that $T$ remains bounded at $\infty$, as desired. More precisely, we have
\begin{equation} \label{eq:Tasympinf s1 neq 0}
T(z) = \left( I + \frac{T_{1}}{z} + \bigO\left(z^{-2}\right) \right) Ne^{-\frac{\pi i}{4}\sigma_{3}}\left\{ \begin{array}{l l}
I, & \Im z > 0, \\
\begin{pmatrix}
0 & -1 \\ 1 & 0 
\end{pmatrix}, & \Im z < 0,
\end{array}	 \right. \qquad \mbox{as } z \to \infty,
\end{equation}
where
\begin{equation}\label{def of T1 s1 neq 0}
T_{1} = \frac{\Phi_{1}(r\vec{x},\vec{s})}{r}.
\end{equation}
We can obtain the jumps for $T$ straightforwardly from those of $\Phi$ and the relation $g_{+}(z)+g_{-}(z) = 0$ for $z \in \mathbb{R}$. For $z \in (x_{j-1},x_{j})$, $j=1,\ldots,m$, we have
\begin{equation}\label{factorization of the jump}
T_{-}(z)^{-1}T_{+}(z) = \begin{pmatrix}
e^{-2rg_{+}(z)} & s_{j} \\ 0 & e^{-2rg_{-}(z)}
\end{pmatrix} = \begin{pmatrix}
1 & 0 \\
s_{j}^{-1}e^{-2rg_{-}(z)} & 1
\end{pmatrix} \begin{pmatrix}
0 & s_{j} \\ -s_{j}^{-1} & 0
\end{pmatrix} \begin{pmatrix}
1 & 0 \\ 
s_{j}^{-1}e^{-2 r g_{+}(z)} & 1
\end{pmatrix}.
\end{equation}
where we have used $s_{j} \neq 0$ to factorize the jump matrix.
\subsection{Opening of the lenses}\label{subsection: S with s1 neq 0}
For $j = 1,\ldots,m$, we let the lenses $\gamma_{j,+}$ and $\gamma_{j,-}$ be open curves starting at $x_{j-1}$, ending at $x_{j}$ and lying in the upper and lower half plane, respectively. The region inside $\gamma_{j,+} \cup (x_{j-1},x_{j})$ is denoted $\Omega_{j,+}$, and the region inside $\gamma_{j,-} \cup (x_{j-1},x_{j})$ is denoted $\Omega_{j,-}$. In particular, $\Omega_{j,+}$ and $\Omega_{j,-}$ are subsets of the upper and lower half plane, respectively. The next transformation is defined by
\begin{equation}\label{def of S s1 neq 0}
S(z) = T(z) \prod_{j=1}^{m} \left\{ \begin{array}{l l}
\begin{pmatrix}
1 & 0 \\
-s_{j}^{-1}e^{-2rg(z)} & 1
\end{pmatrix}, & \mbox{if } z \in \Omega_{j,+}, \\
\begin{pmatrix}
1 & 0 \\
s_{j}^{-1}e^{-2rg(z)} & 1
\end{pmatrix}, & \mbox{if } z \in \Omega_{j,-}, \\
I, & \mbox{if } z \in \mathbb{C}\setminus(\Omega_{j,+}\cup \Omega_{j,-}).
\end{array} \right.
\end{equation}
\begin{figure}
\centering
\begin{tikzpicture}
\node at (3,0) {};
\draw (0,0) -- (13,0);
\draw (3,0) -- ($(3,0)+(135:3)$);
\draw (3,0) -- ($(3,0)+(-135:3)$);
\draw (10,0) -- ($(10,0)+(45:3)$);
\draw (10,0) -- ($(10,0)+(-45:3)$);

\draw[fill] (3,0) circle (0.05);
\draw[fill] (5,0) circle (0.05);
\draw[fill] (7,0) circle (0.05);
\draw[fill] (10,0) circle (0.05);

\node at (3,-0.3) {$x_{0}$};
\node at (5,-0.3) {$x_{1}$};
\node at (7,-0.3) {$x_{2}$};
\node at (10,-0.3) {$x_{m}$};

\draw[black,arrows={-Triangle[length=0.18cm,width=0.12cm]}]
($(3,0)+(135:1.5)$) --  ++(-45:0.001);
\draw[black,arrows={-Triangle[length=0.18cm,width=0.12cm]}]
($(3,0)+(-135:1.5)$) --  ++(45:0.001);
\draw[black,arrows={-Triangle[length=0.18cm,width=0.12cm]}]
(1.5,0) --  ++(0:0.001);

\draw[black,arrows={-Triangle[length=0.18cm,width=0.12cm]}]
(6.05,0) --  ++(0:0.001);
\draw[black,arrows={-Triangle[length=0.18cm,width=0.12cm]}]
(4.05,0) --  ++(0:0.001);
\draw[black,arrows={-Triangle[length=0.18cm,width=0.12cm]}]
(8.6,0) --  ++(0:0.001);

\draw[black,arrows={-Triangle[length=0.18cm,width=0.12cm]}]
($(10,0)+(45:1.5)$) --  ++(45:0.001);
\draw[black,arrows={-Triangle[length=0.18cm,width=0.12cm]}]
($(10,0)+(-45:1.5)$) --  ++(-45:0.001);
\draw[black,arrows={-Triangle[length=0.18cm,width=0.12cm]}]
(11.5,0) --  ++(0:0.001);

\draw (3,0) .. controls (3.7,0.8) and (4.3,0.8) .. (5,0);
\draw (3,0) .. controls (3.7,-0.8) and (4.3,-0.8) .. (5,0);
\draw (5,0) .. controls (5.7,0.8) and (6.3,0.8) .. (7,0);
\draw (5,0) .. controls (5.7,-0.8) and (6.3,-0.8) .. (7,0);
\draw (7,0) .. controls (8,1.3) and (9,1.3) .. (10,0);
\draw (7,0) .. controls (8,-1.3) and (9,-1.3) .. (10,0);

\draw[black,arrows={-Triangle[length=0.18cm,width=0.12cm]}]
(4.05,0.59) --  ++(0:0.001);
\draw[black,arrows={-Triangle[length=0.18cm,width=0.12cm]}]
(4.05,-0.59) --  ++(0:0.001);
\draw[black,arrows={-Triangle[length=0.18cm,width=0.12cm]}]
(6.05,0.59) --  ++(0:0.001);
\draw[black,arrows={-Triangle[length=0.18cm,width=0.12cm]}]
(6.05,-0.59) --  ++(0:0.001);
\draw[black,arrows={-Triangle[length=0.18cm,width=0.12cm]}]
(8.6,0.98) --  ++(0:0.001);
\draw[black,arrows={-Triangle[length=0.18cm,width=0.12cm]}]
(8.6,-0.98) --  ++(0:0.001);
\end{tikzpicture}
\caption{Jump contours $\Sigma_{S}$ for the RH problem for $S$ with $m=3$.}
\label{fig:contour for S s1 neq 0}
\end{figure}
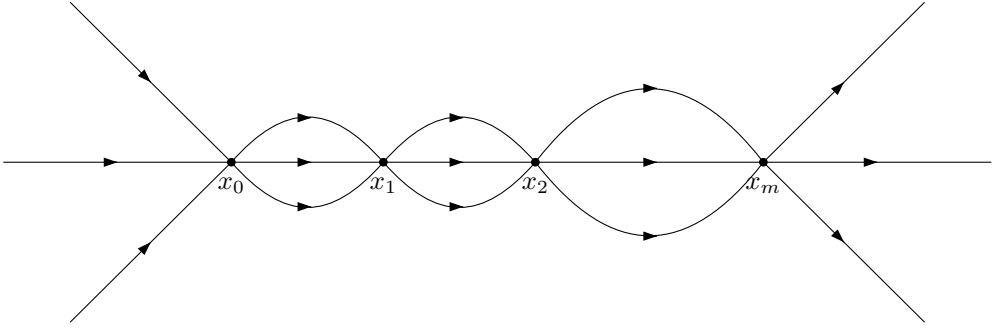
\hspace{-0.1cm}Using the factorization \eqref{factorization of the jump} and the properties of the RH problem for $\Phi$, we easily verify that $S$ satisfies the following RH problem.
\subsubsection*{RH problem for $S$}
\begin{enumerate}[label={(\alph*)}]
\item[(a)] $S : \C \backslash \Sigma_{S} \rightarrow \C^{2\times 2}$ is analytic, with
\begin{equation}\label{eq:defGamma s1 neq 0}
\Sigma_{S}=\mathbb{R}\cup \gamma_{+}\cup \gamma_{-}, \qquad \gamma_{\pm} = \bigcup_{j=0}^{m+1} \gamma_{j,\pm},
\end{equation}
where $\Sigma_{S}$ is oriented as shown in Figure \ref{fig:contour for S s1 neq 0} and
\begin{align*}
\gamma_{0,\pm} := x_{0} + e^{\pm \frac{3 \pi i}{4}}(0,+\infty), \qquad \gamma_{m+1,\pm} := x_{m} + e^{\pm \frac{\pi i}{4}}(0,+\infty).
\end{align*}
\item[(b)] The jumps for $S$ are given by
\begin{align}
& S_{+}(z) = S_{-}(z)\begin{pmatrix}
0 & s_{j} \\ -s_{j}^{-1} & 0
\end{pmatrix}, & & z \in (x_{j-1},x_{j}), \, j = 0,\ldots,m+1, \\
& S_{+}(z) = S_{-}(z)\begin{pmatrix}
1 & 0 \\
s_{j}^{-1} e^{-2 r g(z)} & 1
\end{pmatrix}, & & z \in \gamma_{j,\pm}, \, j = 0,\ldots,m+1, \label{jumps for S 2 first steepest}
\end{align}
where $x_{-1} := -\infty$ and $x_{m+1} := +\infty$ (recall that $s_{0} = s_{m+1} = 1$).
\item[(c)] As $z \rightarrow \infty$, we have
\begin{equation}
\label{eq:Sasympinf s1 neq 0}
S(z) = \left( I + \frac{T_{1}}{z} + \bigO\left(z^{-2}\right) \right) Ne^{-\frac{\pi i}{4}\sigma_{3}}\left\{ \begin{array}{l l}
I, & \Im z > 0, \\
\begin{pmatrix}
0 & -1 \\ 1 & 0 
\end{pmatrix}, & \Im z < 0.
\end{array}	 \right.
\end{equation}
As $z \to x_j$ from outside the lenses, $j = 0,\ldots, m$, we have
\begin{equation}\label{local behavior near -xj of S s1 neq 0}
S(z) = \begin{pmatrix}
\bigO(1) & \bigO(\log(z-x_{j})) \\
\bigO(1) & \bigO(\log(z-x_{j}))
\end{pmatrix}.
\end{equation}
\end{enumerate}
Since $\Re g(z) > 0$ for all $z \in \mathbb{C}\setminus \mathbb{R}$, it follows from \eqref{jumps for S 2 first steepest} that for any $\epsilon > 0$, there exists $c>0$ such that
\begin{align}\label{estimates for jumps for S exp small}
S_{-}(z)^{-1}S_{+}(z) - I = \bigO(e^{-c |z|r}), \qquad \mbox{as } r \to + \infty, 
\end{align}
uniformly for $z \in \gamma_{+}\cup \gamma_{-}$ such that $\min_{j \in \{0,\ldots,m\}}|z-x_{j}|\geq \epsilon$. The estimate \eqref{estimates for jumps for S exp small} does not hold for $z$ close to $x_{j}$ because $\Re g_{\pm}(z) = 0$ for $z \in \mathbb{R}$. More precisely, for $z \in \gamma_{+}\cup \gamma_{-}$ such that $\min_{j \in \{0,\ldots,m\}}|z-x_{j}|\leq \epsilon$, we only have 
\begin{align*}
S_{-}(z)^{-1}S_{+}(z) - I = o(1), \qquad \mbox{as } r \to + \infty.
\end{align*}

\subsection{Global parametrix}\label{subsection: Global param s1 neq 0}
Ignoring the jumps for $S$ on the lenses $\gamma_{+}\cup \gamma_{-}$, we are led to the following RH problem, whose solution is denoted $P^{(\infty)}$. We will show in Subsection \ref{subsection Small norm s1 neq 0} that $P^{(\infty)}$ is a good approximation for $S$ away from neighborhoods of $x_{j}$, $j = 0,1,\ldots,m$. 
\subsubsection*{RH problem for $P^{(\infty)}$}
\begin{enumerate}[label={(\alph*)}]
\item[(a)] $P^{(\infty)} : \C  \setminus \mathbb{R} \rightarrow \C^{2\times 2}$ is analytic.
\item[(b)] The jumps for $P^{(\infty)}$ are given by
\begin{align*}
& P^{(\infty)}_{+}(z) = P^{(\infty)}_{-}(z)\begin{pmatrix}
0 & s_{j} \\ -s_{j}^{-1} & 0
\end{pmatrix}, & & z \in (x_{j-1},x_{j}), \, j = 0,\ldots,m+1.
\end{align*}
\item[(c)] As $z \rightarrow \infty$, we have
\begin{equation}
\label{eq:Pinf asympinf s1 neq 0}
P^{(\infty)}(z) = \left( I + \frac{P^{(\infty)}_{1}}{z} + \bigO\left(z^{-2}\right) \right) Ne^{-\frac{\pi i}{4}\sigma_{3}}\left\{ \begin{array}{l l}
I, & \Im z > 0, \\
\begin{pmatrix}
0 & -1 \\ 1 & 0 
\end{pmatrix}, & \Im z < 0.
\end{array}	 \right.
\end{equation}
for a certain matrix $P_{1}^{(\infty)}$ independent of $z$.

As $z \to x_{j}$, $j \in \{0,\ldots,m\}$, we have $P^{(\infty)}(z) = \begin{pmatrix}
\bigO(1) & \bigO(1) \\
\bigO(1) & \bigO(1)
\end{pmatrix}$.
\end{enumerate}
The behavior of $P^{(\infty)}$ near $x_{j}$, $j \in \{0,\ldots,m\}$ is added to ensure uniqueness of the solution to the RH problem for $P^{(\infty)}$. The construction of $P^{(\infty)}$ relies on the following function:
\begin{align}\label{def of D as an integral}
& D(z) = \exp \left( \frac{\theta(z)}{2\pi i} \sum_{j=1}^{m} \log s_{j} \int_{x_{j-1}}^{x_{j}}\frac{du}{u-z} \right), \qquad \mbox{where} \qquad \theta(z) = \begin{cases}
+1, & \mbox{if } \Im z > 0, \\
-1, & \mbox{if } \Im z < 0.
\end{cases}
\end{align}
$D$ satisfies
\begin{align*}
& D_{+}(z)D_{-}(z) = s_{j}, & &  \mbox{for } z \in (x_{j-1},x_{j}), \, j = 0,\ldots,m+1,
\end{align*}
and has the following behavior at $\infty$:
\begin{equation}
 D(z) =  \exp \left( \theta(z) \sum_{\ell = 1}^{k} \frac{d_{\ell}}{z^{\ell}} + \bigO(z^{-k-1}) \right), \quad \mbox{ as } z \to \infty,
\end{equation}
where $k \in \mathbb{N}_{>0}$ is arbitrary and
\begin{align}\label{def of d ell in terms of sj}
& d_{\ell} = -\frac{1}{2\pi i} \sum_{j=1}^{m}  \log s_{j} \int_{x_{j-1}}^{x_{j}} u^{\ell - 1}du = -\frac{1}{2\pi i \ell} \sum_{j=1}^{m}  \log s_{j} \Big( x_{j}^{\ell}-x_{j-1}^{\ell} \Big).
\end{align}
Using the above properties of $D$, we verify that
\begin{equation}\label{def of Pinf s1 neq 0}
P^{(\infty)}(z) = Ne^{-\frac{\pi i}{4}\sigma_{3}}\left\{ \begin{array}{l l}
I, & \Im z > 0 \\
\begin{pmatrix}
0 & -1 \\ 1 & 0 
\end{pmatrix}, & \Im z < 0
\end{array}	 \right\}D(z)^{-\sigma_{3}}
\end{equation}
satisfies the RH problem for $P^{(\infty)}$ with
\begin{equation}\label{Pinf 1 12 s1 neq 0}
P_{1}^{(\infty)} = \begin{pmatrix}
0 & i d_{1} \\
-id_{1} & 0
\end{pmatrix}.
\end{equation}
In preparation for the analysis of Section \ref{Section: integration s1 >0}, we now compute the first terms in the asymptotics of $D(z)$ as $z \to x_{j}$, $j = 0,1,\ldots,m$. It is straightforward to see from \eqref{def of D as an integral} that $D$ can be rewritten as
\begin{equation}\label{lol9}
D(z) = \prod_{j=0}^{m} (z-x_{j})^{\theta(z)\beta_{j}},
\end{equation}
where $\beta_{0},\ldots,\beta_{m}$ are defined by
\begin{equation}\label{def of beta_j s1 neq 0}
\beta_{j} = \frac{1}{2\pi i}\log \frac{s_{j}}{s_{j+1}} \quad \mbox{ or equivalently } \quad e^{-2i\pi \beta_{j}} = \frac{s_{j+1}}{s_{j}}, \qquad j = 0,\ldots,m.
\end{equation}
Note that, since $s_{0} = s_{m+1} = 1$, we have
\begin{equation}
\beta_{0} + \ldots + \beta_{m} = 0.
\end{equation}
As $z \to x_{j}$, $j \in \{0,1,\ldots,m\}$, $\Im z > 0$, we infer from \eqref{lol9} that
\begin{equation}\label{lol6}
D(z) = \sqrt{s_{j+1}} \, (z-x_{j})^{\beta_{j}} \prod_{\substack{k=0 \\ k \neq j}}^{m} |x_{j}-x_{k}|^{\beta_{k}}  (1+\bigO(z-x_{j})).
\end{equation}
Finally, we note that the constants $d_{\ell}$ defined in \eqref{def of d ell in terms of sj} can be rewritten in terms of the $\beta_{j}$'s as follows
\begin{equation}\label{d_ell in terms of beta_j s1 neq 0}
d_{\ell} = -\frac{1}{\ell}\sum_{j=0}^{m} \beta_{j} x_{j}^{\ell}.
\end{equation}

\subsection{Local parametrices}\label{Section: local param s1 neq 0}
Let $\mathcal{D}_{x_{j}}$ be small open disks centered at $x_{j}$, $j=0,1,\ldots,m$ whose radii are equal to $\frac{\delta}{3}$, where $\delta$ is defined in \eqref{assumption delta}. The definition of the radii ensures that the disks do not intersect each other. 

\medskip The local parametrix $P^{(x_{j})}$ is defined in $\mathcal{D}_{x_{j}}$ and satisfies an RH problem with the same jumps as $S$ inside $\mathcal{D}_{x_{j}}$. On the boundary of the disk, $P^{(x_{j})}$ ``matches" with $P^{(\infty)}$ in the sense that
\begin{equation}\label{matching weak s1 neq 0}
P^{(x_{j})}(z) = (I+o(1))P^{(\infty)}(z), \qquad \mbox{ as } r \to +\infty,
\end{equation}
uniformly for $z \in \partial \mathcal{D}_{x_{j}}$. Furthermore, we require that 
\begin{equation}\label{match at the center s1 neq 0}
S(z) P^{(x_{j})}(z)^{-1} = \bigO(1), \qquad \mbox{ as } z \to x_{j}.
\end{equation}
The construction of $P^{(x_{j})}$ is similar for all $j \in \{0,1,\ldots,m\}$ and relies on the confluent hypergeometric functions. This type of local parametrix is well-understood \cite{ItsKrasovsky,FoulquieMartinezSousa,Charlier} and involves the solution $\Phi_{\mathrm{HG}}$ to a model RH problem, which we recall in Subsection \ref{subsection: model RHP with HG functions}. The function
\begin{equation}\label{def conformal map s1 neq 0}
f_{x_{j}}(z) = -2 \left\{ \begin{array}{l l}
g(z)-g_{+}(x_{j}), & \mbox{if } \Im z > 0 \\
-(g(z)-g_{-}(x_{j})), & \mbox{if } \Im z < 0
\end{array} \right. = 2i(z-x_{j})
\end{equation}
is a conformal map from $\mathcal{D}_{x_{j}}$ to a neighborhood of $0$ which maps the real line on the imaginary axis, that is, $f_{x_{j}}(\mathbb{R}\cap \mathcal{D}_{x_{j}})\subset i \mathbb{R}$. Let us deform the lenses in a small neighborhood of $x_{j}$ such that
\begin{equation}\label{deformation of the lenses local param xj s1 neq 0}
f_{x_{j}}((\gamma_{j,+}\cup \gamma_{j+1,+})\cap \mathcal{D}_{x_{j}}) \subset \Gamma_{3} \cup \Gamma_{2}, \qquad f_{x_{j}}((\gamma_{j,-}\cup \gamma_{j+1,-})\cap \mathcal{D}_{x_{j}}) \subset \Gamma_{5} \cup \Gamma_{6},
\end{equation}
where $\Gamma_{3}$, $\Gamma_{2}$, $\Gamma_{5}$ and $\Gamma_{6}$ are as shown in Figure \ref{Fig:HG}. In this way, $f_{x_{j}}$ maps the jump contour for $P^{(x_{j})}$ in a subset of the jump contour for $\Phi_{\mathrm{HG}}$. We seek for $P^{(x_{j})}$ in the form
\begin{equation}\label{lol10}
P^{(x_{j})}(z) = E_{x_{j}}(z) \Phi_{\mathrm{HG}}(rf_{x_{j}}(z);\beta_{j})(s_{j}s_{j+1})^{-\frac{\sigma_{3}}{4}}e^{-rg(z)\sigma_{3}},
\end{equation}
where we recall that $\beta_{j}$ is given by \eqref{def of beta_j s1 neq 0}. If $E_{x_{j}}$ is analytic, then it is straightforward to verify from the jumps for $\Phi_{\mathrm{HG}}$ (given by \eqref{jumps PHG3}) that $P^{(x_{j})}$ satisfies the same jumps as $S$ inside $\mathcal{D}_{x_{j}}$. From the asymptotics \eqref{Asymptotics HG} of $\Phi_{\mathrm{HG}}$, we see that in order to satisfy \eqref{matching weak s1 neq 0}, we are forced to define $E_{x_{j}}$ by
\begin{equation}\label{def of Ej s1 neq 0}
E_{x_{j}}(z) = P^{(\infty)}(z) (s_{j} s_{j+1})^{\frac{\sigma_{3}}{4}} \left\{ \begin{array}{l l}
\ds \sqrt{\frac{s_{j+1}}{s_{j}}}^{\sigma_{3}}, & \Im z > 0 \\
\begin{pmatrix}
0 & 1 \\ -1 & 0
\end{pmatrix}, & \Im z < 0
\end{array} \right\} e^{rg_{+}(x_{j})\sigma_{3}}(rf_{x_{j}}(z))^{\beta_{j}\sigma_{3}}.
\end{equation}
Using the jumps for $P^{(\infty)}$, we verify that $E_{x_{j}}$ has no jump inside $\mathcal{D}_{x_{j}}$. Also, since $P^{(\infty)}(z)$ remains bounded as $z \to x_{j}$, and since $\beta_{j} \in i \mathbb{R}$, it is directly seen from \eqref{def of Ej s1 neq 0} that $E_{x_{j}}(z)$ remains bounded as $z \to x_{j}$. We conclude that $E_{x_{j}}$ is analytic in the whole disk $\mathcal{D}_{x_{j}}$, as desired. Also, since the jumps of $P^{(x_{j})}$ coincide with those of $S$ on $(\mathbb{R}\cup \gamma_{+}\cup \gamma_{-})\cap \mathcal{D}_{x_{j}}$, $S(z)P^{(x_{j})}(z)^{-1}$ is analytic in $\mathcal{D}_{x_{j}} \setminus \{x_{j}\}$. Using \eqref{lol 35} and condition (d) of the RH problem for $S$, we obtain that $S(z)P^{(x_{j})}(z)^{-1} = \bigO(\log(z-x_{j}))$, which means that $x_{j}$ is a removable singularity and in particular \eqref{match at the center s1 neq 0} holds. In Section \ref{Section: integration s1 >0}, we will need more precise information about the matching condition \eqref{matching weak s1 neq 0}. Using \eqref{Asymptotics HG}, we obtain
\begin{equation}\label{matching strong -x_j s1 neq 0}
P^{(x_{j})}(z)P^{(\infty)}(z)^{-1} = I + \frac{1}{rf_{x_{j}}(z)}E_{x_{j}}(z) \Phi_{\mathrm{HG},1}(\beta_{j})E_{x_{j}}(z)^{-1} + \bigO(r^{-2}),
\end{equation}
as $r \to + \infty$, uniformly for $z \in \partial \mathcal{D}_{x_{j}}$, where $\Phi_{\mathrm{HG},1}(\beta_{j})$ is given by \eqref{def of tau}. Also, using \eqref{def of Pinf s1 neq 0}, \eqref{lol6} and \eqref{def conformal map s1 neq 0} we obtain
\begin{equation}\label{E_j at x_j s1 neq 0}
E_{x_{j}}(x_{j}) = N\Lambda_{j}^{\sigma_{3}}, \quad \mbox{ where } \quad \Lambda_{j} = e^{-\frac{\pi i}{4}} e^{r g_{+}(x_{j})}(2r)^{\beta_{j}}  \prod_{\substack{k=0 \\ k \neq j}}^{m} |x_{j}-x_{k}|^{-\beta_{k}}.
\end{equation}

\subsection{Small norm problem}\label{subsection Small norm s1 neq 0}

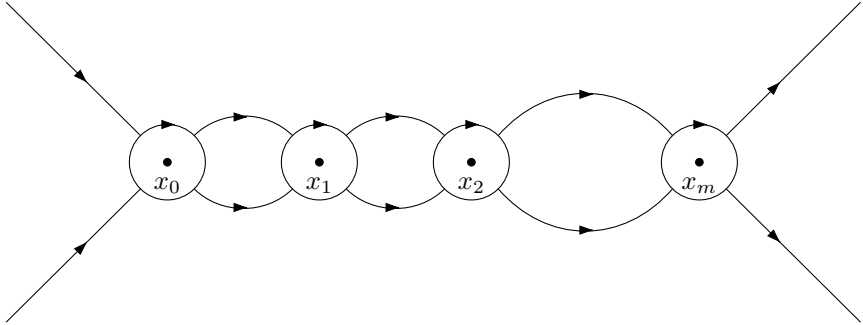
\begin{figure}
\centering
\begin{tikzpicture}
\node at (3,0) {};
\draw ($(3,0)+(135:0.5)$) -- ($(3,0)+(135:3)$);
\draw ($(3,0)+(-135:0.5)$) -- ($(3,0)+(-135:3)$);
\draw ($(10,0)+(45:0.5)$) -- ($(10,0)+(45:3)$);
\draw ($(10,0)+(-45:0.5)$) -- ($(10,0)+(-45:3)$);

\draw[fill] (3,0) circle (0.05);
\draw (3,0) circle (0.5);
\draw[fill] (5,0) circle (0.05);
\draw (5,0) circle (0.5);
\draw[fill] (7,0) circle (0.05);
\draw (7,0) circle (0.5);
\draw[fill] (10,0) circle (0.05);
\draw (10,0) circle (0.5);

\node at (3,-0.3) {$x_{0}$};
\node at (5,-0.3) {$x_{1}$};
\node at (7,-0.3) {$x_{2}$};
\node at (10,-0.3) {$x_{m}$};

\draw[black,arrows={-Triangle[length=0.18cm,width=0.12cm]}]
($(3,0)+(135:1.5)$) --  ++(-45:0.001);
\draw[black,arrows={-Triangle[length=0.18cm,width=0.12cm]}]
($(3,0)+(-135:1.5)$) --  ++(45:0.001);

\draw[black,arrows={-Triangle[length=0.18cm,width=0.12cm]}]
($(10,0)+(45:1.5)$) --  ++(45:0.001);
\draw[black,arrows={-Triangle[length=0.18cm,width=0.12cm]}]
($(10,0)+(-45:1.5)$) --  ++(-45:0.001);

\draw ($(3,0)+(45:0.5)$) .. controls (3.7,0.7) and (4.3,0.7) .. ($(5,0)+(135:0.5)$);
\draw ($(3,0)+(-45:0.5)$) .. controls (3.7,-0.7) and (4.3,-0.7) .. ($(5,0)+(-135:0.5)$);
\draw ($(5,0)+(45:0.5)$) .. controls (5.7,0.7) and (6.3,0.7) .. ($(7,0)+(135:0.5)$);
\draw ($(5,0)+(-45:0.5)$) .. controls (5.7,-0.7) and (6.3,-0.7) .. ($(7,0)+(-135:0.5)$);
\draw ($(7,0)+(45:0.5)$) .. controls (8,1.1) and (9,1.1) .. ($(10,0)+(135:0.5)$);
\draw ($(7,0)+(-45:0.5)$) .. controls (8,-1.1) and (9,-1.1) .. ($(10,0)+(-135:0.5)$);

\draw[black,arrows={-Triangle[length=0.18cm,width=0.12cm]}]
(4.05,0.6) --  ++(0:0.001);
\draw[black,arrows={-Triangle[length=0.18cm,width=0.12cm]}]
(4.05,-0.6) --  ++(0:0.001);
\draw[black,arrows={-Triangle[length=0.18cm,width=0.12cm]}]
(6.05,0.6) --  ++(0:0.001);
\draw[black,arrows={-Triangle[length=0.18cm,width=0.12cm]}]
(6.05,-0.6) --  ++(0:0.001);
\draw[black,arrows={-Triangle[length=0.18cm,width=0.12cm]}]
(8.6,0.9) --  ++(0:0.001);
\draw[black,arrows={-Triangle[length=0.18cm,width=0.12cm]}]
(8.6,-0.9) --  ++(0:0.001);

\draw[black,arrows={-Triangle[length=0.18cm,width=0.12cm]}]
(3.1,0.5) --  ++(0:0.001);
\draw[black,arrows={-Triangle[length=0.18cm,width=0.12cm]}]
(5.1,0.5) --  ++(0:0.001);
\draw[black,arrows={-Triangle[length=0.18cm,width=0.12cm]}]
(7.1,0.5) --  ++(0:0.001);
\draw[black,arrows={-Triangle[length=0.18cm,width=0.12cm]}]
(10.1,0.5) --  ++(0:0.001);
\end{tikzpicture}
\caption{\label{fig:contour for R s1 neq 0}Jump contours $\Sigma_{R}$ for the RH problem for $R$ with $m=3$.}
\end{figure}
In this subsection, we show that $P^{(\infty)}$ and $P^{(x_{j})}$ approximate $S$ as $r \to + \infty$. For this, we define
\begin{equation}\label{def of R s1 neq 0}
R(z) = \left\{ \begin{array}{l l}
S(z)P^{(\infty)}(z)^{-1}, & \mbox{for } z \in \mathbb{C}\setminus \bigcup_{j=0}^{m}\mathcal{D}_{x_{j}}, \\
S(z)P^{(x_{j})}(z)^{-1}, & \mbox{for } z \in \mathcal{D}_{x_{j}}, \, j \in \{0,1,\ldots,m\}.
\end{array} \right.
\end{equation}
It follows from the analysis of Subsection \ref{Section: local param s1 neq 0} that $R$ is analytic inside the $m+1$ disks. Since the jumps of $P^{(\infty)}$ and of $S$ are the same on $(x_{j-1},x_{j})$, $j=1,\ldots,m$, we conclude that $R$ is analytic on $\mathbb{C}\setminus \Sigma_{R}$, where 
\begin{align*}
\Sigma_{R} = \bigcup_{j=0}^{m} \partial \mathcal{D}_{x_{j}} \cup \bigg( (\gamma_{+}\cup  \gamma_{-}) \setminus \bigcup_{j=0}^{m} \mathcal{D}_{x_{j}} \bigg),
\end{align*}
see Figure \ref{fig:contour for R s1 neq 0}. Also, from \eqref{estimates for jumps for S exp small} and \eqref{def of Pinf s1 neq 0}, we infer that the jumps $J_{R} := R_{-}^{-1}R_{+}$ satisfy
\begin{equation}\label{estimates jumps for R lenses}
J_{R}(z) = P^{(\infty)}(z)S_{-}(z)^{-1}S_{+}(z)P^{(\infty)}(z)^{-1} = I + \bigO(e^{-c |z|  r}), \qquad \mbox{as } r \to + \infty,
\end{equation}
uniformly for $z \in \Sigma_{R} \cap (\gamma_{+} \cup  \gamma_{-})$, for a certain $c>0$ independent of $z$ and $r$. As shown in Figure \ref{fig:contour for R s1 neq 0}, we orient the boundaries of the disks in the clockwise direction. It follows from \eqref{matching strong -x_j s1 neq 0} that
\begin{equation}\label{estimates jumps for R disks}
J_{R}(z) = P^{(x_{j})}(z)P^{(\infty)}(z)^{-1} = I + \bigO \Big(\frac{1}{r} \Big), \qquad \mbox{ as } r \to  +\infty
\end{equation}
uniformly $z \in \bigcup_{j=0}^{m} \partial \mathcal{D}_{x_{j}}$. Furthermore, from the behavior of $S(z)$ and $P^{(\infty)}(z)$ as $z \to \infty$ given by \eqref{eq:Sasympinf s1 neq 0} and \eqref{eq:Pinf asympinf s1 neq 0}, we have
\begin{equation}\label{R at inf s1 neq 0}
R(z) = I + \bigO(z^{-1}), \qquad \mbox{as } z \to \infty.
\end{equation}
By standard theory for RH problems \cite{DeiftZhou1992,DKMVZ2,DKMVZ1}, it follows that $R$ exists for sufficiently large $r$ and satisfies 
\begin{align}
& R(z) = I + \frac{R^{(1)}(z)}{r} + \bigO(r^{-2}), \qquad R^{(1)}(z) = \bigO(1), & & \mbox{ as } r \to  +\infty \label{eq: asymp R inf s1 neq 0} 
\end{align}
uniformly for $z \in \mathbb{C}\setminus \Sigma_{R}$. Note that the presence of $r^{\pm \beta_{j}}$ in the entries of $E_{x_{j}}$ (see \eqref{def of Ej s1 neq 0}) implies, by \eqref{matching strong -x_j s1 neq 0}, that the entries of $J_{R}$ contain factors of the form $r^{\pm 2\beta_{j}}$. Thus, a standard analysis of the Cauchy operator associated to $R$ (see e.g. \cite{ItsKrasovsky, ChCl3} for similar situations) shows that
\begin{equation}\label{eq: asymp der beta R inf s1 neq 0}
\partial_{\beta_{j}}R(z) = \frac{\partial_{\beta_{j}}R^{(1)}(z)}{r} + \bigO \Big( \frac{\log r}{r^{2}} \Big), \qquad \partial_{\beta_{j}}R^{(1)}(z) = \bigO(\log r), \qquad \mbox{ as } r \to  +\infty.
\end{equation}
Moreover, since the asymptotics \eqref{estimates jumps for R lenses} and \eqref{estimates jumps for R disks} are uniform in $x_{0},\ldots,x_{m}$ in compact subsets of $\mathbb{R}$ such that \eqref{assumption delta} holds, and uniform for $\beta_{1},\ldots,\beta_{m}$ in compact subsets of $i \mathbb{R}$, the asymptotics \eqref{eq: asymp R inf s1 neq 0} and \eqref{eq: asymp der beta R inf s1 neq 0} also hold uniformly in $x_{0},\ldots,x_{m},\beta_{1},\ldots,\beta_{m}$ in the same way.

\vspace{0.2cm}\hspace{-0.55cm}Now, we compute explicitly $R^{(1)}(z)$ for $z \in \mathbb{C}\setminus \bigcup_{j=0}^{m}\mathcal{D}_{x_{j}}$. For this, we first note that
\begin{equation}\label{integral rep for R loool}
R(z) = I + \frac{1}{2\pi i} \int_{\Sigma_{R}} \frac{R_{-}(s)(J_{R}(s)-I)}{s-z}ds. 
\end{equation}
The above formula directly follows from $R_{+}=R_{-}J_{R}$ and the asymptotics \eqref{R at inf s1 neq 0}. Also, we know from \eqref{matching strong -x_j s1 neq 0} that for
\begin{equation}\label{asymp for JR loool}
J_{R}(z) = I + \frac{J_{R}^{(1)}(z)}{r} + \bigO(r^{-2}), \qquad J_{R}^{(1)}(z) = \frac{1}{f_{x_{j}}(z)}E_{x_{j}}(z) \Phi_{\mathrm{HG},1}(\beta_{j})E_{x_{j}}(z)^{-1},
\end{equation}
as $r \to \infty$ uniformly for $z \in \mathcal{D}_{x_{j}}$, $j=0,\ldots,m$. By substituting \eqref{asymp for JR loool} in \eqref{integral rep for R loool}, we obtain 
\begin{equation}\label{integral form for R1 sp neq 0}
R^{(1)}(z) = \frac{1}{2\pi i}\int_{\bigcup_{j=0}^{m}\partial\mathcal{D}_{x_{j}}} \frac{J_{R}^{(1)}(s)}{s-z}ds.
\end{equation}
Note that the expression \eqref{asymp for JR loool} for $J_{R}^{(1)}$ can be analytically continued from $\partial \mathcal{D}_{x_{j}}$ to $\overline{\mathcal{D}_{x_{j}}}\setminus \{x_{j}\}$, and that $J_{R}^{(1)}$ has a simple pole at $x_{j}$. Recalling that the disks are oriented in the clockwise direction, and using \eqref{def conformal map s1 neq 0}, \eqref{matching strong -x_j s1 neq 0}-\eqref{E_j at x_j s1 neq 0} and \eqref{def of tau}, we obtain
\begin{equation}\label{expression for R^1 s1 neq 0}
R^{(1)}(z) = \sum_{j=0}^{m} \frac{1}{z-x_{j}}\mbox{Res}(J_{R}^{(1)}(s),s = x_{j}), \qquad \mbox{ for } z \in \mathbb{C}\setminus \bigcup_{j=0}^{m}\mathcal{D}_{x_{j}},
\end{equation}
where for $j \in \{0,\ldots,m\}$ we have
\begin{equation*}
\mbox{Res}\left( J_{R}^{(1)}(s),s=x_{j} \right) = \frac{\beta_{j}^{2}}{2i} N \begin{pmatrix}
-1 & \widetilde{\Lambda}_{j,1} \\ -\widetilde{\Lambda}_{j,2} & 1
\end{pmatrix} N^{-1} = \frac{\beta_{j}^{2}}{4}\begin{pmatrix}
-\widetilde{\Lambda}_{j,1}-\widetilde{\Lambda}_{j,2} & -i(\widetilde{\Lambda}_{j,1}-\widetilde{\Lambda}_{j,2}+2i) \\
-i(\widetilde{\Lambda}_{j,1}-\widetilde{\Lambda}_{j,2}-2i) & \widetilde{\Lambda}_{j,1}+\widetilde{\Lambda}_{j,2}
\end{pmatrix},
\end{equation*}
with
\begin{equation}
\widetilde{\Lambda}_{j,1} = \tau(\beta_{j})\Lambda_{j}^{2} \qquad \mbox{ and } \qquad \widetilde{\Lambda}_{j,2} = \tau(-\beta_{j})\Lambda_{j}^{-2}.
\end{equation}
\section{Proof of Theorem \ref{thm:s1 neq 0}}\label{Section: integration s1 >0}
We prove Theorem \ref{thm:s1 neq 0} in two steps. First, we use the RH analysis of Section \ref{Section: Steepest descent with s1>0} to obtain large $r$ asymptotics for the quantities $K_{\infty}$ and $K_{x_{j}}$ defined in \eqref{K inf}-\eqref{K xj}. By substituting these asymptotics in the differential identity
\begin{equation}
\partial_{s_{k}} \log F(r\vec{x},\vec{s})= K_{\infty} + \sum_{j=0}^{m}K_{x_{j}},  \qquad k \in \{1,\ldots,m\},
\end{equation}
we obtain large $r$ asymptotics for $\partial_{s_{k}} \log F(r\vec{x},\vec{s})$. Second, we integrate these asymptotics over the parameters $s_{1},\ldots,s_{m}$ to obtain large $r$ asymptotics for $\log F(r\vec{x},\vec{s})$.

\subsection[lol]{Large $r$ asymptotics for $\partial_{s_{k}} \log F(r\vec{x},\vec{s})$}

\paragraph{Asymptotics for $K_{\infty}$.} By \eqref{def of R s1 neq 0}, \eqref{eq:Sasympinf s1 neq 0} and \eqref{eq:Pinf asympinf s1 neq 0}, we have
\begin{equation}
R(z) = S(z)P^{(\infty)}(z)^{-1} = I + \frac{R_{1}}{z} + \bigO(z^{-2}), \qquad \mbox{ as } z \to \infty,
\end{equation}
for a certain matrix $R_{1}$ satisfying $T_{1} = R_{1} + P_{1}^{(\infty)}$. Hence, by \eqref{eq: asymp R inf s1 neq 0},
\begin{align*}
T_{1} = P_{1}^{(\infty)} + \frac{R_{1}^{(1)}}{r} + \bigO(r^{-2}), \qquad \mbox{ as } r \to + \infty,
\end{align*}
where $R_{1}^{(1)}$ is the $z^{-1}$ coefficient in the large $z$ expansion of $R^{(1)}(z)$. Using \eqref{Pinf 1 12 s1 neq 0} and \eqref{expression for R^1 s1 neq 0}, we find the following large $r$ asymptotics for $T_{1}$:
\begin{equation}\label{asymp for T1 block matrix}
T_{1} = \begin{pmatrix}
0 & id_{1} \\
-id_{1} & 0
\end{pmatrix} + \sum_{j=0}^{m} \frac{\beta_{j}^{2}}{4r}\begin{pmatrix}
-\widetilde{\Lambda}_{j,1}-\widetilde{\Lambda}_{j,2} & -i(\widetilde{\Lambda}_{j,1}-\widetilde{\Lambda}_{j,2}+2i) \\
-i(\widetilde{\Lambda}_{j,1}-\widetilde{\Lambda}_{j,2}-2i) & \widetilde{\Lambda}_{j,1}+\widetilde{\Lambda}_{j,2}
\end{pmatrix} + \bigO(r^{-2}).
\end{equation}
By \eqref{K inf}, \eqref{def of T1 s1 neq 0}, \eqref{eq: asymp der beta R inf s1 neq 0} and \eqref{asymp for T1 block matrix}, we obtain
\begin{equation}\label{K inf asymp s1 neq 0}
K_{\infty} = r \big( \partial_{s_{k}} T_{1,21} - \partial_{s_{k}}T_{1,12} \big) = -2i \partial_{s_{k}}d_{1}r - \sum_{j=0}^{m} \partial_{s_{k}}(\beta_{j}^{2}) + \bigO \Big( \frac{\log r}{r} \Big), \qquad \mbox{ as } r \to + \infty.
\end{equation}
\paragraph{Asymptotics for $K_{x_{j}}$ with $j \in \{0,\ldots,m\}$.} For $z$ outside the lenses and inside $\mathcal{D}_{x_{j}}$, by \eqref{def of S s1 neq 0}, \eqref{def of R s1 neq 0} and \eqref{lol10}, we have
\begin{equation}\label{lol11}
T(z) = R(z)E_{x_{j}}(z)\Phi_{\mathrm{HG}}(rf_{x_{j}}(z);\beta_{j})(s_{j}s_{j+1})^{-\frac{\sigma_{3}}{4}}e^{- r g(z)\sigma_{3}},
\end{equation}
and by \eqref{def conformal map s1 neq 0} and \eqref{model RHP HG in different sector}, we also have
\begin{equation}
\Phi_{\mathrm{HG}}(rf_{x_{j}}(z);\beta_{j}) = \widehat{\Phi}_{\mathrm{HG}}(rf_{x_{j}}(z);\beta_{j}), \qquad \mbox{for } \Im z > 0.
\end{equation}
Using \eqref{def of beta_j s1 neq 0} and Euler's reflection formula for the $\Gamma$-function (see e.g. \cite[equation 5.5.3]{NIST}), we verify that
\begin{equation}\label{relation Gamma beta_j and s_j s1 neq 0}
\frac{\sin (\pi \beta_{j})}{\pi} = \frac{1}{\Gamma(\beta_{j})\Gamma(1-\beta_{j})} = -\frac{s_{j+1}-s_{j}}{2\pi i \sqrt{s_{j}s_{j+1}}}.
\end{equation}
This relation, combined with \eqref{def conformal map s1 neq 0} and \eqref{precise asymptotics of Phi HG near 0}, allows to verify that
\begin{equation}\label{lol 2 s1 neq 0}
\Phi_{\mathrm{HG}}(rf_{x_{j}}(z);\beta_{j})(s_{j}s_{j+1})^{-\frac{\sigma_{3}}{4}} = \begin{pmatrix}
\Psi_{j,11} & \Psi_{j,12} \\
\Psi_{j,21} & \Psi_{j,22}
\end{pmatrix} (I + \bigO(z-x_{j})) \begin{pmatrix}
1 & -\frac{s_{j+1}-s_{j}}{2\pi i}\log(r(z - x_{j})) \\
0 & 1
\end{pmatrix} ,
\end{equation}
as $z \to x_{j}$ from $\Im z >0$ and $z$ outside the lenses, where the principal branch is taken for the log and
\begin{align}
& \Psi_{j,11} = \frac{\Gamma(1-\beta_{j})}{(s_{j}s_{j+1})^{\frac{1}{4}}}, \qquad \Psi_{j,12} = \frac{(s_{j}s_{j+1})^{\frac{1}{4}}}{\Gamma(\beta_{j})} \left( \log 2 - \frac{i\pi}{2} + \frac{\Gamma^{\prime}(1-\beta_{j})}{\Gamma(1-\beta_{j})}+2\gamma_{\mathrm{E}} \right), \nonumber \\
& \Psi_{j,21} = \frac{\Gamma(1+\beta_{j})}{(s_{j}s_{j+1})^{\frac{1}{4}}}, \qquad \Psi_{j,22} = \frac{-(s_{j}s_{j+1})^{\frac{1}{4}}}{\Gamma(-\beta_{j})} \left( \log 2 - \frac{i\pi}{2} + \frac{\Gamma^{\prime}(-\beta_{j})}{\Gamma(-\beta_{j})} + 2\gamma_{\mathrm{E}} \right). \label{Psi j entries s1 neq 0}
\end{align}
In particular, we note that
\begin{equation}\label{Psi_j first column connection formula s1 neq 0}
\Psi_{j,11}\Psi_{j,21} = -\beta_{j} \frac{2\pi i}{s_{j+1}-s_{j}}, \qquad j =0,\ldots,m,
\end{equation}
where we have used the well-known formula $\Gamma(1+z)=z\Gamma(z)$ and \eqref{relation Gamma beta_j and s_j s1 neq 0}. We deduce from \eqref{Asymp of Phi near x_j}, \eqref{def of T s1 neq 0}, \eqref{lol11} and \eqref{lol 2 s1 neq 0} that
\begin{equation}\label{formula for Gj lol}
G_{j}(rx_{j};r\vec{x},\vec{s}) = R(x_{j})E_{x_{j}}(x_{j})\begin{pmatrix}
\Psi_{j,11} & \Psi_{j,12} \\ \Psi_{j,21} & \Psi_{j,22}
\end{pmatrix}.
\end{equation}
Also, from \eqref{E_j at x_j s1 neq 0}, we have
\begin{align}
& \partial_{s_{k}}E_{x_{j},11}(x_{j}) = E_{x_{j},11}(x_{j}) \partial_{s_{k}} \log \Lambda_{j}, \qquad  \partial_{s_{k}}E_{x_{j},12}(x_{j}) = -E_{x_{j},12}(x_{j}) \partial_{s_{k}} \log \Lambda_{j}, \nonumber \\[0.2cm]
& \partial_{s_{k}}E_{x_{j},21}(x_{j}) = E_{x_{j},21}(x_{j}) \partial_{s_{k}} \log \Lambda_{j}, \qquad  \partial_{s_{k}}E_{x_{j},22}(x_{j}) = -E_{x_{j},22}(x_{j}) \partial_{s_{k}} \log \Lambda_{j}. \nonumber 
\end{align}
Substituting \eqref{formula for Gj lol} in the formula for $K_{x_{j}}$ given by \eqref{K xj}, and using \eqref{eq: asymp R inf s1 neq 0}, \eqref{eq: asymp der beta R inf s1 neq 0}, the above relations and the fact that $\det E_{x_{j}}(x_{j})=1$, we obtain the following large $r$ asymptotics
\begin{multline}\label{K xj part 1 asymp s1 neq 0}
\sum_{j=0}^{m} K_{x_{j}} = \sum_{j=0}^{m} -\frac{s_{j+1}-s_j}{2\pi i} \Big( \Psi_{j,11} \partial_{s_{k}}\Psi_{j,21} - \Psi_{j,21}\partial_{s_{k}}\Psi_{j,11} \Big) - \sum_{j=0}^{m} 2\beta_{j} \partial_{s_{k}} \log \Lambda_{j} + \bigO \Big( \frac{\log r}{r} \Big).
\end{multline}
\paragraph{Asymptotics for $\partial_{s_k}\log F(r\vec{x},\vec{s})$.} By summing the large $r$ asymptotics of $K_{x_{j}}$, $j=0,\ldots,m$ and $K_{\infty}$ using \eqref{K inf asymp s1 neq 0} and \eqref{K xj part 1 asymp s1 neq 0}, we obtain
\begin{multline}\label{lol5 s1 neq 0}
\partial_{s_k}\log F(r\vec{x},\vec{s}) = -2i\partial_{s_{k}}d_{1} \, r - \sum_{j=0}^{m} \Big( 2\beta_{j} \partial_{s_{k}} \log \Lambda_{j} + \partial_{s_{k}} (\beta_{j}^{2}) \Big) \\ + \sum_{j=0}^{m}  -\frac{s_{j+1}-s_{j}}{2\pi i}\big( \Psi_{j,11}\partial_{s_{k}}\Psi_{j,21}-\Psi_{j,21}\partial_{s_{k}}\Psi_{j,11} \big) + \bigO\Big( \frac{\log r}{r} \Big), \qquad \mbox{as } r \to + \infty.
\end{multline}
This last sum can be simplified using the expressions \eqref{Psi j entries s1 neq 0} for $\Psi_{j,11}$ and $\Psi_{j,21}$ together with \eqref{Psi_j first column connection formula s1 neq 0}:
\begin{equation}\label{lol3 s1 neq 0}
\sum_{j=0}^{m} -\frac{s_{j+1}-s_{j}}{2\pi i}\big( \Psi_{j,11}\partial_{s_{k}}\Psi_{j,21}-\Psi_{j,21}\partial_{s_{k}}\Psi_{j,11} \big) = \sum_{j=0}^{m}\beta_{j} \partial_{s_{k}} \log \frac{\Gamma(1+\beta_{j})}{\Gamma(1-\beta_{j})}.
\end{equation}
Also, using \eqref{E_j at x_j s1 neq 0}, we have
\begin{equation}\label{lol4 s1 neq 0}
\sum_{j=0}^{m} - 2\beta_{j} \partial_{s_{k}} \log \Lambda_{j} = -2 \sum_{j=0}^{m} \beta_{j} \partial_{s_{k}}(\beta_{j}) \log (2r) -2\sum_{j=0}^{m} \beta_{j} \sum_{\substack{\ell = 0 \\ \ell \neq j}}^{m} \partial_{s_{k}}(\beta_{\ell})\log |x_{j}-x_{\ell}|^{-1}.
\end{equation}
For convenience, we will integrate with respect to $\beta_{1},\ldots,\beta_{m}$ rather than in the variables $s_{1},\ldots,s_{m}$ (recall that $\beta_{0}=-\beta_{1}-\ldots - \beta_{m}$). Let us define
\begin{equation}\label{def of F tilde}
\widetilde{F}(r \vec{x}, \vec{\beta}) = F(r \vec{x},\vec{s}),
\end{equation}
where $\vec{\beta} = (\beta_{1},\ldots,\beta_{m})$ and $\vec{s} = (s_{1},\ldots,s_{m})$ are related via \eqref{def of beta_j s1 neq 0}. Substituting \eqref{lol3 s1 neq 0} and \eqref{lol4 s1 neq 0} into \eqref{lol5 s1 neq 0}, and taking the derivative with respect to $\beta_{k}$ instead of $s_{k}$, $k \in \{1,\ldots,m\}$, we obtain
\begin{multline}\label{lol5 part 2 s1 neq 0}
\partial_{\beta_k}\log\widetilde{F}(r \vec{x}, \vec{\beta}) = -2 i\partial_{\beta_{k}}d_{1} \, r -2 \sum_{j=0}^{m} \beta_{j} \partial_{\beta_{k}}(\beta_{j}) \log (2r) -2\sum_{j=0}^{m} \beta_{j} \sum_{\substack{\ell = 0 \\ \ell \neq j}}^{m} \partial_{\beta_{k}}(\beta_{\ell})\log  |x_{j}-x_{\ell}|^{-1} \\  - \sum_{j=0}^{m} \partial_{\beta_{k}}(\beta_{j}^{2}) + \sum_{j=0}^{m} \beta_{j} \partial_{\beta_{k}} \log \frac{\Gamma(1+\beta_{j})}{\Gamma(1-\beta_{j})} + \bigO\Big(\frac{\log r}{r}\Big), \qquad \mbox{as } r \to + \infty.
\end{multline}
Using $\beta_{0}=-\beta_{1}-\ldots - \beta_{m}$ and the explicit expression \eqref{d_ell in terms of beta_j s1 neq 0} of $d_{1}$, we can simplify the different terms that appear on the right-hand side of \eqref{lol5 part 2 s1 neq 0}:
\begin{align*}
& -2i \partial_{\beta_{k}}d_{1}r = 2i(x_{k}-x_{0})r, \\
& -2 \sum_{j=0}^{m} \beta_{j} \partial_{\beta_{k}}(\beta_{j}) \log (2r) = -2 \beta_{k} \log(2r) - 2(\beta_{1}+\ldots+\beta_{m})\log(2r), \\[-0.05cm]
& -2\sum_{j=0}^{m} \beta_{j} \sum_{\substack{\ell = 0 \\ \ell \neq j}}^{m} \partial_{\beta_{k}}(\beta_{\ell})\log  |x_{j}-x_{\ell}|^{-1} = -2 \sum_{\substack{j=1 \\ j \neq k}}^{m} \beta_{j} \log \left| \frac{(x_{j}-x_{0})(x_{k}-x_{0})}{x_{j}-x_{k}}\right| - 4 \beta_{k} \log |x_{k}-x_{0}|, \\[-0.2cm]
& - \sum_{j=0}^{m} \partial_{\beta_{k}}(\beta_{j}^{2}) = -2\beta_{k} - 2 \sum_{j=1}^{m} \beta_{j}, \\
& \sum_{j=0}^{m} \beta_{j} \partial_{\beta_{k}} \log \frac{\Gamma(1+\beta_{j})}{\Gamma(1-\beta_{j})} = \beta_{k} \partial_{\beta_{k}} \log \frac{\Gamma(1+\beta_{k})}{\Gamma(1-\beta_{k})} + (\beta_{1}+\ldots+\beta_{m})\partial_{\beta_{k}}\log \frac{\Gamma(1+\beta_{1}+\ldots+\beta_{m})}{\Gamma(1-\beta_{1}-\ldots-\beta_{m})}.
\end{align*}
These identities allow to rewrite \eqref{lol5 part 2 s1 neq 0} as follows
\begin{multline}\label{DIFF IDENTITY s1 neq 0}
\partial_{\beta_k}\log \widetilde{F}(r \vec{x}, \vec{\beta}) = 2i (x_{k}-x_{0})r - 4\beta_{k} \log \big( 2r(x_{k}-x_{0}) \big) - 2 \sum_{\substack{j=1 \\ j \neq k}}^{m}\beta_{j}\log\left( \frac{2r(x_{j}-x_{0})(x_{k}-x_{0})}{|x_{j}-x_{k}|}\right) \\
\hspace{0.55cm}-2\beta_{k} - 2 \sum_{j=1}^{m} \beta_{j} + \beta_{k} \partial_{\beta_{k}} \log \frac{\Gamma(1+\beta_{k})}{\Gamma(1-\beta_{k})} + (\beta_{1}+\ldots+\beta_{m})\partial_{\beta_{k}}\log \frac{\Gamma(1+\beta_{1}+\ldots+\beta_{m})}{\Gamma(1-\beta_{1}-\ldots-\beta_{m})} + \bigO \Big( \frac{\log r}{r} \Big), \\
\mbox{as } r \to + \infty.
\end{multline}
The analysis of Subsection \ref{subsection Small norm s1 neq 0} implies that the error term in \eqref{DIFF IDENTITY s1 neq 0} is uniform in $x_{0},\ldots,x_{m}$ in compact subsets of $\mathbb{R}$ such that \eqref{assumption delta} holds, and uniform in $\beta_{1},\ldots,\beta_{m}$ in compact subsets of $i \mathbb{R}$. 
\subsection{Integration of the differential identity}
We first find an explicit formula for
\begin{equation}\label{def of I ell}
I_{\ell}(\beta_{\ell};\beta_{1},\ldots,\beta_{\ell-1}) = \int_{0}^{\beta_{\ell}}(\beta_{1}+\ldots+\beta_{\ell-1}+x)\partial_{x} \log \frac{\Gamma(1+\beta_{1}+\ldots+\beta_{\ell-1}+x)}{\Gamma(1-\beta_{1}-\ldots-\beta_{\ell-1}-x)}dx,
\end{equation}
with $\ell \in \{1,\ldots,m\}$. Integrating \eqref{def of I ell} by parts, we obtain
\begin{multline}\label{first expr for I ell}
I_{\ell}(\beta_{\ell};\beta_{1},...,\beta_{\ell-1}) = (\beta_{1}+...+\beta_{\ell})\log \frac{\Gamma(1+\beta_{1}+...+\beta_{\ell})}{\Gamma(1-\beta_{1}-...-\beta_{\ell})} - (\beta_{1}+...+\beta_{\ell-1})\log \frac{\Gamma(1+\beta_{1}+...+\beta_{\ell-1})}{\Gamma(1-\beta_{1}-...-\beta_{\ell-1})} \\
- \int_{0}^{\beta_{\ell}} \log \Gamma(1+\beta_{1}+\ldots+\beta_{\ell-1}+x)dx + \int_{0}^{\beta_{\ell}} \log \Gamma(1-\beta_{1}-\ldots-\beta_{\ell-1}-x)dx.
\end{multline}
We recall the following integral relation for the $\Gamma$ function (see e.g. \cite[formula 5.17.4]{NIST}):
\begin{equation}\label{integral of Gamma function}
\int_{0}^{z}\log \Gamma (1+x) dx = \frac{z}{2} \log 2\pi - \frac{z(z+1)}{2} + z \log \Gamma(z+1) - \log G(z+1),
\end{equation}
where $G$ is Barnes' $G$-function. Using twice \eqref{integral of Gamma function} with suitable changes of variables, we obtain
\begin{multline*}
\int_{0}^{\beta_{\ell}} \log \frac{\Gamma(1-\beta_{1}-...-\beta_{\ell-1}-x)}{\Gamma(1+\beta_{1}+...+\beta_{\ell-1}+x)}dx = \beta_{\ell}^{2} + 2 \beta_{\ell}(\beta_{1}+...+\beta_{\ell-1}) + (\beta_{1}+...+\beta_{\ell})\log \frac{\Gamma(1-\beta_{1}-...-\beta_{\ell})}{\Gamma(1+\beta_{1}+...+\beta_{\ell})} \\ 
- (\beta_{1}+...+\beta_{\ell-1})\log \frac{\Gamma(1-\beta_{1}-...-\beta_{\ell-1})}{\Gamma(1+\beta_{1}+...+\beta_{\ell-1})} + \log \frac{G(1+\beta_{1}+...+\beta_{\ell})G(1-\beta_{1}-...-\beta_{\ell})}{G(1+\beta_{1}+...+\beta_{\ell-1})G(1-\beta_{1}-...-\beta_{\ell-1})}.
\end{multline*}
Substituting this identity in \eqref{first expr for I ell} and simplifying, we arrive at
\begin{multline}\label{final expression for I ell}
I_{\ell}(\beta_{\ell};\beta_{1},...,\beta_{\ell-1}) = \beta_{\ell}^{2} + 2 \beta_{\ell}(\beta_{1}+...+\beta_{\ell-1}) + \log \frac{G(1+\beta_{1}+...+\beta_{\ell})G(1-\beta_{1}-...-\beta_{\ell})}{G(1+\beta_{1}+...+\beta_{\ell-1})G(1-\beta_{1}-...-\beta_{\ell-1})}.
\end{multline}
Now, we will use the identity \eqref{DIFF IDENTITY s1 neq 0} for $k=1,\ldots,m$. We start with $k = 1$ and $\beta_{2} = 0 = \beta_{3} = \ldots = \beta_{m}$. With the notation $\vec{\beta}_{1}=(\beta_{1},0,\ldots,0)$, \eqref{DIFF IDENTITY s1 neq 0} becomes
\begin{multline*}
\partial_{\beta_1}\log \widetilde{F}(r \vec{x}, \vec{\beta}_{1}) = 2i (x_{1}-x_{0})r - 4 \beta_{1} \log \big( 2r |x_{1}-x_{0}|\big) -4\beta_{1} + 2\beta_{1} \partial_{\beta_{1}} \log \frac{\Gamma(1+\beta_{1})}{\Gamma(1-\beta_{1})} + \bigO \Big( \frac{\log r}{r} \Big),
\end{multline*}
as $r \to + \infty$. Since the above asymptotics are uniform for $\beta_{1}$ in compact subsets of $i \mathbb{R}$, we can integrate over $\beta_{1}$ from $\beta_{1} = 0$ to an arbitrary $\beta_{1} \in i \mathbb{R}$ without changing the order of the error term. Using formula \eqref{final expression for I ell} with $\ell = 1$, we obtain
\begin{equation}
\log \frac{\widetilde{F}(r\vec{x},\vec{\beta}_{1})}{\widetilde{F}(r\vec{x},\vec{0})} = 2i\beta_{1}(x_{1}-x_{0})r - 2\beta_{1}^{2} \log \big(2r (x_{1}-x_{0})\big) + 2 \log \big(G(1+\beta_{1})G(1-\beta_{1})\big) + \bigO \Big( \frac{\log r}{r} \Big)
\end{equation}
as $r \to + \infty$, where $\vec{0}=(0,\ldots,0)$. This result matches with the known asymptotics \eqref{known result m=1 s1 neq 0}, with a slightly worse error term. Now, we use \eqref{DIFF IDENTITY s1 neq 0} with $k = 2$,  $\beta_{3}=\ldots=\beta_{m} = 0$, and with $\beta_{1} \in i \mathbb{R}$ fixed. With the notation $\vec{\beta}_{2} = (\beta_{1},\beta_{2},0,\ldots,0)$, we first rewrite \eqref{DIFF IDENTITY s1 neq 0} as follows
\begin{multline}\label{DIFF IDENTITY s1 neq 0 k = 2}
\partial_{\beta_2}\log \widetilde{F}(r \vec{x}, \vec{\beta}_{2}) = 2i (x_{2}-x_{0})r - 4\beta_{2} \log \big( 2r(x_{2}-x_{0}) \big) - 2 \beta_{1}\log\left( \frac{2r(x_{1}-x_{0})(x_{2}-x_{0})}{x_{2}-x_{1}}\right) \\
\hspace{0.55cm}-2\beta_{2} - 2 (\beta_{1}+\beta_{2}) + \beta_{2} \partial_{\beta_{2}} \log \frac{\Gamma(1+\beta_{2})}{\Gamma(1-\beta_{2})} + (\beta_{1}+\beta_{2})\partial_{\beta_{2}}\log \frac{\Gamma(1+\beta_{1}+\beta_{2})}{\Gamma(1-\beta_{1}-\beta_{2})} + \bigO \Big( \frac{\log r}{r} \Big).
\end{multline}
Again, since the above asymptotics are uniform in $\beta_{2}$ in compact subsets of $i \mathbb{R}$, they can be integrated over $\beta_{2}$ from $\beta_{2} = 0$ to an arbitrary $\beta_{2} \in i \mathbb{R}$ without changing the order of the error term. Using twice formula \eqref{final expression for I ell} with $\ell = 2$ (once for $I_{2}(\beta_{2};0)$ and once for $I_{2}(\beta_{2};\beta_{1})$), we obtain
\begin{multline}
\log \frac{\widetilde{F}(r\vec{x},\vec{\beta}_{2})}{\widetilde{F}(r\vec{x},\vec{\beta}_{1})} = 2i  \beta_{2}(x_{2}-x_{0})r - 2 \beta_{2}^{2} \log \big( 2r(x_{2}-x_{0}) \big) - 2 \beta_{1}\beta_{2}\log\left( \frac{2r(x_{1}-x_{0})(x_{2}-x_{0})}{x_{2}-x_{1}}\right) \\ 
+ \log \big( G(1+\beta_{2})G(1-\beta_{2}) \big) + \log \frac{G(1+\beta_{1}+\beta_{2})G(1-\beta_{1}-\beta_{2})}{G(1+\beta_{1})G(1-\beta_{1})} + \bigO \Big( \frac{\log r}{r} \Big).
\end{multline} 
We proceed similarly to integrate over the variables $\beta_{3},\ldots,\beta_{m}$. At the last step, we use \eqref{DIFF IDENTITY s1 neq 0} with $k = m$ and $\beta_{1},\ldots,\beta_{m-1}$ arbitrary. The integration of \eqref{DIFF IDENTITY s1 neq 0} over $\beta_{m}$ gives
\begin{multline}
\log \frac{\widetilde{F}(r\vec{x},\vec{\beta})}{\widetilde{F}(r\vec{x},\vec{\beta}_{m-1})} = 2i  \beta_{m}(x_{m}-x_{0})r - 2 \beta_{m}^{2} \log \big( 2r(x_{m}-x_{0}) \big) - 2 \sum_{j=1}^{m-1} \beta_{j}\beta_{m}\log\left( \frac{2r(x_{j}-x_{0})(x_{m}-x_{0})}{x_{m}-x_{j}}\right) \\ 
+ \log \big( G(1+\beta_{m})G(1-\beta_{m}) \big) + \log \frac{G(1+\beta_{1}+...+\beta_{m})G(1-\beta_{1}-...-\beta_{m})}{G(1+\beta_{1}+...+\beta_{m-1})G(1-\beta_{1}-...-\beta_{m-1})}+ \bigO \Big( \frac{\log r}{r} \Big),
\end{multline}
as $r \to + \infty$, where we have used the notation $\vec{\beta}_{m-1}=(\beta_{1},\ldots,\beta_{m-1},0)$. By summing the asymptotics obtained after each integration, we obtain
\begin{multline}
\log \frac{\widetilde{F}(r\vec{x},\vec{\beta})}{\widetilde{F}(r\vec{x},\vec{0})} = 2i \sum_{j=1}^{m} \beta_{j}(x_{j}-x_{0})r - \sum_{j=1}^{m} 2 \beta_{j}^{2} \log \big( 2r(x_{j}-x_{0}) \big) \\ - 2 \sum_{1 \leq j < k \leq m} \beta_{j}\beta_{k}\log\left( \frac{2r(x_{j}-x_{0})(x_{k}-x_{0})}{x_{k}-x_{j}}\right) + \sum_{j=1}^{m} \log \big( G(1+\beta_{j})G(1-\beta_{j}) \big) \\ + \log \big(G(1+\beta_{1}+...+\beta_{m})G(1-\beta_{1}-...-\beta_{m})\big) + \bigO \Big( \frac{\log r}{r} \Big),
\end{multline}
as $r \to + \infty$. Note from \eqref{def beta thm s1 neq 0} and \eqref{def of beta_j s1 neq 0} that $u_{j} = 2\pi i \beta_{j}$. Since $\widetilde{F}(r\vec{x},\vec{0}) = F(r\vec{x},\vec{1}) = 1$ (by \eqref{def of F tilde} and \eqref{F Fredholm}), this finishes the proof of Theorem \ref{thm:s1 neq 0}.
\section{Riemann-Hilbert analysis for $s_{p} = 0$}\label{Section: Steepest descent with sp=0}
In this section, $p \in \{1,\ldots,m\}$ is fixed and we obtain large $r$ asymptotics for $\Phi(rz;r\vec{x},\vec{s})$, in the case where the parameters are such that
\begin{itemize}
\item \vspace{-0.1cm} $s_{p}=0$ and $s_{1},\ldots,s_{p-1},s_{p+1},\ldots,s_{m}$ are in a compact subset of $(0,+\infty)$,
\item \vspace{-0.1cm} $x_{0},\ldots,x_{m}$ are in a compact subset of $\mathbb{R}$,
\item \vspace{-0.1cm} there exists $\delta > 0$ independent of $r$ such that
\begin{equation}\label{assumption xj delta sp=0}
\min_{0 \leq j < k \leq m} x_{k}-x_{j} \geq \delta.
\end{equation}
\end{itemize}
\subsection{Normalization of the RH problem with $g$-function}\label{subsection: g function sp=0}
Let us define
\begin{equation}\label{def of g sp=0}
g(z)=-i \theta(z)\sqrt{z-x_{p-1}}\sqrt{z-x_{p}},
\end{equation}
where the principal branches are taken for the square roots, and $\theta$ is defined as in \eqref{def of D as an integral}. The $g$-function satisfies the jumps
\begin{equation}\label{jumps for g sp=0}
g_{+}(z)+g_{-}(z) = 0 \qquad  \mbox{for} \qquad z \in (-\infty,x_{p-1})\cup (x_{p},+\infty),
\end{equation}
with an asymptotic behavior at $\infty$ given by
\begin{equation}\label{g at inf sp=0}
g(z) \sim \left\{ \begin{array}{l l}
-iz, & \Im z > 0, \\
iz, & \Im z < 0,
\end{array} \right. \qquad \mbox{ as } z \to \infty.
\end{equation}
We define the first transformation by
\begin{equation}\label{def of T sp=0}
T(z)= \begin{pmatrix}
\cos \Big( \frac{r}{2}(x_{p-1}+x_{p}) \Big) & \sin \Big( \frac{r}{2}(x_{p-1}+x_{p}) \Big) \\
- \sin \Big( \frac{r}{2}(x_{p-1}+x_{p}) \Big) & \cos \Big( \frac{r}{2}(x_{p-1}+x_{p}) \Big)
\end{pmatrix} \Phi(rz;r\vec{x},\vec{s})e^{-rg(z)\sigma_{3}}.
\end{equation}
The purpose of the constant prefactor matrix in \eqref{def of T sp=0} is to simplify the behavior of $T$ at $\infty$. After a computation using the behavior of $\Phi(rz;r\vec{x},\vec{s})$ as $z \to \infty$ given by \eqref{Asymp of Phi near infinity}, we obtain
\begin{equation} \label{eq:Tasympinf sp=0}
T(z) = \left( I + \frac{T_{1}}{z} + \bigO\left(z^{-2}\right) \right) Ne^{-\frac{\pi i}{4}\sigma_{3}}\left\{ \begin{array}{l l}
I, & \Im z > 0, \\
\begin{pmatrix}
0 & -1 \\ 1 & 0 
\end{pmatrix}, & \Im z < 0,
\end{array}	 \right. \qquad \mbox{as } z \to \infty,
\end{equation}
where
\begin{equation}\label{Kinf in terms of T1 sp=0}
\Phi_{1,21}(r\vec{x},\vec{s})-\Phi_{1,12}(r\vec{x},\vec{s}) = - \frac{r^{2}}{4}(x_{p}-x_{p-1})^{2} + r(T_{1,21}-T_{1,12}).
\end{equation}
The jumps for $T$ are obtained straightforwardly from those of $\Phi$ and from \eqref{jumps for g sp=0}. Since $s_{j} \neq 0$ for $j \neq p$, we verify that the jump matrix $T_{-}^{-1}T_{+}$ on $(x_{j-1},x_{j})$, $j \neq p$, can be factorized as in \eqref{factorization of the jump}. 
\subsection{Opening of the lenses}\label{subsection: S with sp=0}
Around each interval $(x_{j-1},x_{j})$, $j = 1,\ldots,m$, $j \neq p$, we let the lenses $\gamma_{j,+}$ and $\gamma_{j,-}$ denote open curves in the upper and lower half plane respectively, as shown in Figure \ref{fig:contour for S sp=0}. We also let $\Omega_{j,+}$ denote the region inside $\gamma_{j,+}\cup(x_{j-1},x_{j})$, and we let $\Omega_{j,-}$ denote the region inside $\gamma_{j,-}\cup(x_{j-1},x_{j})$. We define the $T \mapsto S$ transformation by
\begin{equation}\label{def of S sp=0}
S(z) = T(z) \prod_{\substack{j=1 \\ j \neq p}}^{m} \left\{ \begin{array}{l l}
\begin{pmatrix}
1 & 0 \\
-s_{j}^{-1}e^{-2rg(z)} & 1
\end{pmatrix}, & \mbox{if } z \in \Omega_{j,+}, \\
\begin{pmatrix}
1 & 0 \\
s_{j}^{-1}e^{-2rg(z)} & 1
\end{pmatrix}, & \mbox{if } z \in \Omega_{j,-}, \\
I, & \mbox{if } z \in \mathbb{C}\setminus(\Omega_{j,+}\cup \Omega_{j,-}).
\end{array} \right.
\end{equation}
\begin{figure}
\centering
\begin{tikzpicture}
\node at (3,0) {};
\draw (0,0) -- (5,0);
\draw (7,0) -- (13,0);
\draw (3,0) -- ($(3,0)+(135:3)$);
\draw (3,0) -- ($(3,0)+(-135:3)$);
\draw (10,0) -- ($(10,0)+(45:3)$);
\draw (10,0) -- ($(10,0)+(-45:3)$);

\draw[fill] (3,0) circle (0.05);
\draw[fill] (5,0) circle (0.05);
\draw[fill] (7,0) circle (0.05);
\draw[fill] (10,0) circle (0.05);

\node at (3,-0.3) {$x_{0}$};
\node at (5,-0.3) {$x_{1}$};
\node at (7,-0.3) {$x_{2}$};
\node at (10,-0.3) {$x_{m}$};

\draw[black,arrows={-Triangle[length=0.18cm,width=0.12cm]}]
($(3,0)+(135:1.5)$) --  ++(-45:0.001);
\draw[black,arrows={-Triangle[length=0.18cm,width=0.12cm]}]
($(3,0)+(-135:1.5)$) --  ++(45:0.001);
\draw[black,arrows={-Triangle[length=0.18cm,width=0.12cm]}]
(1.5,0) --  ++(0:0.001);

\draw[black,arrows={-Triangle[length=0.18cm,width=0.12cm]}]
(4.05,0) --  ++(0:0.001);
\draw[black,arrows={-Triangle[length=0.18cm,width=0.12cm]}]
(8.6,0) --  ++(0:0.001);

\draw[black,arrows={-Triangle[length=0.18cm,width=0.12cm]}]
($(10,0)+(45:1.5)$) --  ++(45:0.001);
\draw[black,arrows={-Triangle[length=0.18cm,width=0.12cm]}]
($(10,0)+(-45:1.5)$) --  ++(-45:0.001);
\draw[black,arrows={-Triangle[length=0.18cm,width=0.12cm]}]
(11.5,0) --  ++(0:0.001);

\draw (3,0) .. controls (3.7,0.8) and (4.3,0.8) .. (5,0);
\draw (3,0) .. controls (3.7,-0.8) and (4.3,-0.8) .. (5,0);
\draw (7,0) .. controls (8,1.3) and (9,1.3) .. (10,0);
\draw (7,0) .. controls (8,-1.3) and (9,-1.3) .. (10,0);

\draw[black,arrows={-Triangle[length=0.18cm,width=0.12cm]}]
(4.05,0.59) --  ++(0:0.001);
\draw[black,arrows={-Triangle[length=0.18cm,width=0.12cm]}]
(4.05,-0.59) --  ++(0:0.001);
\draw[black,arrows={-Triangle[length=0.18cm,width=0.12cm]}]
(8.6,0.98) --  ++(0:0.001);
\draw[black,arrows={-Triangle[length=0.18cm,width=0.12cm]}]
(8.6,-0.98) --  ++(0:0.001);
\end{tikzpicture}
\caption{Jump contours $\Sigma_{S}$ for the RH problem for $S$ with $m=3$ and $p = 2$.}
\label{fig:contour for S sp=0}
\end{figure}
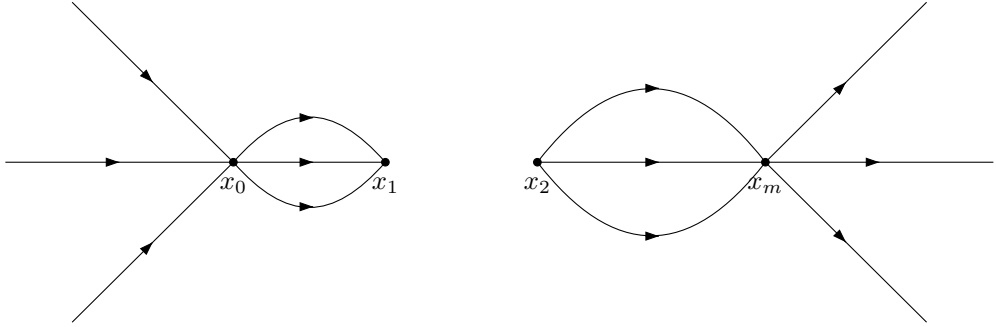
\hspace{-0.1cm}In a similar way in Subsection \ref{subsection: S with s1 neq 0}, we verify that $S$ satisfies the following RH problem.
\subsubsection*{RH problem for $S$}
\begin{enumerate}[label={(\alph*)}]
\item[(a)] $S : \C \backslash \Sigma_{S} \rightarrow \C^{2\times 2}$ is analytic, with
\begin{equation}\label{eq:defGamma sp=0}
\Sigma_{S}=(-\infty,x_{p-1}]\cup [x_{p},+\infty) \cup \gamma_{+}\cup \gamma_{-}, \qquad \gamma_{\pm} = \bigcup_{\substack{j=0 \\ j \neq p}}^{m+1} \gamma_{j,\pm},
\end{equation}
where $\Sigma_{S}$ is oriented as shown in Figure \ref{fig:contour for S sp=0} and
\begin{align*}
\gamma_{0,\pm} := x_{0} + e^{\pm \frac{3 \pi i}{4}}(0,+\infty), \qquad \gamma_{m+1,\pm} := x_{m} + e^{\pm \frac{\pi i}{4}}(0,+\infty).
\end{align*}
\item[(b)] The jumps for $S$ are given by
\begin{align*}
& S_{+}(z) = S_{-}(z)\begin{pmatrix}
0 & s_{j} \\ -s_{j}^{-1} & 0
\end{pmatrix}, & & z \in (x_{j-1},x_{j}), \, j = 0,\ldots,p-1,p+1,\ldots,m+1, \\
& S_{+}(z) = S_{-}(z)\begin{pmatrix}
1 & 0 \\
s_{j}^{-1} e^{-2 r g(z)} & 1
\end{pmatrix}, & & z \in \gamma_{j,\pm}, \, j = 0,\ldots,p-1,p+1,\ldots,m+1,
\end{align*}
where $x_{-1} := -\infty$, $x_{m+1} := +\infty$, and we recall that $s_{0} = s_{m+1} = 1$.
\item[(c)] As $z \rightarrow \infty$, we have
\begin{equation}
\label{eq:Sasympinf sp=0}
S(z) = \left( I + \frac{T_{1}}{z} + \bigO\left(z^{-2}\right) \right) Ne^{-\frac{\pi i}{4}\sigma_{3}}\left\{ \begin{array}{l l}
I, & \Im z > 0, \\
\begin{pmatrix}
0 & -1 \\ 1 & 0 
\end{pmatrix}, & \Im z < 0.
\end{array}	 \right.
\end{equation}
As $z \to x_j$ from outside the lenses, $j = 0,\ldots, m$, we have
\begin{equation}\label{local behavior near -xj of S sp=0}
S(z) = \begin{pmatrix}
\bigO(1) & \bigO(\log(z-x_{j})) \\
\bigO(1) & \bigO(\log(z-x_{j}))
\end{pmatrix}.
\end{equation}
\end{enumerate}
\begin{lemma}\label{lemma: real part of g}
The $g$-function defined in \eqref{def of g sp=0} satisfies
\begin{equation}\label{real part g function sp=0}
\{z : \Re(g(z)) > 0\} = \mathbb{C}\setminus \big( (-\infty,x_{p-1}]\cup [x_{p},+\infty) \big).
\end{equation}
\end{lemma}
\begin{proof}
Clearly, $\Re(g(z)) = 0$ if and only if $g(z)^{2} = -(z-x_{p-1})(z-x_{p}) \leq 0$. Since $g(z)^{2} \leq 0$ for $z \in (-\infty,x_{p-1}]\cup [x_{p},+\infty)$, this proves $(-\infty,x_{p-1}]\cup [x_{p},+\infty) \subseteq \{z : \Re(g(z)) = 0\}$. On the other hand, for each $c \in \mathbb{R}^{-}$, the equation $-(z-x_{p-1})(z-x_{p})=c$ admits exactly two solutions (counting multiplicities), and from the graph of $g(z)^{2}$ for $z \in \mathbb{R}$, it is immediate to verify that these two solutions lie on $(-\infty,x_{p-1}]\cup [x_{p},+\infty)$, which proves $(-\infty,x_{p-1}]\cup [x_{p},+\infty) \supseteq \{z : \Re(g(z)) = 0\}$. Since $\Re(g(z))$ is continuous, all what remains is to determine the sign of $\Re(g(z))$ on $\mathbb{C}\setminus \big( (-\infty,x_{p-1}]\cup [x_{p},+\infty) \big)$. From the behavior of $g(z)$ as $z \to i \infty$, see \eqref{g at inf sp=0}, we conclude that this sign is positive.
\end{proof}
We deduce from Lemma \ref{lemma: real part of g} that the jump matrices for $S$ tend to the identity matrix exponentially fast as $r \to + \infty$ on the lenses. This convergence is uniform for $z$ outside of fixed neighborhoods of $x_{j}$, $j \in \{0,1,\ldots,m\}$, but is not uniform as $r \to + \infty$ and simultaneously $z \to x_{j}$, $j \in \{0,1,\ldots,m\}$.

\subsection{Global parametrix}\label{subsection: Global param sp=0}
By ignoring the jumps for $S$ on the lenses, we are led to consider the following RH problem, whose solution is  denoted $P^{(\infty)}$. We will show in Subsection \ref{subsection Small norm sp=0} that $P^{(\infty)}$ is a good approximation for $S$ outside neighborhoods of $x_{j}$, $j = 0,1,\ldots,m$.
\subsubsection*{RH problem for $P^{(\infty)}$}
\begin{enumerate}[label={(\alph*)}]
\item[(a)] $P^{(\infty)} : \C  \setminus \big( (-\infty,x_{p-1}]\cup[x_{p},+\infty) \big) \rightarrow \C^{2\times 2}$ is analytic.
\item[(b)] The jumps for $P^{(\infty)}$ are given by
\begin{align*}
& P^{(\infty)}_{+}(z) = P^{(\infty)}_{-}(z)\begin{pmatrix}
0 & s_{j} \\ -s_{j}^{-1} & 0
\end{pmatrix}, & & z \in (x_{j-1},x_{j}), \, j = 0,\ldots,p-1,p+1,\ldots,m+1.
\end{align*}
\item[(c)] As $z \rightarrow \infty$, we have
\begin{equation}
\label{eq:Pinf asympinf sp=0}
P^{(\infty)}(z) = \left( I + \frac{P^{(\infty)}_{1}}{z} + \bigO\left(z^{-2}\right) \right) Ne^{-\frac{\pi i}{4}\sigma_{3}}\left\{ \begin{array}{l l}
I, & \Im z > 0, \\
\begin{pmatrix}
0 & -1 \\ 1 & 0 
\end{pmatrix}, & \Im z < 0.
\end{array}	 \right.
\end{equation}
for a certain matrix $P_{1}^{(\infty)}$ independent of $z$.
\item[(d)] As $z \to x_{j}$, $j \in \{0,\ldots,m\}\setminus \{p-1,p\}$, we have $P^{(\infty)}(z) = \begin{pmatrix}
\bigO(1) & \bigO(1) \\
\bigO(1) & \bigO(1)
\end{pmatrix}$.

As $z \to x_{j}$, $j \in \{p-1,p\}$, we have $P^{(\infty)}(z) = \begin{pmatrix}
\bigO((z-x_{j})^{-1/4}) & \bigO((z-x_{j})^{-1/4}) \\
\bigO((z-x_{j})^{-1/4}) & \bigO((z-x_{j})^{-1/4})
\end{pmatrix}$.
\end{enumerate}
The construction of $P^{(\infty)}$ relies on the following function $D$:
\begin{multline}
D(z) = \exp \Bigg( \theta(z)\sqrt{z-x_{p-1}}\sqrt{z-x_{p}} \bigg[ -\sum_{j=1}^{p-1} \frac{\log s_{j}}{2\pi i} \int_{x_{j-1}}^{x_{j}} \frac{1}{\sqrt{x_{p-1}-u}\sqrt{x_{p}-u}}\frac{du}{u-z} \\ + \sum_{j=p+1}^{m} \frac{\log s_{j}}{2\pi i} \int_{x_{j-1}}^{x_{j}} \frac{1}{\sqrt{u-x_{p-1}}\sqrt{u-x_{p}}} \frac{du}{u-z} \bigg] \Bigg),
\end{multline}
where the principal branches are taken for $\sqrt{z-x_{p-1}}$ and $\sqrt{z-x_{p}}$. $D$ satisfies the following jumps
\begin{align*}
& D_{+}(z)D_{-}(z) = s_{j}, & &  \mbox{for } z \in (x_{j-1},x_{j}), \, j \in \{ 0,\ldots,m+1\}\setminus \{p\}.
\end{align*}
Using primitives, one can rewrite $D$ as follows
\begin{align*}
D(z) = & \prod_{j=0}^{p-2}\Bigg( \frac{\sqrt{z-x_{p-1}}\sqrt{x_{p}-x_{j}}-\sqrt{z-x_{p}}\sqrt{x_{p-1}-x_{j}}}{\sqrt{z-x_{p-1}}\sqrt{x_{p}-x_{j}}+\sqrt{z-x_{p}}\sqrt{x_{p-1}-x_{j}}} \Bigg)^{\beta_{j}\theta(z)} \\
& \times \prod_{j=p+1}^{m}\Bigg( \frac{\sqrt{z-x_{p}}\sqrt{x_{j}-x_{p-1}}-\sqrt{z-x_{p-1}}\sqrt{x_{j}-x_{p}}}{\sqrt{z-x_{p}}\sqrt{x_{j}-x_{p-1}}+\sqrt{z-x_{p-1}}\sqrt{x_{j}-x_{p}}} \Bigg)^{\beta_{j}\theta(z)}, 
\end{align*}
where again the principal branches are taken for $\sqrt{z-x_{p}}$ and $\sqrt{z-x_{p-1}}$, and
\begin{align}
& \beta_{j} = \frac{1}{2\pi i} \log \frac{s_{j}}{s_{j+1}}, \qquad j \in \{0,\ldots,m \}\setminus \{p-1,p\} \label{def of beta sp=0} \\
& s_{0} = s_{m+1} = 1.
\end{align}
As $z \to \infty$, $\Im z > 0$, $D(z) = D_{\infty}(1+d_{1}z^{-1}+\bigO(z^{-2}))$, where
\begin{equation}
D_{\infty} = \prod_{j=0}^{p-2} \Bigg( \frac{\sqrt{x_{p}-x_{j}}-\sqrt{x_{p-1}-x_{j}}}{\sqrt{x_{p}-x_{j}}+\sqrt{x_{p-1}-x_{j}}} \Bigg)^{\beta_{j}} \times  \prod_{j=p+1}^{m} \Bigg( \frac{\sqrt{x_{j}-x_{p-1}}-\sqrt{x_{j}-x_{p}}}{\sqrt{x_{j}-x_{p-1}}+\sqrt{x_{j}-x_{p}}} \Bigg)^{\beta_{j}},
\end{equation}
and
\begin{equation}\label{def of d1}
d_{1} = \sum_{j=0}^{p-2} \beta_{j} \sqrt{x_{p}-x_{j}}\sqrt{x_{p-1}-x_{j}}-\sum_{j=p+1}^{m} \beta_{j} \sqrt{x_{j}-x_{p}}\sqrt{x_{j}-x_{p-1}}.
\end{equation}
Let us define
\begin{equation}\label{def of Pinf sp=0}
P^{(\infty)}(z) = \widehat{D} \begin{pmatrix}
\frac{\beta(z)}{\sqrt{2}} & -\frac{\beta(z)^{-1}}{\sqrt{2}} \\
\frac{\beta(z)}{\sqrt{2}} & \frac{\beta(z)^{-1}}{\sqrt{2}}
\end{pmatrix} N D(z)^{-\sigma_{3}}, \qquad \widehat{D} = \begin{pmatrix}
\frac{D_{\infty}+D_{\infty}^{-1}}{2} & -\frac{i(D_{\infty}-D_{\infty}^{-1})}{2} \\
\frac{i(D_{\infty}-D_{\infty}^{-1})}{2} & \frac{D_{\infty}+D_{\infty}^{-1}}{2}
\end{pmatrix},
\end{equation}
where $\beta(z) = \sqrt[4]{\frac{z-x_{p}}{z-x_{p-1}}}$ has branch cuts on $(-\infty,x_{p-1})\cup (x_{p},+\infty)$ and  satisfies 
\begin{align*}
\beta(z) \sim 1 \quad \mbox{as }z \to \infty, \Im z > 0 \qquad \mbox{ and } \qquad \beta(z) \sim i \quad \mbox{as }z \to \infty, \Im z < 0.
\end{align*}
We verify that $P^{(\infty)}$ satisfies the properties (a), (b) and (c) of the RH problem for $P^{(\infty)}$. Furthermore, after a computation we obtain an explicit expression for $P_{1}^{(\infty)}$:
\begin{equation}\label{Pinf 1 12 sp=0}
P_{1}^{(\infty)} = \begin{pmatrix}
\frac{i}{8}(x_{p}-x_{p-1})(D_{\infty}^{2}-D_{\infty}^{-2}) & i d_{1} - \frac{1}{8}(x_{p}-x_{p-1})(D_{\infty}^{2}+D_{\infty}^{-2}) \\
-id_{1}- \frac{1}{8}(x_{p}-x_{p-1})(D_{\infty}^{2}+D_{\infty}^{-2}) & -\frac{i}{8}(x_{p}-x_{p-1})(D_{\infty}^{2}-D_{\infty}^{-2})
\end{pmatrix}.
\end{equation}
In the rest of this subsection, we compute the leading terms in the asymptotics of $D(z)$ as $z \to x_{j}$, $j=0,\ldots,m$. 
As $z \to x_{j}$, $j \neq p,p-1$, $\Im z > 0$, we have
\begin{equation}\label{lol6 sp=0}
D(z) = \sqrt{s_{j+1}} \, (z-x_{j})^{\beta_{j}} \prod_{\substack{k=0 \\ k \neq p-1,p}}^{m} T_{k,j}^{-\beta_{k}}  (1+\bigO(z-x_{j})),
\end{equation}
where
\begin{align}
& T_{k,j} = \frac{\sqrt{|x_{k}-x_{p}|}\sqrt{|x_{j}-x_{p-1}|}+\sqrt{|x_{k}-x_{p-1}|}\sqrt{|x_{j}-x_{p}|}}{\big|\sqrt{|x_{k}-x_{p}|}\sqrt{|x_{j}-x_{p-1}|}-\sqrt{|x_{k}-x_{p-1}|}\sqrt{|x_{j}-x_{p}|}\big|}, \quad k \neq j, \\
& T_{j,j} = \frac{4|x_{j}-x_{p-1}| \, |x_{j}-x_{p}|}{x_{p}-x_{p-1}}. \label{T diagonal sp=0}
\end{align}
As $z \to x_{p}$, $\Im z > 0$, we have
\begin{equation}\label{asymp D at xp sp=0}
D(z) = \sqrt{s_{p+1}}\Big(1-\frac{2 d_{x_{p}}}{\sqrt{x_{p}-x_{p-1}}}\sqrt{z-x_{p}}+\bigO(z-x_{p})\Big),
\end{equation}
with
\begin{equation}\label{def of dxp}
d_{x_{p}} = \sum_{\substack{j=0 \\ j \neq p-1,p}}^{m} \beta_{j} \frac{\sqrt{|x_{j}-x_{p-1}|}}{\sqrt{|x_{j}-x_{p}|}},
\end{equation}
and as $z \to x_{p-1}$, $\Im z > 0$, we have
\begin{equation}\label{asymp D at xp-1 sp=0}
D(z) = \sqrt{s_{p-1}}  \Big(1+\frac{2 i d_{x_{p-1}}}{\sqrt{x_{p}-x_{p-1}}}\sqrt{z-x_{p-1}}+\bigO(z-x_{p-1})\Big),
\end{equation}
with
\begin{equation}\label{def of dxp-1}
d_{x_{p-1}} = \sum_{\substack{j=0 \\ j \neq p-1,p}}^{m} \beta_{j}\frac{\sqrt{|x_{j}-x_{p}|}}{\sqrt{|x_{j}-x_{p-1}|}}.
\end{equation}

\subsection{Local parametrices}\label{subsection: local param sp = 0}

In this subsection, we construct local parametrices $P^{(x_{j})}$ around $x_{j}$, $j \in \{0,\ldots,m\}$. To be more precise, let $\mathcal{D}_{x_{j}}$ be small open disks centered at $x_{j}$, $j=0,1,\ldots,m$ whose radii are equal to $\frac{\delta}{3}$, where $\delta$ is defined in \eqref{assumption xj delta sp=0}. The definition of the radii ensures that the disks do not intersect each other.  We require $P^{(x_{j})}$ to satisfy the same jumps as $S$ in $\mathcal{D}_{x_{j}}$, and to match with $P^{(\infty)}$ on $\partial \mathcal{D}_{x_{j}}$ in the sense that
\begin{equation}\label{matching weak sp=0}
P^{(x_{j})}(z) = (I+o(1))P^{(\infty)}(z), \qquad \mbox{ as } r \to +\infty,
\end{equation}
is required to hold uniformly for $z \in \partial \mathcal{D}_{x_{j}}$. Finally, we also require 
\begin{equation}\label{match at the center sp=0}
S(z) P^{(x_{j})}(z)^{-1} = \bigO(1), \qquad \mbox{ as } z \to x_{j}.
\end{equation}
\subsubsection{Local parametrices around $x_{j}$, $j \in \{ 0,\ldots,m\}\setminus \{p-1,p\}$}\label{subsection: local param HG sp=0}
The construction of $P^{(x_{j})}$ for $j \in \{0,1,\ldots,m\}\setminus \{p-1,p\}$ is similar to the one done in Subsection \ref{Section: local param s1 neq 0}, and involves confluent hypergeometric functions. The function
\begin{equation}\label{def conformal map sp=0}
f_{x_{j}}(z) = -2 \left\{ \begin{array}{l l}
g(z)-g_{+}(x_{j}), & \mbox{if } \Im z > 0, \\
-(g(z)-g_{-}(x_{j})), & \mbox{if } \Im z < 0,
\end{array} \right. 
\end{equation}
is a conformal map from $\mathcal{D}_{x_{j}}$ to a neighborhood of $0$ whose expansion as $z \to x_{j}$ is given by
\begin{equation}\label{expansion conformal map xj sp=0}
f_{x_{j}}(z) = i c_{x_{j}}(z-x_{j})(1+\bigO(z-x_{j})), \qquad c_{x_{j}} = \left\{ \begin{array}{l l}
\ds \frac{x_{p-1}+x_{p}-2x_{j}}{\sqrt{x_{p-1}-x_{j}}\sqrt{x_{p}-x_{j}}}, & \mbox{if } j = 0,\ldots,p-2, \\[0.3cm]
\ds \frac{2x_{j}-x_{p-1}-x_{p}}{\sqrt{x_{j}-x_{p-1}}\sqrt{x_{j}-x_{p}}}, & \mbox{if } j = p+1,\ldots,m.
\end{array} \right.
\end{equation}
Note that $f_{x_{j}}(\mathbb{R}\cap \mathcal{D}_{x_{j}})\subset i \mathbb{R}$. In order to use the model RH problem for $\Phi_{\mathrm{HG}}$, we deform the lenses is a small neighborhood of $x_{j}$ such that the jump contour for $P^{(x_{j})}$ is mapped by $f_{x_{j}}$ to a subset of $\Sigma_{\mathrm{HG}}$ (see Figure \ref{Fig:HG}), i.e.
\begin{equation}\label{deformation of the lenses local param xj sp=0}
f_{x_{j}}((\gamma_{j,+}\cup \gamma_{j+1,+})\cap \mathcal{D}_{x_{j}}) \subset \Gamma_{3} \cup \Gamma_{2}, \qquad f_{x_{j}}((\gamma_{j,-}\cup \gamma_{j+1,-})\cap \mathcal{D}_{x_{j}}) \subset \Gamma_{5} \cup \Gamma_{6},
\end{equation}
where $\Gamma_{3}$, $\Gamma_{2}$, $\Gamma_{5}$ and $\Gamma_{6}$ are as shown in Figure \ref{Fig:HG}.
We seek for $P^{(x_{j})}$ in the form
\begin{equation}\label{lol10 sp=0}
P^{(x_{j})}(z) = E_{x_{j}}(z) \Phi_{\mathrm{HG}}(rf_{x_{j}}(z);\beta_{j})(s_{j}s_{j+1})^{-\frac{\sigma_{3}}{4}}e^{-rg(z)\sigma_{3}},
\end{equation}
for a suitable analytic matrix valued function $E_{x_{j}}$. We recall that the parameter $\beta_{j}$ in \eqref{lol10 sp=0} is given by \eqref{def of beta sp=0}. Since $E_{x_{j}}$ is analytic, it is straightforward to see from the jumps of $\Phi_{\mathrm{HG}}$ (given by \eqref{jumps PHG3}) that $P^{(x_{j})}$ given by \eqref{lol10 sp=0} satisfies the same jumps as $S$ inside $\mathcal{D}_{x_{j}}$. In view of \eqref{Asymptotics HG}, we see that to satisfy the matching condition \eqref{matching weak sp=0}, we are forced to define $E_{x_{j}}$ by
\begin{equation}\label{def of Ej sp=0}
E_{x_{j}}(z) = P^{(\infty)}(z) (s_{j} s_{j+1})^{\frac{\sigma_{3}}{4}} \left\{ \begin{array}{l l}
\ds \sqrt{\frac{s_{j+1}}{s_{j}}}^{\sigma_{3}}, & \Im z > 0 \\
\begin{pmatrix}
0 & 1 \\ -1 & 0
\end{pmatrix}, & \Im z < 0
\end{array} \right\} e^{rg_{+}(x_{j})\sigma_{3}}(rf_{x_{j}}(z))^{\beta_{j}\sigma_{3}}.
\end{equation}
It can be verified from the jumps for $P^{(\infty)}$ that $E_{x_{j}}$ defined by \eqref{def of Ej sp=0} has no jump at all inside $\mathcal{D}_{x_{j}}$. Furthermore, using \eqref{lol6 sp=0}, we verify that $E_{x_{j}}(z)$ is bounded as $z \to x_{j}$ and $E_{x_{j}}$ is then analytic in the whole disk $\mathcal{D}_{x_{j}}$, as desired. Since $P^{(x_{j})}$ and $S$ have exactly the same jumps on $(\mathbb{R}\cup \gamma_{+}\cup \gamma_{-})\cap \mathcal{D}_{x_{j}}$, $S(z)P^{(x_{j})}(z)^{-1}$ is analytic in $\mathcal{D}_{x_{j}} \setminus \{x_{j}\}$. As $z \to x_{j}$ from outside the lenses, by condition (d) in the RH problem for $S$ and by \eqref{lol 35}, $S(z)P^{(x_{j})}(z)^{-1}$ behaves as $\bigO(\log(z-x_{j}))$. This means that $x_{j}$ is a removable singularity of $S(z)P^{(x_{j})}(z)^{-1}$ and therefore \eqref{match at the center sp=0} holds. We will need later a more detailed matching condition than \eqref{matching weak sp=0}, which can be obtained from \eqref{Asymptotics HG}:
\begin{equation}\label{matching strong -x_j sp=0}
P^{(x_{j})}(z)P^{(\infty)}(z)^{-1} = I + \frac{1}{rf_{x_{j}}(z)}E_{x_{j}}(z) \Phi_{\mathrm{HG},1}(\beta_{j})E_{x_{j}}(z)^{-1} + \bigO(r^{-2}),
\end{equation}
as $r \to + \infty$, uniformly for $z \in \partial \mathcal{D}_{x_{j}}$, where $\Phi_{\mathrm{HG},1}(\beta_{j})$ is given by \eqref{def of tau}. Also, using \eqref{def of Pinf sp=0}, \eqref{lol6 sp=0}-\eqref{T diagonal sp=0} and \eqref{expansion conformal map xj sp=0}, we note for later use that 
\begin{equation}\label{E_j at x_j sp=0}
E_{x_{j}}(x_{j}) = \frac{1}{\sqrt{2}} \widehat{D}\begin{pmatrix}
1 & - 1 \\
1 & 1
\end{pmatrix} \left( \frac{|x_{j}-x_{p}|}{|x_{j}-x_{p-1}|} \right)^{\frac{\sigma_{3}}{4}}
N\Lambda_{j}^{\sigma_{3}},
\end{equation}
where
\begin{equation}\label{def of Lambda sp=0}
\Lambda_{j} = e^{r g_{+}(x_{j})} \left( T_{j,j} c_{x_{j}}r \right)^{\beta_{j}} \prod_{\substack{k=0 \\ k \neq j,p-1,p}}^{m} T_{k,j}^{\beta_{k}}.
\end{equation}
\subsubsection{Local parametrix around $x_{p}$}
For the local parametrix $P^{(x_{p})}$, we need to use another model RH problem whose solution $\Phi_{\mathrm{Be}}$ is expressed in terms of Bessel functions. This model RH problem is well-known, see e.g. \cite{KMcLVAV}, and is recalled in Subsection \ref{subsection:Model Bessel} for the convenience of the reader. Consider the function
\begin{equation}\label{conformal map near xp sp=0}
f_{x_{p}}(z) = -\frac{g(z)^{2}}{4} = \frac{(z-x_{p-1})(z-x_{p})}{4}.
\end{equation}
This is a conformal map from $\mathcal{D}_{x_{p}}$ to a neighborhood of $0$ whose expansion as $z \to x_{p}$ is given by
\begin{equation}\label{expansion fxp at xp sp=0}
f_{x_{p}}(z) = c_{x_{p}}^{2}(z-x_{p})\left(1+\frac{z-x_{p}}{x_{p}-x_{p-1}}+\bigO\big((z-x_{p})^{2}\big)\right), \qquad c_{x_{p}} = \frac{\sqrt{x_{p}-x_{p-1}}}{2} > 0.
\end{equation}
We choose the lenses such that they are mapped by $-f_{x_{p}}$ onto a subset of $\Sigma_{\mathrm{Be}}$ ($\Sigma_{\mathrm{Be}}$ is the jump contour for $\Phi_{\mathrm{Be}}$, see Figure \ref{figBessel}):
\begin{equation*}
-f_{x_{p}}(\gamma_{p+1,+}) \subset e^{-\frac{2\pi i}{3}}\mathbb{R}^{+}, \qquad -f_{x_{p}}(\gamma_{p+1,-}) \subset e^{\frac{2\pi i}{3}}\mathbb{R}^{+}.
\end{equation*} 
If we take $P^{(x_{p})}$ of the form
\begin{equation}\label{def of P^xp sp=0}
P^{(x_{p})}(z) = E_{x_{p}}(z)\sigma_{3}\Phi_{\mathrm{Be}}(-r^{2}f_{x_{p}}(z))\sigma_{3}s_{p+1}^{-\frac{\sigma_{3}}{2}}e^{-rg(z)\sigma_{3}},
\end{equation}
with $E_{x_{p}}$ analytic in $\mathcal{D}_{x_{p}}$, then it is straightforward to verify from \eqref{Jump for P_Be} that $P^{(x_{p})}$ has the same jumps as $S$ in $\mathcal{D}_{x_{p}}$. To satisfy the matching condition, by \eqref{large z asymptotics Bessel}, we need to define $E_{x_{p}}$ by
\begin{equation}\label{def E at xp sp=0}
E_{x_{p}}(z) = P^{(\infty)}(z)s_{p+1}^{\frac{\sigma_{3}}{2}}N \left( 2\pi r (-f_{x_{p}}(z))^{1/2} \right)^{\frac{\sigma_{3}}{2}},
\end{equation}
where we take the principal branches for the square roots. We verify from the jumps for $P^{(\infty)}$ that $E_{x_{p}}$ has no jumps in $\mathcal{D}_{x_{p}}$, and has a removable singularity at $x_{p}$; therefore $E_{x_{p}}$ is analytic in $\mathcal{D}_{x_{p}}$, as required. We will need later a more detailed matching condition than \eqref{matching weak sp=0}, which can be obtained using \eqref{large z asymptotics Bessel}:
\begin{equation}\label{matching strong xp sp=0}
P^{(x_{p})}(z)P^{(\infty)}(z)^{-1} = I + \frac{1}{r(-f_{x_{p}}(z))^{1/2}}P^{(\infty)}(z)s_{p+1}^{\frac{\sigma_{3}}{2}}\sigma_{3}\Phi_{\mathrm{Be},1}\sigma_{3}s_{p+1}^{-\frac{\sigma_{3}}{2}}P^{(\infty)}(z)^{-1} + \bigO(r^{-2}),
\end{equation}
as $r \to +\infty$ uniformly for $z \in \partial \mathcal{D}_{x_{p}}$, where $\Phi_{\mathrm{Be},1}$ is given below \eqref{large z asymptotics Bessel}. Using \eqref{def of Pinf sp=0}, \eqref{asymp D at xp sp=0}, \eqref{expansion fxp at xp sp=0} and the expansion
\begin{equation}\label{explicit asymp for branch cut at xp sp=0}
(-f_{x_{p}}(z))^{1/2} = -i \, c_{x_{p}} \sqrt{z-x_{p}}(1+\bigO(z-x_{p})), \qquad \mbox{as } z \to x_{p}, \; \Im z >0,
\end{equation}
we obtain $E_{x_{p}}(x_{p})$ by taking the limit of $E_{x_{p}}(z)$ as $z \to x_{p}$ from the upper half plane:
\begin{equation}\label{E0 at 0 sp=0}
E_{x_{p}}(x_{p}) = \frac{1}{\sqrt{2}} \widehat{D}  \begin{pmatrix}
1 & -1 \\
1 & 1
\end{pmatrix} \begin{pmatrix}
0 & i \\
i & -2d_{x_{p}}
\end{pmatrix} e^{-\frac{\pi i}{4}\sigma_{3}}\big( \pi(x_{p}-x_{p-1})r \big)^{\frac{\sigma_{3}}{2}}.
\end{equation}

\subsubsection{Local parametrix around $x_{p-1}$}
The local parametrix $P^{(x_{p-1})}$ is also constructed in terms of Bessel functions, and relies on the model RH problem $\Phi_{\mathrm{Be}}$. The function
\begin{equation}\label{conformal map near xp-1 sp=0}
f_{x_{p-1}}(z) = \frac{g(z)^{2}}{4} = -\frac{(z-x_{p-1})(z-x_{p})}{4}
\end{equation}
is a conformal map from $\mathcal{D}_{x_{p-1}}$ to a neighborhood of $0$ whose expansion as $z \to x_{p-1}$ is given by
\begin{equation}\label{expansion fxp-1 at xp-1 sp=0}
f_{x_{p-1}}(z) = c_{x_{p-1}}^{2}(z-x_{p-1})\left(1-\frac{z-x_{p-1}}{x_{p}-x_{p-1}}+\bigO\big((z-x_{p-1})^{2}\big)\right), \qquad c_{x_{p-1}} = \frac{\sqrt{x_{p}-x_{p-1}}}{2} > 0.
\end{equation}
In a neighborhood of $x_{p-1}$, we deform the lenses such that
\begin{equation*}
f_{x_{p-1}}(\gamma_{p-1,+}) \subset e^{\frac{2\pi i}{3}}\mathbb{R}^{+}, \qquad f_{x_{p-1}}(\gamma_{p-1,-}) \subset e^{-\frac{2\pi i}{3}}\mathbb{R}^{+}.
\end{equation*} 
In this way, the jump contour for $P^{(x_{p-1})}$ is mapped by $f_{x_{p-1}}$ onto a subset of $\Sigma_{\mathrm{Be}}$. We take $P^{(x_{p-1})}$ of the form
\begin{equation}\label{def of P^xp-1 sp=0}
P^{(x_{p-1})}(z) = E_{x_{p-1}}(z)\Phi_{\mathrm{Be}}(r^{2}f_{x_{p-1}}(z))s_{p-1}^{-\frac{\sigma_{3}}{2}}e^{-rg(z)\sigma_{3}},
\end{equation}
where $E_{x_{p-1}}$ is analytic in $\mathcal{D}_{x_{p-1}}$. Using \eqref{Jump for P_Be}, it is straightforward to see that $P^{(x_{p-1})}$ has the same jumps as $S$ in $\mathcal{D}_{x_{p-1}}$. To satisfy the matching condition \eqref{matching weak sp=0}, using \eqref{large z asymptotics Bessel} we conclude that $E_{x_{p-1}}$ needs to be defined as follows
\begin{equation}\label{def of E xp-1 sp=0}
E_{x_{p-1}}(z) = P^{(\infty)}(z)s_{p-1}^{\frac{\sigma_{3}}{2}}N^{-1} \left( 2\pi r (f_{x_{p-1}}(z))^{1/2} \right)^{\frac{\sigma_{3}}{2}}.
\end{equation}
It can be verified from the jumps for $P^{(\infty)}$ that $E_{x_{p-1}}$ has no jumps in $\mathcal{D}_{x_{p-1}}$ and has a removable singularity at $x_{p-1}$. We conclude that $E_{x_{p-1}}$ is analytic in $\mathcal{D}_{x_{p-1}}$, as required. We will need later a more detailed matching condition than \eqref{matching weak sp=0}, which can be obtained using \eqref{large z asymptotics Bessel}:
\begin{equation}\label{matching strong xp-1 sp=0}
P^{(x_{p-1})}(z)P^{(\infty)}(z)^{-1} = I + \frac{1}{r(f_{x_{p-1}}(z))^{1/2}}P^{(\infty)}(z)s_{p-1}^{\frac{\sigma_{3}}{2}}\Phi_{\mathrm{Be},1}s_{p-1}^{-\frac{\sigma_{3}}{2}}P^{(\infty)}(z)^{-1} + \bigO(r^{-2}),
\end{equation}
as $r \to +\infty$ uniformly for $z \in \partial \mathcal{D}_{x_{p-1}}$, where $\Phi_{\mathrm{Be},1}$ is given below \eqref{large z asymptotics Bessel}. Furthermore, using \eqref{def of Pinf sp=0}, \eqref{asymp D at xp-1 sp=0} and \eqref{conformal map near xp-1 sp=0}, one shows that
\begin{equation}\label{E xp-1 at xp-1 sp=0}
E_{x_{p-1}}(x_{p-1}) = \frac{1}{\sqrt{2}}\widehat{D}  \begin{pmatrix}
1 & -1 \\
1 & 1
\end{pmatrix}\begin{pmatrix}
1 & -2id_{x_{p-1}} \\
0 & 1
\end{pmatrix} e^{\frac{\pi i}{4}\sigma_{3}}\big( \pi(x_{p}-x_{p-1})r \big)^{\frac{\sigma_{3}}{2}}.
\end{equation}
\subsection{Small norm problem}\label{subsection Small norm sp=0}
The last transformation of the steepest descent is defined by
\begin{equation}\label{def of R sp=0}
R(z) = \left\{ \begin{array}{l l}
S(z)P^{(\infty)}(z)^{-1}, & \mbox{for } z \in \mathbb{C}\setminus \bigcup_{j=0}^{m}\mathcal{D}_{x_{j}}, \\
S(z)P^{(x_{j})}(z)^{-1}, & \mbox{for } z \in \mathcal{D}_{x_{j}}, \, j \in \{0,1,\ldots,m\}.
\end{array} \right.
\end{equation}
It follows from the analysis of Subsection \ref{subsection: local param sp = 0} that $R$ is analytic inside the $m+1$ disks. Since the jumps of $P^{(\infty)}$ and of $S$ are the same on $(x_{j-1},x_{j})$, $j=1,\ldots,m$, we conclude that $R$ is analytic on $\mathbb{C}\setminus \Sigma_{R}$, where 
\begin{align*}
\Sigma_{R} = \bigcup_{j=0}^{m} \partial \mathcal{D}_{x_{j}} \cup \bigg( (\gamma_{+}\cup  \gamma_{-}) \setminus \bigcup_{j=0}^{m} \mathcal{D}_{x_{j}} \bigg),
\end{align*}
see Figure \ref{fig:contour for R sp=0}.
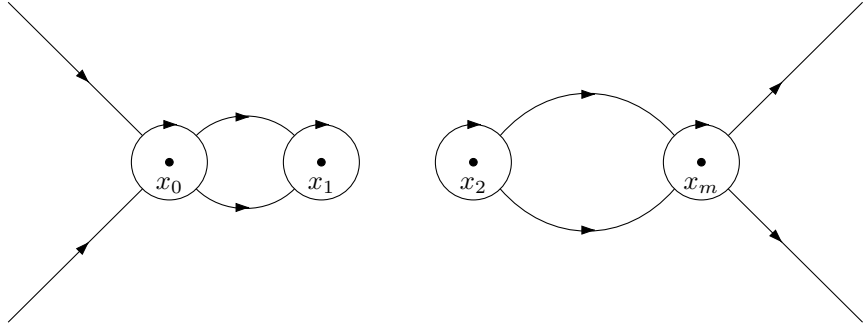
\begin{figure}
\centering
\begin{tikzpicture}
\node at (3,0) {};
\draw ($(3,0)+(135:0.5)$) -- ($(3,0)+(135:3)$);
\draw ($(3,0)+(-135:0.5)$) -- ($(3,0)+(-135:3)$);
\draw ($(10,0)+(45:0.5)$) -- ($(10,0)+(45:3)$);
\draw ($(10,0)+(-45:0.5)$) -- ($(10,0)+(-45:3)$);

\draw[fill] (3,0) circle (0.05);
\draw (3,0) circle (0.5);
\draw[fill] (5,0) circle (0.05);
\draw (5,0) circle (0.5);
\draw[fill] (7,0) circle (0.05);
\draw (7,0) circle (0.5);
\draw[fill] (10,0) circle (0.05);
\draw (10,0) circle (0.5);

\node at (3,-0.3) {$x_{0}$};
\node at (5,-0.3) {$x_{1}$};
\node at (7,-0.3) {$x_{2}$};
\node at (10,-0.3) {$x_{m}$};

\draw[black,arrows={-Triangle[length=0.18cm,width=0.12cm]}]
($(3,0)+(135:1.5)$) --  ++(-45:0.001);
\draw[black,arrows={-Triangle[length=0.18cm,width=0.12cm]}]
($(3,0)+(-135:1.5)$) --  ++(45:0.001);

\draw[black,arrows={-Triangle[length=0.18cm,width=0.12cm]}]
($(10,0)+(45:1.5)$) --  ++(45:0.001);
\draw[black,arrows={-Triangle[length=0.18cm,width=0.12cm]}]
($(10,0)+(-45:1.5)$) --  ++(-45:0.001);

\draw ($(3,0)+(45:0.5)$) .. controls (3.7,0.7) and (4.3,0.7) .. ($(5,0)+(135:0.5)$);
\draw ($(3,0)+(-45:0.5)$) .. controls (3.7,-0.7) and (4.3,-0.7) .. ($(5,0)+(-135:0.5)$);
\draw ($(7,0)+(45:0.5)$) .. controls (8,1.1) and (9,1.1) .. ($(10,0)+(135:0.5)$);
\draw ($(7,0)+(-45:0.5)$) .. controls (8,-1.1) and (9,-1.1) .. ($(10,0)+(-135:0.5)$);

\draw[black,arrows={-Triangle[length=0.18cm,width=0.12cm]}]
(4.05,0.6) --  ++(0:0.001);
\draw[black,arrows={-Triangle[length=0.18cm,width=0.12cm]}]
(4.05,-0.6) --  ++(0:0.001);
\draw[black,arrows={-Triangle[length=0.18cm,width=0.12cm]}]
(8.6,0.9) --  ++(0:0.001);
\draw[black,arrows={-Triangle[length=0.18cm,width=0.12cm]}]
(8.6,-0.9) --  ++(0:0.001);

\draw[black,arrows={-Triangle[length=0.18cm,width=0.12cm]}]
(3.1,0.5) --  ++(0:0.001);
\draw[black,arrows={-Triangle[length=0.18cm,width=0.12cm]}]
(5.1,0.5) --  ++(0:0.001);
\draw[black,arrows={-Triangle[length=0.18cm,width=0.12cm]}]
(7.1,0.5) --  ++(0:0.001);
\draw[black,arrows={-Triangle[length=0.18cm,width=0.12cm]}]
(10.1,0.5) --  ++(0:0.001);
\end{tikzpicture}
\caption{Jump contours $\Sigma_{R}$ for the RH problem for $R$ with $m=3$ and $p=2$.}
\label{fig:contour for R sp=0}
\end{figure}
From Lemma \ref{lemma: real part of g} and the fact that $P^{(\infty)}$ is independent of $r$ (see \eqref{def of Pinf sp=0}), we infer that the jumps $J_{R} := R_{-}^{-1}R_{+}$ satisfy
\begin{equation}\label{estimates jumps for R lenses sp=0}
J_{R}(z) = P^{(\infty)}(z)S_{-}(z)^{-1}S_{+}(z)P^{(\infty)}(z)^{-1} = I + \bigO(e^{-c |z| r}), \qquad \mbox{as } r \to + \infty,
\end{equation}
uniformly for $z \in \Sigma_{R} \cap (\gamma_{+} \cup  \gamma_{-})$, for a certain $c>0$ independent of $z$ and $r$. Let us orient the boundaries of the disks in the clockwise direction as shown in Figure \ref{fig:contour for R sp=0}. For $z \in \bigcup_{j=0}^{m} \partial \mathcal{D}_{x_{j}}$, from \eqref{matching strong -x_j sp=0}, \eqref{matching strong xp sp=0} and \eqref{matching strong xp-1 sp=0}, we have
\begin{equation}\label{estimates jumps for R disks sp=0}
J_{R}(z) = P^{(\infty)}(z)P^{(x_{j})}(z)^{-1} = I + \bigO \Big(\frac{1}{r} \Big), \qquad \mbox{ as } r \to  +\infty.
\end{equation}
Therefore, $R$ satisfies a small norm RH problem. By standard theory for small norm RH problems \cite{DKMVZ2,DKMVZ1}, $R$ exists for sufficiently large $r$ and satisfies 
\begin{align}
& R(z) = I + \frac{R^{(1)}(z)}{r} + \bigO(r^{-2}), \qquad R^{(1)}(z) = \bigO(1), & & \mbox{ as } r \to  +\infty \label{eq: asymp R inf sp=0} 
\end{align}
uniformly for $z \in \mathbb{C}\setminus \Sigma_{R}$. For any $j \in \{0,\ldots,m\}\setminus \{p-1,p\}$, we see from \eqref{def of Ej sp=0} that some factors $r^{\pm \beta_{j}}$ are present in the entries of $E_{x_{j}}$. Hence, by \eqref{matching strong -x_j sp=0}, some factors of the form $r^{\pm 2\beta_{j}}$ also appear in the entries of $J_{R}$, and therefore
\begin{equation}\label{eq: asymp der beta R inf sp=0}
\partial_{\beta_{j}}R(z) = \frac{\partial_{\beta_{j}}R^{(1)}(z)}{r} + \bigO \Big( \frac{\log r}{r^{2}} \Big), \qquad \partial_{\beta_{j}}R^{(1)}(z) = \bigO(\log r), \qquad \mbox{ as } r \to  +\infty.
\end{equation}
Furthermore, since the asymptotics \eqref{estimates jumps for R lenses sp=0} and \eqref{estimates jumps for R disks sp=0} hold uniformly for $\beta_{0},\ldots,\beta_{p-2},\beta_{p+1},\ldots,\beta_{m}$ in compact subsets of $i \mathbb{R}$, and uniformly in $x_{0},x_{1},\ldots,x_{m}$ in compact subsets of $\mathbb{R}$ such that \eqref{assumption xj delta sp=0} holds, the asymptotics \eqref{eq: asymp R inf sp=0} and \eqref{eq: asymp der beta R inf sp=0} also hold uniformly in $\beta_{0},\ldots,\beta_{p-2},\beta_{p+1},\ldots,\beta_{m},x_{0},\ldots,x_{m}$ in the same way.

\vspace{0.2cm}\hspace{-0.55cm}Now, we compute explicitly $R^{(1)}(x_{p})$, $R^{(1)}(x_{p-1})$ and $R^{(1)}(z)$ for $z \in \mathbb{C}\setminus \bigcup_{j=0}^{m}\mathcal{D}_{x_{j}}$. As in \eqref{integral form for R1 sp neq 0}, $R^{(1)}$ admits the following integral representation
\begin{equation}
R^{(1)}(z) = \frac{1}{2\pi i}\int_{\bigcup_{j=0}^{m}\partial\mathcal{D}_{x_{j}}} \frac{J_{R}^{(1)}(s)}{s-z}ds,
\end{equation}
where $J_{R}^{(1)}$ is defined via the expansion
\begin{equation}
J_{R}(z) = I + \frac{J_{R}^{(1)}(z)}{r} + \bigO(r^{-2}), \qquad J_{R}^{(1)}(z) = \bigO(1), \qquad \mbox{ as } r \to + \infty, \quad z \in \bigcup_{j=0}^{m}\partial\mathcal{D}_{x_{j}}.
\end{equation}
Recall that $J_{R}^{(1)}(z)$ is defined only for $z$ on the boundaries of the disks. However, from the explicit expressions for $J_{R}^{(1)}$ given by \eqref{matching strong -x_j sp=0}, \eqref{matching strong xp sp=0} and \eqref{matching strong xp-1 sp=0}, we see that $J_{R}^{(1)}$ can be analytically continued on $\bigcup_{j=0}^{m} \overline{\mathcal{D}_{x_{j}}}\setminus \{x_{j}\}$, and that  $J_{R}^{(1)}$ has a pole of order 1 at each of the $x_{j}$'s. Since the disks are oriented in the clockwise direction, a direct residue calculation shows that
\begin{align}
& R^{(1)}(z) = \sum_{j=0}^{m} \frac{1}{z-x_{j}}\mbox{Res}(J_{R}^{(1)}(s),s = x_{j}), \qquad \mbox{ for } z \in \mathbb{C}\setminus \bigcup_{j=0}^{m}\mathcal{D}_{x_{j}}, \label{expression for R^1 sp=0} \\
& R^{(1)}(x_{p}) = \sum_{\substack{j=0 \\ j \neq p}}^{m} \frac{1}{x_{p}-x_{j}}\mbox{Res}(J_{R}^{(1)}(s),s = x_{j}) - \mbox{Res} \Big( \frac{J_{R}^{(1)}(s)}{s-x_{p}}, s = x_{p} \Big), \label{expression for R^1 at xp sp=0} \\
& R^{(1)}(x_{p-1}) = \sum_{\substack{j=0 \\ j \neq p-1}}^{m} \frac{1}{x_{p-1}-x_{j}}\mbox{Res}(J_{R}^{(1)}(s),s = x_{j}) - \mbox{Res} \Big( \frac{J_{R}^{(1)}(s)}{s-x_{p-1}}, s = x_{p-1} \Big). \label{expression for R^1 at xp-1 sp=0}
\end{align}
Using \eqref{expansion conformal map xj sp=0} and \eqref{matching strong -x_j sp=0}-\eqref{E_j at x_j sp=0}, for $j \in \{0,\ldots,m\}\setminus \{p-1,p\}$ we obtain
\begin{equation*}
\begin{array}{r c l}
\ds \mbox{Res}\left( J_{R}^{(1)}(s),s=x_{j} \right) & = & \hspace{-0.25cm} \ds \frac{\beta_{j}^{2}}{2i c_{x_{j}}} \widehat{D} \begin{pmatrix}
1 & -1 \\ 1 & 1
\end{pmatrix} \left( \frac{|x_{j}-x_{p}|}{|x_{j}-x_{p-1}|} \right)^{\frac{\sigma_{3}}{4}} N \begin{pmatrix}
-1 & \widetilde{\Lambda}_{j,1} \\ -\widetilde{\Lambda}_{j,2} & 1
\end{pmatrix} N^{-1} \\[0.35cm]
& & \ds \times \left( \frac{|x_{j}-x_{p}|}{|x_{j}-x_{p-1}|} \right)^{-\frac{\sigma_{3}}{4}} \begin{pmatrix}
1 & 1 \\ -1 & 1
\end{pmatrix} \widehat{D}^{-1} \\[0.5cm]
& = & \ds \hspace{-0.25cm} \frac{\beta_{j}^{2}}{4c_{x_{j}}}\widehat{D} \begin{pmatrix}
1 & -1 \\ 1 & 1
\end{pmatrix} \left( \frac{|x_{j}-x_{p}|}{|x_{j}-x_{p-1}|} \right)^{\frac{\sigma_{3}}{4}}
\hspace{-0.1cm}
\begin{pmatrix}
-\widetilde{\Lambda}_{j,1}-\widetilde{\Lambda}_{j,2} & \hspace{-0.25cm}-i(\widetilde{\Lambda}_{j,1}-\widetilde{\Lambda}_{j,2}+2i) \\
-i(\widetilde{\Lambda}_{j,1}-\widetilde{\Lambda}_{j,2}-2i) & \hspace{-0.25cm}\widetilde{\Lambda}_{j,1}+\widetilde{\Lambda}_{j,2}
\end{pmatrix} \\[0.35cm]
& & \times \ds  \left( \frac{|x_{j}-x_{p}|}{|x_{j}-x_{p-1}|} \right)^{-\frac{\sigma_{3}}{4}} \begin{pmatrix}
1 & 1 \\ -1 & 1
\end{pmatrix}\widehat{D}^{-1},
\end{array}
\end{equation*}
where
\begin{equation}\label{def of Lambda tilde}
\widetilde{\Lambda}_{j,1} = \tau(\beta_{j})\Lambda_{j}^{2} \qquad \mbox{ and } \qquad \widetilde{\Lambda}_{j,2} = \tau(-\beta_{j})\Lambda_{j}^{-2}.
\end{equation}
Using \eqref{def of Pinf sp=0}, \eqref{asymp D at xp sp=0}, \eqref{expansion fxp at xp sp=0}, \eqref{matching strong xp sp=0} and \eqref{explicit asymp for branch cut at xp sp=0}, we obtain
\begin{equation}\label{residu xp easy sp=0}
\mbox{Res}\left( J_{R}^{(1)}(s),s=x_{p} \right) = \frac{1}{16}\widehat{D} \begin{pmatrix}
1 & 1 \\ -1 & -1
\end{pmatrix}\widehat{D}^{-1},
\end{equation}
and by \eqref{def of Pinf sp=0}, \eqref{asymp D at xp-1 sp=0}, \eqref{expansion fxp-1 at xp-1 sp=0} and \eqref{matching strong xp-1 sp=0}, we have
\begin{equation}\label{residu xp-1 easy sp=0}
\mbox{Res}\left( J_{R}^{(1)}(s),s=x_{p-1} \right) = \frac{1}{16}\widehat{D} \begin{pmatrix}
-1 & 1 \\ -1 & 1
\end{pmatrix}\widehat{D}^{-1}.
\end{equation}
In the same way as we derived the residues \eqref{residu xp easy sp=0} and \eqref{residu xp-1 easy sp=0}, but with more efforts, we also obtain
\begin{multline}
\mbox{Res} \Big( \frac{J_{R}^{(1)}(s)}{s-x_{p}}, s = x_{p} \Big) = \frac{\widehat{D}}{16(x_{p}-x_{p-1})} \begin{pmatrix}
3 + 16 d_{x_{p}}^{2} & \hspace{-0.4cm}-3 + 16 d_{x_{p}}^{2} + 16 i d_{x_{p}} \\[0.2cm]
3 - 16 d_{x_{p}}^{2} + 16 i d_{x_{p}} & \hspace{-0.4cm}-3 - 16 d_{x_{p}}^{2}
\end{pmatrix}\widehat{D}^{-1}
\end{multline}
and
\begin{multline}
\mbox{Res} \Big( \frac{J_{R}^{(1)}(s)}{s-x_{p-1}}, s = x_{p-1} \Big) = \frac{\widehat{D}}{16(x_{p}-x_{p-1})} \begin{pmatrix}
3 + 16 d_{x_{p-1}}^{2} & \hspace{-0.4cm}3 - 16 d_{x_{p-1}}^{2} + 16 i d_{x_{p-1}} \\[0.2cm]
-3 + 16 d_{x_{p-1}}^{2} + 16 i d_{x_{p-1}} & \hspace{-0.4cm} -3 - 16 d_{x_{p-1}}^{2}
\end{pmatrix}\widehat{D}^{-1}.
\end{multline}
\section{Proof of Theorem \ref{thm:sp=0}}\label{Section: integration sp=0}
We prove Theorem \ref{thm:sp=0} using the same strategy as in Section \ref{Section: integration s1 >0}. First, we use the RH analysis done in Section \ref{Section: Steepest descent with sp=0} to find large $r$ asymptotics for 
\begin{equation}
\partial_{s_{k}} \log F(r\vec{x},\vec{s})= K_{\infty} + \sum_{j=0}^{m}K_{x_{j}}, \qquad k=1,\ldots,p-1,p+1,\ldots,m.
\end{equation}
The above identity was obtained in \eqref{DIFF identity final form general case} and the quantities $K_{\infty}$ and $K_{x_{j}}$ are defined in \eqref{K inf}-\eqref{K xj}. Then, we integrate these asymptotics over the parameters $s_{1},\ldots,s_{p-1},s_{p+1},\ldots,s_{m}$. 

\subsection[]{Large $r$ asymptotics for $\partial_{s_{k}} \log F(r\vec{x},\vec{s})$}

\paragraph{Asymptotics for $K_{\infty}$.} Using \eqref{eq:Sasympinf sp=0}, \eqref{def of R sp=0} and \eqref{eq:Pinf asympinf sp=0}, we obtain
\begin{align*}
& T_{1} = R_{1} + P_{1}^{(\infty)}, 
\end{align*}
where $R_{1}$ is the $z^{-1}$ coefficient in the large $z$ expansion of $R(z)$. Hence, by \eqref{eq: asymp R inf sp=0}, we have
\begin{align*}
& T_{1} = P_{1}^{(\infty)} + \frac{R_{1}^{(1)}}{r} + \bigO(r^{-2}), \qquad \mbox{ as } r \to + \infty,
\end{align*}
where $R_{1}^{(1)}$ is defined through the expansion
\begin{equation}
R^{(1)}(z) = \frac{R_{1}^{(1)}}{z} + \bigO(z^{-2}), \qquad \mbox{ as } z \to \infty.
\end{equation}
Hence, using \eqref{Pinf 1 12 sp=0} and \eqref{expression for R^1 sp=0}, we get
\begin{multline}
T_{1} = \begin{pmatrix}
\frac{i}{8}(x_{p}-x_{p-1})(D_{\infty}^{2}-D_{\infty}^{-2}) & i d_{1} - \frac{1}{8}(x_{p}-x_{p-1})(D_{\infty}^{2}+D_{\infty}^{-2}) \\
-id_{1}- \frac{1}{8}(x_{p}-x_{p-1})(D_{\infty}^{2}+D_{\infty}^{-2}) & -\frac{i}{8}(x_{p}-x_{p-1})(D_{\infty}^{2}-D_{\infty}^{-2})
\end{pmatrix} + \sum_{\substack{j=0 \\ j\neq p-1,p}}^{m}\frac{\beta_{j}^{2}}{4c_{x_{j}}r}\widehat{D} \begin{pmatrix}
1 & -1 \\ 1 & 1
\end{pmatrix} \\ 
\times \left( \frac{|x_{j}-x_{p}|}{|x_{j}-x_{p-1}|} \right)^{\frac{\sigma_{3}}{4}}
\hspace{-0.1cm}
\begin{pmatrix}
-\widetilde{\Lambda}_{j,1}-\widetilde{\Lambda}_{j,2} & \hspace{-0.25cm}-i(\widetilde{\Lambda}_{j,1}-\widetilde{\Lambda}_{j,2}+2i) \\
-i(\widetilde{\Lambda}_{j,1}-\widetilde{\Lambda}_{j,2}-2i) & \hspace{-0.25cm}\widetilde{\Lambda}_{j,1}+\widetilde{\Lambda}_{j,2}
\end{pmatrix} \left( \frac{|x_{j}-x_{p}|}{|x_{j}-x_{p-1}|} \right)^{-\frac{\sigma_{3}}{4}} \begin{pmatrix}
1 & 1 \\ -1 & 1
\end{pmatrix}\widehat{D}^{-1} \\
+\frac{1}{16r}\widehat{D} \begin{pmatrix}
1 & 1 \\ -1 & -1
\end{pmatrix}\widehat{D}^{-1} + 
\frac{1}{16r}\widehat{D} \begin{pmatrix}
-1 & 1 \\ -1 & 1
\end{pmatrix}\widehat{D}^{-1} 
+ \bigO(r^{-2}), \qquad \mbox{as } r \to + \infty,
\end{multline}
which implies, by \eqref{K inf}, \eqref{Kinf in terms of T1 sp=0} and \eqref{eq: asymp der beta R inf sp=0}, that
\begin{multline}\label{K inf asymp sp=0}
K_{\infty} = r \big( \partial_{s_{k}} T_{1,21} - \partial_{s_{k}}T_{1,12} \big) = -2i \partial_{s_{k}}d_{1}r \\ + \sum_{\substack{j=0 \\ j \neq p-1,p}}^{m} \frac{|x_{j}-x_{p-1}| \partial_{s_{k}}\big( \beta_{j}^{2}(\widetilde{\Lambda}_{j,1}-\widetilde{\Lambda}_{j,2}-2i) \big)-|x_{j}-x_{p}| \partial_{s_{k}}\big( \beta_{j}^{2}(\widetilde{\Lambda}_{j,1}-\widetilde{\Lambda}_{j,2}+2i) \big)}{2 i c_{x_{j}} \sqrt{|x_{j}-x_{p-1}|\, |x_{j}-x_{p}|}} + \bigO \Big( \frac{\log r}{r} \Big)
\end{multline} \\[-0.4cm]
as $r \to + \infty$.
\paragraph{Asymptotics for $K_{x_{j}}$ with $j \in \{0,\ldots,p-2,p+1,\ldots,m\}$.} For $z$ outside the lenses and inside $\mathcal{D}_{x_{j}}$, by \eqref{def of S sp=0}, \eqref{def of R sp=0}, and \eqref{lol10 sp=0}, we have
\begin{equation}\label{lol11 2}
T(z) = R(z)E_{x_{j}}(z)\Phi_{\mathrm{HG}}(rf_{x_{j}}(z);\beta_{j})(s_{j}s_{j+1})^{-\frac{\sigma_{3}}{4}}e^{- r g(z)\sigma_{3}},
\end{equation}
and by \eqref{expansion conformal map xj sp=0} and \eqref{model RHP HG in different sector}, we also have
\begin{equation}
\Phi_{\mathrm{HG}}(rf_{x_{j}}(z);\beta_{j}) = \widehat{\Phi}_{\mathrm{HG}}(rf_{x_{j}}(z);\beta_{j}), \qquad \mbox{for } \Im z > 0.
\end{equation}
Using \eqref{def of beta sp=0} and Euler's reflection formula (see e.g. \cite[equation 5.5.3]{NIST}), we note that
\begin{equation}\label{relation Gamma beta_j and s_j sp=0}
\frac{\sin (\pi \beta_{j})}{\pi} = \frac{1}{\Gamma(\beta_{j})\Gamma(1-\beta_{j})} = -\frac{s_{j+1}-s_{j}}{2\pi i \sqrt{s_{j}s_{j+1}}},\qquad j =0,\ldots,p-2,p+1,\ldots,m.
\end{equation}
This identity, combined with \eqref{expansion conformal map xj sp=0} and \eqref{precise asymptotics of Phi HG near 0}, implies that
\begin{equation}\label{lol 2 sp=0}
\Phi_{\mathrm{HG}}(rf_{x_{j}}(z);\beta_{j})(s_{j}s_{j+1})^{-\frac{\sigma_{3}}{4}} = \begin{pmatrix}
\Psi_{j,11} & \Psi_{j,12} \\
\Psi_{j,21} & \Psi_{j,22}
\end{pmatrix} (I + \bigO(z-x_{j})) \begin{pmatrix}
1 & -\frac{s_{j+1}-s_{j}}{2\pi i}\log(r(z - x_{j})) \\
0 & 1
\end{pmatrix} ,
\end{equation}
as $z \to x_{j}$ from $\Im z >0$ and outside the lenses,
where the principal branch is taken for the log and
\begin{align}
& \Psi_{j,11} = \frac{\Gamma(1-\beta_{j})}{(s_{j}s_{j+1})^{\frac{1}{4}}}, \qquad \Psi_{j,12} = \frac{(s_{j}s_{j+1})^{\frac{1}{4}}}{\Gamma(\beta_{j})} \left( \log c_{x_{j}} - \frac{i\pi}{2} + \frac{\Gamma^{\prime}(1-\beta_{j})}{\Gamma(1-\beta_{j})}+2\gamma_{\mathrm{E}} \right), \nonumber \\
& \Psi_{j,21} = \frac{\Gamma(1+\beta_{j})}{(s_{j}s_{j+1})^{\frac{1}{4}}}, \qquad \Psi_{j,22} = \frac{-(s_{j}s_{j+1})^{\frac{1}{4}}}{\Gamma(-\beta_{j})} \left( \log c_{x_{j}} - \frac{i\pi}{2} + \frac{\Gamma^{\prime}(-\beta_{j})}{\Gamma(-\beta_{j})} + 2\gamma_{\mathrm{E}} \right), \label{Psi j entries}
\end{align}
and using $\Gamma(1+z)=z\Gamma(z)$ and \eqref{relation Gamma beta_j and s_j sp=0}, we verify that
\begin{equation}\label{Psi_j first column connection formula}
\Psi_{j,11}\Psi_{j,21} = -\beta_{j} \frac{2\pi i}{s_{j+1}-s_{j}}, \qquad j =0,\ldots,p-2,p+1,\ldots,m.
\end{equation}
From \eqref{Asymp of Phi near x_j}, \eqref{def of T sp=0}, \eqref{lol11 2} and \eqref{lol 2 sp=0}, we get
\begin{equation*}
G_{j}(rx_{j};r\vec{x},\vec{s}) = \begin{pmatrix}
\cos \Big( \frac{r}{2}(x_{p-1}+x_{p}) \Big) & -\sin \Big( \frac{r}{2}(x_{p-1}+x_{p}) \Big) \\
\sin \Big( \frac{r}{2}(x_{p-1}+x_{p}) \Big) & \cos \Big( \frac{r}{2}(x_{p-1}+x_{p}) \Big)
\end{pmatrix}R(x_{j})E_{x_{j}}(x_{j})\begin{pmatrix}
\Psi_{j,11} & \Psi_{j,12} \\ \Psi_{j,21} & \Psi_{j,22}
\end{pmatrix}.
\end{equation*}
Also, from \eqref{def of Pinf sp=0}, we see that
\begin{align}
\widehat{D}_{11} \partial_{s_{k}}\widehat{D}_{21}-\widehat{D}_{21} \partial_{s_{k}}\widehat{D}_{11} = i \partial_{s_{k}} \log D_{\infty}.
\end{align}
Therefore, using \eqref{eq: asymp R inf sp=0}, \eqref{eq: asymp der beta R inf sp=0} and $\det E_{x_{j}}(x_{j})=1$ in the definition of $K_{x_{j}}$ given by \eqref{K xj}, we obtain after a long calculation that
\begin{multline}\label{K xj part 1 asymp sp=0}
\hspace{-0.4cm}\sum_{\substack{j=0 \\ j \neq p-1,p}}^{m} K_{x_{j}} = \sum_{\substack{j=0 \\ j \neq p-1,p}}^{m} -\frac{s_{j+1}-s_j}{2\pi i} \Big( \Psi_{j,11} \partial_{s_{k}}\Psi_{j,21} - \Psi_{j,21}\partial_{s_{k}}\Psi_{j,11} \Big) - \sum_{\substack{j=0 \\ j \neq p-1,p}}^{m} 2\beta_{j}  \partial_{s_{k}} \log \Lambda_{j} + \frac{i}{2} \partial_{s_{k}} \log D_{\infty} \\
\times \sum_{\substack{j=0 \\ j \neq p-1,p}}^{m} - \frac{s_{j+1}-s_{j}}{2\pi i}\left( \frac{\sqrt{|x_{j}-x_{p}|}}{\sqrt{|x_{j}-x_{p-1}|}}(\Lambda_{j} \Psi_{j,11} + i \Lambda_{j}^{-1} \Psi_{j,21})^{2}-\frac{\sqrt{|x_{j}-x_{p-1}|}}{\sqrt{|x_{j}-x_{p}|}}(\Lambda_{j} \Psi_{j,11} - i \Lambda_{j}^{-1} \Psi_{j,21})^{2} \right) \\
 + \bigO \Big( \frac{\log r}{r} \Big), \qquad \mbox{ as } r \to + \infty.
\end{multline}
Using \eqref{relation Gamma beta_j and s_j sp=0}, \eqref{Psi j entries}, \eqref{Psi_j first column connection formula} and the definitions \eqref{def of Lambda tilde} of $\widetilde{\Lambda}_{j,1}$ and $\widetilde{\Lambda}_{j,2}$, for $j = 1,\ldots,p-2,p+1,\ldots,m$ we get
\begin{align*}
& - \frac{s_{j+1}-s_{j}}{2\pi i}(\Lambda_{j} \Psi_{j,11} + i \Lambda_{j}^{-1} \Psi_{j,21})^{2} = \beta_{j}^{2} (\widetilde{\Lambda}_{j,1}+\widetilde{\Lambda}_{j,2}) + 2 i \beta_{j}, \\
& - \frac{s_{j+1}-s_{j}}{2\pi i}(\Lambda_{j} \Psi_{j,11} - i \Lambda_{j}^{-1} \Psi_{j,21})^{2} = \beta_{j}^{2} (\widetilde{\Lambda}_{j,1}+\widetilde{\Lambda}_{j,2}) - 2 i \beta_{j}, \\
& -\frac{s_{j+1}-s_j}{2\pi i} \Big( \Psi_{j,11} \partial_{s_{k}}\Psi_{j,21} - \Psi_{j,21}\partial_{s_{k}}\Psi_{j,11} \Big) = \beta_{j} \partial_{s_{k}} \log \frac{\Gamma(1+\beta_{j})}{\Gamma(1-\beta_{j})}.
\end{align*}
Substituting the above identities in \eqref{K xj part 1 asymp sp=0}, we finally arrive at
\begin{multline}\label{K xj part 2 asymp sp=0}
\hspace{-0.4cm}\sum_{\substack{j=0 \\ j \neq p-1,p}}^{m} K_{x_{j}} = \sum_{\substack{j=0 \\ j \neq p-1,p}}^{m} \beta_{j} \partial_{s_{k}} \log \frac{\Gamma(1+\beta_{j})}{\Gamma(1-\beta_{j})} - \sum_{\substack{j=0 \\ j \neq p-1,p}}^{m} 2\beta_{j}  \partial_{s_{k}} \log \Lambda_{j} \\ + \frac{i}{2} \partial_{s_{k}} \log D_{\infty} \sum_{\substack{j=0 \\ j \neq p-1,p}}^{m} \beta_{j}^{2} (\widetilde{\Lambda}_{j,1}+\widetilde{\Lambda}_{j,2}) \left( \frac{\sqrt{|x_{j}-x_{p}|}}{\sqrt{|x_{j}-x_{p-1}|}} - \frac{\sqrt{|x_{j}-x_{p-1}|}}{\sqrt{|x_{j}-x_{p}|}} \right)  \\
+ \frac{i}{2} \partial_{s_{k}} \log D_{\infty} \sum_{\substack{j=0 \\ j \neq p-1,p}}^{m} 2 i \beta_{j} \left( \frac{\sqrt{|x_{j}-x_{p}|}}{\sqrt{|x_{j}-x_{p-1}|}} + \frac{\sqrt{|x_{j}-x_{p-1}|}}{\sqrt{|x_{j}-x_{p}|}} \right)  + \bigO \Big( \frac{\log r}{r} \Big),
\end{multline}
as $r \to + \infty$.
\paragraph{Asymptotics for $K_{x_{p}}$.} It follows from \eqref{def of S sp=0}, \eqref{def of R sp=0}, and \eqref{def of P^xp sp=0} that for $z \in \mathcal{D}_{x_{p}}$, $z$ outside the lenses, we have
\begin{equation}\label{expr for T near xp form 1 sp=0}
T(z) = R(z)E_{x_{p}}(z) \sigma_{3} \Phi_{\mathrm{Be}}(-r^{2}f_{x_{p}}(z))\sigma_{3} \sqrt{s_{p+1}}^{-\sigma_{3}}e^{-rg(z)\sigma_{3}}.
\end{equation}
Using \eqref{precise asymptotics of Phi Bessel near 0} and \eqref{expansion fxp at xp sp=0}, we obtain
\begin{equation*}
\sigma_{3}\Phi_{\mathrm{Be}}(-r^{2}f_{x_{p}}(z))\sqrt{s_{p+1}}^{-\sigma_{3}}\sigma_{3} = \begin{pmatrix}
\Psi_{p,11} & \Psi_{p,12} \\
\Psi_{p,21} & \Psi_{p,22}
\end{pmatrix}(I + \bigO(z-x_{p})) \begin{pmatrix}
1 & -\frac{s_{p+1}}{2\pi i} \log (r(z-x_{p})) \\
0 & 1 
\end{pmatrix}
\end{equation*}
as $z \to x_{p}$ from $\Im z > 0$ and outside the lenses, where
\begin{align}
& \Psi_{p,11} = s_{p+1}^{-1/2}, \qquad \quad \Psi_{p,12} = -s_{p+1}^{1/2}\Big(\frac{\gamma_{\mathrm{E}}}{\pi i} + \frac{\log(c_{x_{p}}^{2}r)-\pi i}{2\pi i}\Big), \nonumber \\
& \Psi_{p,21} = 0, \qquad \qquad \hspace{0.3cm} \Psi_{p,22} = s_{p+1}^{1/2}. \label{Psi p entries}
\end{align}
On the other hand, using \eqref{Asymp of Phi near x_j} and \eqref{def of T sp=0}, as $z \to x_{p}$, $\Im z >0$, we also have
\begin{equation}\label{expr for T near xp form 2 sp=0}
T(z) = \begin{pmatrix}
\cos \Big( \frac{r}{2}(x_{p-1}+x_{p}) \Big) & \sin \Big( \frac{r}{2}(x_{p-1}+x_{p}) \Big) \\
- \sin \Big( \frac{r}{2}(x_{p-1}+x_{p}) \Big) & \cos \Big( \frac{r}{2}(x_{p-1}+x_{p}) \Big)
\end{pmatrix} \hspace{-0.1cm} G_{p}(rz; r \vec{x},\vec{s}) \begin{pmatrix}
1 & \hspace{-0.1cm}- \frac{s_{p+1}}{2\pi i} \log(r(z-x_{p})) \\
0 & \hspace{-0.1cm}1
\end{pmatrix}
e^{-rg(z)\sigma_{3}}.
\end{equation}
By combining \eqref{expr for T near xp form 1 sp=0} with \eqref{expr for T near xp form 2 sp=0}, we arrive at
\begin{equation*}
G_{p}(rx_{p};r\vec{x},\vec{s}) = \begin{pmatrix}
\cos \Big( \frac{r}{2}(x_{p-1}+x_{p}) \Big) & -\sin \Big( \frac{r}{2}(x_{p-1}+x_{p}) \Big) \\
\sin \Big( \frac{r}{2}(x_{p-1}+x_{p}) \Big) & \cos \Big( \frac{r}{2}(x_{p-1}+x_{p}) \Big)
\end{pmatrix}R(x_{p})E_{x_{p}}(x_{p})\begin{pmatrix}
\Psi_{p,11} & \Psi_{p,12} \\
\Psi_{p,21} & \Psi_{p,22}
\end{pmatrix}.
\end{equation*}
From the definition \eqref{K xj} of $K_{x_{p}}$ and the explicit expressions for $E_{x_{p}}(x_{p})$ and $R^{(1)}(x_{p})$ given by \eqref{E0 at 0 sp=0} and \eqref{expression for R^1 at xp sp=0}, we find after a computation that
\begin{multline}\label{K xp asymp sp=0}
K_{x_{p}} = - \frac{ir}{2}(x_{p}-x_{p-1})\partial_{s_{k}} \log D_{\infty}  + \partial_{s_{k}} \log D_{\infty} \bigg( d_{x_{p}} + \sum_{\substack{j=0 \\ j \neq p-1,p}}^{m} \frac{x_{p}-x_{p-1}}{2ic_{x_{j}}(x_{p}-x_{j})}\beta_{j}^{2}(\widetilde{\Lambda}_{j,1}+\widetilde{\Lambda}_{j,2}) \bigg) \\
+ \sum_{\substack{j=0 \\ j \neq p-1,p}}^{m} \frac{\sqrt{|x_{j}-x_{p}|}(x_{p}-x_{p-1})}{4i \sqrt{|x_{j}-x_{p-1}|}(x_{p}-x_{j})c_{x_{j}}}\partial_{s_{k}}\big( \beta_{j}^{2}(\widetilde{\Lambda}_{j,1}-\widetilde{\Lambda}_{j,2}+2i) \big)
+ \bigO\Big( \frac{\log r}{r} \Big)
\end{multline}
as $r \to + \infty$.
\paragraph{Asymptotics for $K_{x_{p-1}}$.} For $z$ outside the lenses and inside $\mathcal{D}_{x_{p-1}}$, we deduce from \eqref{def of S sp=0}, \eqref{def of P^xp-1 sp=0} and \eqref{def of R sp=0} that
\begin{equation}\label{expr for T near xp-1 form 1 sp=0}
T(z) = R(z)E_{x_{p-1}}(z) \Phi_{\mathrm{Be}}(r^{2}f_{x_{p-1}}(z)) \sqrt{s_{p-1}}^{-\sigma_{3}}e^{-rg(z)\sigma_{3}}.
\end{equation}
Also, using \eqref{precise asymptotics of Phi Bessel near 0} and \eqref{expansion fxp-1 at xp-1 sp=0}, we get
\begin{equation*}
\Phi_{\mathrm{Be}}(r^{2}f_{x_{p-1}}(z))\sqrt{s_{p-1}}^{-\sigma_{3}} = \begin{pmatrix}
\Psi_{p-1,11} & \Psi_{p-1,12} \\
\Psi_{p-1,21} & \Psi_{p-1,22}
\end{pmatrix}(I + \bigO(z-x_{p-1})) \begin{pmatrix}
1 & \frac{s_{p-1}}{2\pi i} \log (r(z-x_{p-1})) \\
0 & 1 
\end{pmatrix},
\end{equation*}
as $z \to x_{p-1}$ from $\Im z > 0$ and outside the lenses, where
\begin{align}
& \Psi_{p-1,11} = s_{p-1}^{-1/2}, \qquad \quad \Psi_{p-1,12} = s_{p-1}^{1/2}\Big(\frac{\gamma_{\mathrm{E}}}{\pi i} + \frac{\log(c_{x_{p-1}}^{2}r)}{2\pi i}\Big), \nonumber \\
& \Psi_{p-1,21} = 0, \qquad \qquad \hspace{0.3cm} \Psi_{p-1,22} = s_{p-1}^{1/2}. \label{Psi p-1 entries}
\end{align}
On the other hand, using \eqref{Asymp of Phi near x_j} and \eqref{def of T sp=0}, as $z \to x_{p-1}$ from $\Im z >0$ we also have that
\begin{multline}
T(z) = \begin{pmatrix}
\cos \Big( \frac{r}{2}(x_{p-1}+x_{p}) \Big) & \sin \Big( \frac{r}{2}(x_{p-1}+x_{p}) \Big) \\
- \sin \Big( \frac{r}{2}(x_{p-1}+x_{p}) \Big) & \cos \Big( \frac{r}{2}(x_{p-1}+x_{p}) \Big)
\end{pmatrix} \nonumber \\
 \times  G_{p-1}(rz; r \vec{x},\vec{s}) \begin{pmatrix}
1 & \hspace{-0.1cm}\frac{s_{p-1}}{2\pi i} \log(r(z-x_{p-1})) \\
0 & \hspace{-0.1cm}1
\end{pmatrix}
e^{-rg(z)\sigma_{3}}. \label{expr for T near xp-1 form 2 sp=0}
\end{multline}
Combining \eqref{expr for T near xp-1 form 1 sp=0} with \eqref{expr for T near xp-1 form 2 sp=0}, we arrive at
\begin{multline}
G_{p-1}(rx_{p-1};r\vec{x},\vec{s}) = \begin{pmatrix}
\cos \Big( \frac{r}{2}(x_{p-1}+x_{p}) \Big) & -\sin \Big( \frac{r}{2}(x_{p-1}+x_{p}) \Big) \\
\sin \Big( \frac{r}{2}(x_{p-1}+x_{p}) \Big) & \cos \Big( \frac{r}{2}(x_{p-1}+x_{p}) \Big)
\end{pmatrix} \\ \times R(x_{p-1})E_{x_{p-1}}(x_{p-1})\begin{pmatrix}
\Psi_{p-1,11} & \Psi_{p-1,12} \\
\Psi_{p-1,21} & \Psi_{p-1,22}
\end{pmatrix}.
\end{multline}
Using \eqref{K xj} and the explicit expressions for $E_{x_{p-1}}(x_{p-1})$ and $R^{(1)}(x_{p-1})$ given by \eqref{E xp-1 at xp-1 sp=0} and \eqref{expression for R^1 at xp-1 sp=0}, we find after a computation that
\begin{multline}\label{K xp-1 asymp sp=0}
K_{x_{p-1}} = \frac{ir}{2}(x_{p}-x_{p-1})\partial_{s_{k}} \log D_{\infty} + \partial_{s_{k}} \log D_{\infty} \bigg( d_{x_{p-1}} + \sum_{\substack{j=0 \\ j \neq p-1,p}}^{m} \frac{x_{p}-x_{p-1}}{2ic_{x_{j}}(x_{p-1}-x_{j})}\beta_{j}^{2}(\widetilde{\Lambda}_{j,1}+\widetilde{\Lambda}_{j,2}) \bigg) \\
+ \sum_{\substack{j=0 \\ j \neq p-1,p}}^{m} \frac{\sqrt{|x_{j}-x_{p-1}|}(x_{p}-x_{p-1})}{4i \sqrt{|x_{j}-x_{p}|}(x_{p-1}-x_{j})c_{x_{j}}}\partial_{s_{k}}\big( \beta_{j}^{2}(\widetilde{\Lambda}_{j,1}-\widetilde{\Lambda}_{j,2}-2i) \big)
+ \bigO\Big( \frac{\log r}{r} \Big)
\end{multline}
as $r \to + \infty$.
\paragraph{Asymptotics for $\partial_{s_{k}} \log F(r\vec{x},\vec{s})$.} After substituting the explicit expression \eqref{expansion conformal map xj sp=0} for $c_{x_{j}}$, $j = 0,\ldots,p-2,p+1,\ldots,m$ into \eqref{K inf asymp sp=0}, \eqref{K xj part 1 asymp sp=0}, \eqref{K xp asymp sp=0} and \eqref{K xp-1 asymp sp=0} and simplifying, we obtain the following asymptotics as $r \to + \infty$:
\begin{align*}
& K_{\infty} = -2 i \partial_{s_{k}}d_{1} \, r + \sum_{\substack{j=0 \\ j \neq p-1,p}}^{m} \left( \frac{1}{2i}\frac{|x_{j}-x_{p-1}|-|x_{j}-x_{p}|}{|x_{p-1}+x_{p}-2x_{j}|} \partial_{s_{k}} \big( \beta_{j}^{2}(\widetilde{\Lambda}_{j,1}-\widetilde{\Lambda}_{j,2}) \big) - \partial_{s_{k}}(\beta_{j}^{2}) \right) + \bigO \Big( \frac{\log r}{r} \Big), \\
& \sum_{\substack{j=0 \\ j \neq p-1,p}}^{m} \hspace{-0.25cm} K_{x_{j}} = - \partial_{s_{k}} \log D_{\infty} \hspace{-0.25cm}\sum_{\substack{j=0 \\ j \neq p-1,p}}^{m} \hspace{-0.25cm} \left( \frac{|x_{j}-x_{p}|-|x_{j}-x_{p-1}|}{\sqrt{|x_{j}-x_{p}| \, |x_{j}-x_{p-1}|}} \frac{\beta_{j}^{2}}{2i} (\widetilde{\Lambda}_{j,1}+\widetilde{\Lambda}_{j,2}) + \beta_{j}\frac{|x_{j}-x_{p}|+|x_{j}-x_{p-1}|}{\sqrt{|x_{j}-x_{p}| \, |x_{j}-x_{p-1}|}}  \right)  \\
& \hspace{2cm} + \sum_{\substack{j=0 \\ j \neq p-1,p}}^{m} \left( \beta_{j} \partial_{s_{k}} \log \frac{\Gamma(1+\beta_{j})}{\Gamma(1- \beta_{j}} - 2 \beta_{j} \partial_{s_{k}} \log \Lambda_{j} \right) + \bigO \Big( \frac{\log r}{r} \Big), \\
& K_{x_{p}} + K_{x_{p-1}} = \partial_{s_{k}} \log D_{\infty} \Bigg( d_{x_{p-1}}+d_{x_{p}} + \sum_{\substack{j=0 \\ j \neq p-1,p}}^{m} \hspace{-0.25cm}  \frac{|x_{j}-x_{p}|-|x_{j}-x_{p-1}|}{\sqrt{|x_{j}-x_{p}| \, |x_{j}-x_{p-1}|}} \frac{\beta_{j}^{2}}{2i} (\widetilde{\Lambda}_{j,1}+\widetilde{\Lambda}_{j,2}) \Bigg) \\
& \hspace{2.5cm}+ \sum_{\substack{j=0 \\ j \neq p-1,p}}^{m}\frac{x_{p}-x_{p-1}}{2i(x_{p}+x_{p-1}-2x_{j})}\partial_{s_{k}}\big( \beta_{j}^{2}(\widetilde{\Lambda}_{j,1}-\widetilde{\Lambda}_{j,2}) \big) + \bigO \Big( \frac{\log r}{r} \Big).
\end{align*}
Also, from the definitions of $d_{x_{p}}$ and $d_{x_{p-1}}$ given by \eqref{def of dxp} and \eqref{def of dxp-1}, we have
\begin{equation}
d_{x_{p-1}}+d_{x_{p}} = \sum_{\substack{j=0 \\ j \neq p-1,p}}^{m} \beta_{j}\frac{|x_{j}-x_{p}|+|x_{j}-x_{p-1}|}{\sqrt{|x_{j}-x_{p}| \, |x_{j}-x_{p-1}|}}.
\end{equation}
Therefore, summing the above asymptotics and simplifying, we obtain
\begin{multline}\label{lol12}
\partial_{s_{k}} \log F(r\vec{x},\vec{s}) = -2 i \partial_{s_{k}}d_{1} \, r + \hspace{-0.15cm} \sum_{\substack{j=0 \\ j \neq p-1,p}}^{m} \hspace{-0.15cm} \left( \beta_{j} \partial_{s_{k}} \log \frac{\Gamma(1+\beta_{j})}{\Gamma(1- \beta_{j}} - 2 \beta_{j} \partial_{s_{k}} \log \Lambda_{j} - \partial_{s_{k}}(\beta_{j}^{2}) \right) + \bigO\Big( \frac{\log r}{r} \Big),
\end{multline} \\[-0.35cm]
as $r \to + \infty$. We also note from \eqref{def of Lambda sp=0} that
\begin{equation}\label{lol13}
\partial_{s_{k}} \log \Lambda_{j} = \partial_{s_{k}}(\beta_{j}) \log \left( \frac{4 \sqrt{|x_{j}-x_{p}| \, |x_{j}-x_{p-1}|}|2x_{j}-x_{p}-x_{p-1}|r}{x_{p}-x_{p-1}} \right) + \sum_{\substack{\ell = 0 \\ \ell \neq j,p-1,p}}^{m} \partial_{s_{k}}(\beta_{\ell})\log T_{\ell,j}.
\end{equation}\\[-0.35cm]
It is clear from \eqref{def of beta sp=0} that there is a one-to-one correspondence between 
\begin{align*}
\vec{s}=(s_{1},\ldots,s_{p-1},0,s_{p+1},\ldots,s_{m}) \in (\mathbb{R}^{+})^{p-1}\times \{0\}\times (\mathbb{R}^{+})^{m-p} 
\end{align*}
and $\vec{\beta} := (\beta_{0},\ldots,\beta_{p-2},\beta_{p+1},\ldots,\beta_{m}) \in (i\mathbb{R})^{m-1}$. Let us define $\widetilde{F}(r\vec{x},\vec{\beta}) := F(r \vec{x},\vec{s})$. By substituting \eqref{lol13} in \eqref{lol12} and then writing the derivatives with respect to $\beta_{k}$ instead of $s_{k}$, we obtain
\begin{multline}\label{diff identity asymptotics sp=0}
\partial_{\beta_{k}} \log \widetilde{F}(r\vec{x},\vec{\beta}) = -2 i \partial_{\beta_{k}}d_{1} \, r + \beta_{k} \partial_{\beta_{k}} \log \frac{\Gamma(1+\beta_{k})}{\Gamma(1- \beta_{k}} - 2 \beta_{k} - \sum_{\substack{j=0 \\ j \neq k,p-1,p}}^{m} 2 \beta_{j} \log T_{k,j} \\ -2 \beta_{k} \log \left( \frac{4 \sqrt{|x_{k}-x_{p}| \, |x_{k}-x_{p-1}|}|2x_{k}-x_{p}-x_{p-1}|r}{x_{p}-x_{p-1}} \right)
+ \bigO\Big( \frac{\log r}{r} \Big), \quad \mbox{as } r \to + \infty.
\end{multline}
It follows from the analysis of Subsection \ref{subsection Small norm sp=0} that the asymptotics \eqref{diff identity asymptotics sp=0} are valid uniformly for $\beta_{0},\ldots,\beta_{p-2},\beta_{p+1},\ldots,\beta_{m}$ in compact subsets of $i \mathbb{R}$, and uniformly in $x_{0},\ldots,x_{m}$ in compact subsets of $\mathbb{R}$ such that \eqref{assumption xj delta sp=0} holds.
\subsection{Integration of the differential identity}
For convenience, we define $\vec{\beta}_{j} \in (i \mathbb{R})^{m-1}$ by
\begin{align*}
\vec{\beta}_{j} = \begin{cases}
(\beta_{0},\ldots,\beta_{j},0,\ldots,0), & \mbox{if } j \in \{0,\ldots,p-2\}, \\
(\beta_{0},\ldots,\beta_{p-2},\beta_{p+1},\ldots,\beta_{j},0,\ldots,0), & \mbox{if } j \in \{p+1,\ldots,m\}.
\end{cases}
\end{align*}
For $k=0$ and $\beta_{1}=\ldots=\beta_{p-2}=\beta_{p+1}=\ldots=\beta_{m} = 0$, the asymptotics \eqref{diff identity asymptotics sp=0} are as follows:
\begin{multline}\label{asymp diff k=0 sp=0}
\partial_{\beta_{0}} \log \widetilde{F}(r\vec{x},\vec{\beta}_{0}) = -2 i \sqrt{x_{p}-x_{0}}\sqrt{x_{p-1}-x_{0}} \, r + \beta_{0} \partial_{\beta_{0}} \log \frac{\Gamma(1+\beta_{0})}{\Gamma(1- \beta_{0}} - 2 \beta_{0} \\ -2 \beta_{0} \log \left( \frac{4 \sqrt{|x_{p}-x_{0}| \, |x_{p-1}-x_{0}|}(x_{p}+x_{p-1}-2x_{0})r}{x_{p}-x_{p-1}} \right)
+ \bigO\Big( \frac{\log r}{r} \Big), \quad \mbox{as } r \to + \infty,
\end{multline}
where we have used the definition \eqref{def of d1} of $d_{1}$. Since these asymptotics are uniform for $\beta_{0}$ in compact subsets of $i \mathbb{R}$, we can integrate \eqref{asymp diff k=0 sp=0} from $\beta_{0} = 0$ to an arbitrary $\beta_{0} \in i \mathbb{R}$ without worsening the order of the error term. Recalling from \eqref{final expression for I ell} that
\begin{equation}\label{int gamma k=0 sp=0}
\int_{0}^{\beta_{0}} x \partial_{x} \log \frac{\Gamma(1+x)}{\Gamma(1-x)}dx = \beta_{0}^{2} + \log \big( G(1+\beta_{0})G(1-\beta_{0}) \big),
\end{equation}
an integration of \eqref{asymp diff k=0 sp=0} yields
\begin{multline}
\log \frac{\widetilde{F}(r\vec{x},\vec{\beta}_{0})}{\widetilde{F}(r\vec{x},\vec{0})} = -2 i \beta_{0} \sqrt{x_{p}-x_{0}}\sqrt{x_{p-1}-x_{0}} \, r + \log \big( G(1+\beta_{0})G(1-\beta_{0}) \big) \\ - \beta_{0}^{2} \log \left( \frac{4 \sqrt{|x_{p}-x_{0}| \, |x_{p-1}-x_{0}|}(x_{p}+x_{p-1}-2x_{0})r}{x_{p}-x_{p-1}} \right)
+ \bigO\Big( \frac{\log r}{r} \Big), \quad \mbox{as } r \to +\infty,
\end{multline}
where $\vec{0} = (0,\ldots,0)$.
In a similar way, we integrate successively in the variables $\beta_{1},\ldots,\beta_{p-2}$. At the last step, we use \eqref{diff identity asymptotics sp=0} with $k=p-2$, and with $\beta_{0},\ldots,\beta_{p-3}$ fixed but arbitrary:
\begin{multline}\label{asymp diff k=p-2 sp=0}
\partial_{\beta_{p-2}} \log \widetilde{F}(r\vec{x},\vec{\beta}_{p-2}) = -2 i \sqrt{x_{p}-x_{p-2}}\sqrt{x_{p-1}-x_{p-2}} \, r + \beta_{p-2} \partial_{\beta_{p-2}} \log \frac{\Gamma(1+\beta_{p-2})}{\Gamma(1- \beta_{p-2})} - 2 \beta_{p-2} \\ - \sum_{\substack{j=0}}^{p-3} 2 \beta_{j} \log T_{p-2,j}-2 \beta_{p-2} \log \left( \frac{4 \sqrt{|x_{p}-x_{p-2}| \, |x_{p-1}-x_{p-2}|}(x_{p}+x_{p-1}-2x_{p-2})r}{x_{p}-x_{p-1}} \right)
+ \bigO\Big( \frac{\log r}{r} \Big),
\end{multline}
as $r \to + \infty$. Since the above asymptotics are uniform for $\beta_{p-2}$ in compact subsets of $i\mathbb{R}$, an integration over $\beta_{p-2}$ from $\beta_{p-2} = 0$ to an arbitrary $\beta_{p-2} \in i \mathbb{R}$ let the order of the error term unchanged, and using again the formula \eqref{int gamma k=0 sp=0} (with $\beta_{0}$ now replaced by $\beta_{p-2}$), we obtain 
\begin{multline}\label{asymp diff k=p-2 int sp=0}
\log \frac{\widetilde{F}(r\vec{x},\vec{\beta}_{p-2})}{\widetilde{F}(r\vec{x},\vec{\beta}_{p-3})} = -2 i \beta_{p-2} \sqrt{x_{p}-x_{p-2}}\sqrt{x_{p-1}-x_{p-2}} \, r + \log \big( G(1+\beta_{p-2})G(1-\beta_{p-2}) \big) \\ - \sum_{\substack{j=0}}^{p-3} 2 \beta_{j}\beta_{p-2} \log T_{p-2,j}-\beta_{p-2}^{2} \log \left( \frac{4 \sqrt{|x_{p}-x_{p-2}| \, |x_{p-1}-x_{p-2}|}(x_{p}+x_{p-1}-2x_{p-2})r}{x_{p}-x_{p-1}} \right)
+ \bigO\Big( \frac{\log r}{r} \Big),
\end{multline}
as $r \to + \infty$. The successive integrations in $\beta_{p+1},\ldots,\beta_{m}$ can be done similarly. At the last step, we use \eqref{diff identity asymptotics sp=0} with $k=m$ and $\vec{\beta}_{m-1}$ arbitrary but fixed:
\begin{multline}\label{asymp diff k=m sp=0}
\partial_{\beta_{m}} \log \widetilde{F}(r\vec{x},\vec{\beta}_{m}) = 2 i \sqrt{x_{m}-x_{p}}\sqrt{x_{m}-x_{p-1}} \, r + \beta_{m} \partial_{\beta_{m}} \log \frac{\Gamma(1+\beta_{m})}{\Gamma(1- \beta_{m}} - 2 \beta_{m} \\ - \sum_{\substack{j=0 \\ j \neq p-1,p}}^{m-1} 2 \beta_{j} \log T_{m,j}-2 \beta_{m} \log \left( \frac{4 \sqrt{|x_{p}-x_{m}| \, |x_{p-1}-x_{m}|}|x_{p}+x_{p-1}-2x_{m}|r}{x_{p}-x_{p-1}} \right)
+ \bigO\Big( \frac{\log r}{r} \Big),
\end{multline}
as $r \to + \infty$. After an integration of the above asymptotics from $\beta_{m} = 0$ to an arbitrary $\beta_{m} \in i \mathbb{R}$, we obtain an asymptotic formula similar to \eqref{asymp diff k=p-2 int sp=0}. Finally, summing all the successive asymptotic formulas for the ratios
\begin{equation*}
\log \frac{\widetilde{F}(r\vec{x},\vec{\beta}_{0})}{\widetilde{F}(r\vec{x},\vec{0})}, \, \log \frac{\widetilde{F}(r\vec{x},\vec{\beta}_{1})}{\widetilde{F}(r\vec{x},\vec{\beta}_{0})}, \ldots, \log \frac{\widetilde{F}(r\vec{x},\vec{\beta}_{p-2})}{\widetilde{F}(r\vec{x},\vec{\beta}_{p-3})}, \log \frac{\widetilde{F}(r\vec{x},\vec{\beta}_{p+1})}{\widetilde{F}(r\vec{x},\vec{\beta}_{p-2})}, \ldots, \log \frac{\widetilde{F}(r\vec{x},\vec{\beta})}{\widetilde{F}(r\vec{x},\vec{\beta}_{m-1})},
\end{equation*}
we obtain
\begin{multline}\label{lol14}
\log \frac{\widetilde{F}(r\vec{x},\vec{\beta})}{\widetilde{F}(r\vec{x},\vec{0})} = -2 i d_{1} \, r + \sum_{\substack{j=0 \\ j \neq p-1,p}}^{m} \log \big( G(1+\beta_{j})G(1-\beta_{j}) \big) - 2 \sum_{\substack{0 \leq j < k \leq m \\ j,k \neq p-1,p}} \beta_{j} \beta_{k} \log T_{k,j} \\ - \sum_{\substack{j=0 \\ j \neq p-1,p}}^{m} \beta_{j}^{2} \log \left( \frac{4 \sqrt{|x_{j}-x_{p}| \, |x_{j}-x_{p-1}|}|2x_{j}-x_{p}-x_{p-1}|r}{x_{p}-x_{p-1}} \right)
+ \bigO\Big( \frac{\log r}{r} \Big)
\end{multline}
as $r \to + \infty$. Note from \eqref{def of beta in thm sp=0} and \eqref{def of beta sp=0} that $u_{j} = 2\pi i \beta_{j}$. We obtain \eqref{explicit asymp for F in the case where sp=0 in thm} after substituting in \eqref{lol14} the known large $r$ asymptotics of $\widetilde{F}(r\vec{x},\vec{0}) = F((rx_{p-1},rx_{p}),0)$ given by \eqref{known result m=1 s1 = 0}. This finishes the proof of Theorem \ref{thm:sp=0}.
\appendix
\section{Model RH problems}\label{Section:Appendix}
In this section, we recall two well-known RH problems.
\subsection{Bessel model RH problem}\label{subsection:Model Bessel}
\begin{itemize}
\item[(a)] $\Phi_{\mathrm{Be}} : \mathbb{C} \setminus \Sigma_{\mathrm{Be}} \to \mathbb{C}^{2\times 2}$ is analytic, where
$\Sigma_{\mathrm{Be}}$ is shown in Figure \ref{figBessel}.
\item[(b)] $\Phi_{\mathrm{Be}}$ satisfies the jump conditions
\begin{equation}\label{Jump for P_Be}
\begin{array}{l l} 
\Phi_{\mathrm{Be},+}(z) = \Phi_{\mathrm{Be},-}(z) \begin{pmatrix}
0 & 1 \\ -1 & 0
\end{pmatrix}, & z \in \mathbb{R}^{-}, \\

\Phi_{\mathrm{Be},+}(z) = \Phi_{\mathrm{Be},-}(z) \begin{pmatrix}
1 & 0 \\ 1 & 1
\end{pmatrix}, & z \in e^{ \frac{2\pi i}{3} }  \mathbb{R}^{+}, \\

\Phi_{\mathrm{Be},+}(z) = \Phi_{\mathrm{Be},-}(z) \begin{pmatrix}
1 & 0 \\ 1 & 1
\end{pmatrix}, & z \in e^{ -\frac{2\pi i}{3} }  \mathbb{R}^{+}. \\
\end{array}
\end{equation}
\item[(c)] As $z \to \infty$, $z \notin \Sigma_{\mathrm{Be}}$, we have
\begin{equation}\label{large z asymptotics Bessel}
\Phi_{\mathrm{Be}}(z) = ( 2\pi z^{\frac{1}{2}} )^{-\frac{\sigma_{3}}{2}}N
\left(I+\frac{ \Phi_{\mathrm{Be},1}}{z^{\frac{1}{2}}} + \bigO(z^{-1})\right) e^{2 z^{\frac{1}{2}}\sigma_{3}},
\end{equation}
where $\ds \Phi_{\mathrm{Be},1} = \frac{1}{16}\begin{pmatrix}
-1 & -2i \\ -2i & 1
\end{pmatrix}$.
\item[(d)] As $z$ tends to 0, the behavior of $\Phi_{\mathrm{Be}}(z)$ is
\begin{equation}\label{local behavior near 0 of P_Be}
\Phi_{\mathrm{Be}}(z) = \left\{ \begin{array}{l l}
\begin{pmatrix}
\bigO(1) & \bigO(\log z) \\
\bigO(1) & \bigO(\log z) 
\end{pmatrix}, & |\arg z| < \frac{2\pi}{3}, \\
\begin{pmatrix}
\bigO(\log z) & \bigO(\log z) \\
\bigO(\log z) & \bigO(\log z) 
\end{pmatrix}, & \frac{2\pi}{3}< |\arg z| < \pi.
\end{array}  \right.
\end{equation}
\end{itemize}
\begin{figure}[t]
    \begin{center}
    \setlength{\unitlength}{1truemm}
    \begin{picture}(100,55)(-5,10)
        \put(50,40){\line(-1,0){30}}
        \put(50,39.8){\thicklines\circle*{1.2}}
        \put(50,40){\line(-0.5,0.866){15}}
        \put(50,40){\line(-0.5,-0.866){15}}
        \put(50.3,36.8){$0$}
        \put(35,39.9){\thicklines\vector(1,0){.0001}}
        \put(41,55.588){\thicklines\vector(0.5,-0.866){.0001}}
        \put(41,24.412){\thicklines\vector(0.5,0.866){.0001}}
    \end{picture}
    \caption{\label{figBessel}The jump contour $\Sigma_{\mathrm{Be}}$ for $\Phi_{\mathrm{Be}}$.}
\end{center}
\end{figure}
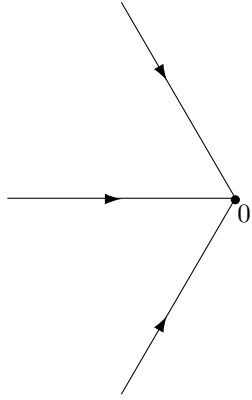
The unique solution to the above RH problem was obtained in \cite{KMcLVAV} and is given by 
\begin{equation}\label{Psi explicit}
\Phi_{\mathrm{Be}}(z)=
\begin{cases}
\begin{pmatrix}
I_{0}(2 z^{\frac{1}{2}}) & \frac{ i}{\pi} K_{0}(2 z^{\frac{1}{2}}) \\
2\pi i z^{\frac{1}{2}} I_{0}^{\prime}(2 z^{\frac{1}{2}}) & -2 z^{\frac{1}{2}} K_{0}^{\prime}(2 z^{\frac{1}{2}})
\end{pmatrix}, & |\arg z | < \frac{2\pi}{3}, \\

\begin{pmatrix}
\frac{1}{2} H_{0}^{(1)}(2(-z)^{\frac{1}{2}}) & \frac{1}{2} H_{0}^{(2)}(2(-z)^{\frac{1}{2}}) \\
\pi z^{\frac{1}{2}} \left( H_{0}^{(1)} \right)^{\prime} (2(-z)^{\frac{1}{2}}) & \pi z^{\frac{1}{2}} \left( H_{0}^{(2)} \right)^{\prime} (2(-z)^{\frac{1}{2}})
\end{pmatrix}, & \frac{2\pi}{3} < \arg z < \pi, \\

\begin{pmatrix}
\frac{1}{2} H_{0}^{(2)}(2(-z)^{\frac{1}{2}}) & -\frac{1}{2} H_{0}^{(1)}(2(-z)^{\frac{1}{2}}) \\
-\pi z^{\frac{1}{2}} \left( H_{0}^{(2)} \right)^{\prime} (2(-z)^{\frac{1}{2}}) & \pi z^{\frac{1}{2}} \left( H_{0}^{(1)} \right)^{\prime} (2(-z)^{\frac{1}{2}})
\end{pmatrix}, & -\pi < \arg z < -\frac{2\pi}{3},
\end{cases}
\end{equation}
where $H_{0}^{(1)}$ and $H_{0}^{(2)}$ are the Hankel functions of the first and second kind (of order $0$), and $I_0$ and $K_0$ are the modified Bessel functions of the first and second kind.

\vspace{0.2cm}\hspace{-0.55cm}It is easy to see from the properties (b) and (d) of the RH problem for $\Phi_{\mathrm{Be}}$ that in a neighborhood of $0$, we have
\begin{equation}\label{precise asymptotics of Phi Bessel near 0}
\Phi_{\mathrm{Be}}(z) = \Phi_{\mathrm{Be},0}(z)\begin{pmatrix}
1 & \frac{1}{2\pi i}\log z \\ 0 & 1
\end{pmatrix} \widetilde{H}_{0}(z),
\end{equation}
where $\Phi_{\mathrm{Be},0}$ is analytic in a neighborhood of $0$ and $\widetilde{H}_{0}$ is given by 
\begin{equation}\label{def of H}
\widetilde{H}_{0}(z) = \left\{  \begin{array}{l l}

I, & \mbox{ for } -\frac{2\pi}{3}< \arg(z)< \frac{2\pi}{3},\\

\begin{pmatrix}
1 & 0 \\
-1 & 1 \\
\end{pmatrix}, & \mbox{ for } \frac{2\pi}{3}< \arg(z)< \pi, \\

\begin{pmatrix}
1 & 0 \\
1 & 1 \\
\end{pmatrix}, & \mbox{ for } -\pi< \arg(z)< -\frac{2\pi}{3}.
\end{array} \right.
\end{equation}
Using the asymptotics of the Bessel functions near the origin (see e.g. \cite[Chapter 10.30(i)]{NIST}), we obtain after a computation that
\begin{equation}\label{precise matrix at 0 in asymptotics of Phi Bessel near 0}
\Phi_{\mathrm{Be},0}(0) = 
\begin{pmatrix}
1 & \frac{\gamma_{\mathrm{E}}}{\pi i} \\ 0 & 1
\end{pmatrix},
\end{equation}
where $\gamma_{\mathrm{E}}$ is Euler's gamma constant.
\subsection{Confluent hypergeometric model RH problem}\label{subsection: model RHP with HG functions}
\begin{itemize}
\item[(a)] $\Phi_{\mathrm{HG}} : \mathbb{C} \setminus \Sigma_{\mathrm{HG}} \rightarrow \mathbb{C}^{2 \times 2}$ is analytic, where $\Sigma_{\mathrm{HG}}$ is shown in Figure \ref{Fig:HG}.
\item[(b)] For $z \in \Gamma_{k}$ (see Figure \ref{Fig:HG}), $k = 1,\ldots,6$, $\Phi_{\mathrm{HG}}$ satisfies the jump relations
\begin{equation}\label{jumps PHG3}
\Phi_{\mathrm{HG},+}(z) = \Phi_{\mathrm{HG},-}(z)J_{k},
\end{equation}
where
\begin{align*}
& J_{1} = \begin{pmatrix}
0 & e^{-i\pi \beta} \\ -e^{i\pi\beta} & 0
\end{pmatrix}, \quad J_{4} = \begin{pmatrix}
0 & e^{i\pi\beta} \\ -e^{-i\pi\beta} & 0
\end{pmatrix}, \\
& J_{2} = \begin{pmatrix}
1 & 0 \\ e^{i\pi\beta} & 1
\end{pmatrix}\hspace{-0.1cm}, \hspace{-0.3cm} \quad J_{3} = \begin{pmatrix}
1 & 0 \\ e^{-i\pi\beta} & 1
\end{pmatrix}\hspace{-0.1cm}, \hspace{-0.3cm} \quad J_{5} = \begin{pmatrix}
1 & 0 \\ e^{-i\pi\beta} & 1
\end{pmatrix}\hspace{-0.1cm}, \hspace{-0.3cm} \quad J_{6} = \begin{pmatrix}
1 & 0 \\ e^{i\pi\beta} & 1
\end{pmatrix}.
\end{align*}
\item[(c)] As $z \to \infty$, $z \notin \Sigma_{\mathrm{HG}}$, we have
\begin{equation}\label{Asymptotics HG}
\Phi_{\mathrm{HG}}(z) = \left( I +  \frac{\Phi_{\mathrm{HG},1}(\beta)}{z} + \bigO(z^{-2}) \right) z^{-\beta\sigma_{3}}e^{-\frac{z}{2}\sigma_{3}}\left\{ \begin{array}{l l}
\displaystyle e^{i\pi\beta  \sigma_{3}}, & \displaystyle \frac{\pi}{2} < \arg z <  \frac{3\pi}{2}, \\
\begin{pmatrix}
0 & -1 \\ 1 & 0
\end{pmatrix}, & \displaystyle -\frac{\pi}{2} < \arg z < \frac{\pi}{2},
\end{array} \right.
\end{equation}
where $z^{\beta} = |z|^{\beta}e^{i\beta \arg z}$ with $\arg z \in (-\frac{\pi}{2},\frac{3\pi}{3})$ and
\begin{equation}\label{def of tau}
\Phi_{\mathrm{HG},1}(\beta) = \beta^{2} \begin{pmatrix}
-1 & \tau(\beta) \\ - \tau(-\beta) & 1
\end{pmatrix}, \qquad \tau(\beta) = \frac{- \Gamma\left( -\beta \right)}{\Gamma\left( \beta + 1 \right)}.
\end{equation}

As $z \to 0$, we have
\begin{equation}\label{lol 35}
\Phi_{\mathrm{HG}}(z) = \left\{ \begin{array}{l l}
\begin{pmatrix}
\bigO(1) & \bigO(\log z) \\
\bigO(1) & \bigO(\log z)
\end{pmatrix}, & \mbox{if } z \in II \cup V, \\
\begin{pmatrix}
\bigO(\log z) & \bigO(\log z) \\
\bigO(\log z) & \bigO(\log z)
\end{pmatrix}, & \mbox{if } z \in I\cup III \cup IV \cup VI.
\end{array} \right.
\end{equation}
\end{itemize}
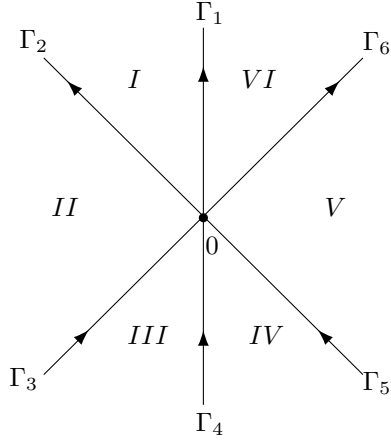
\begin{figure}[t!]
    \begin{center}
    \setlength{\unitlength}{1truemm}
    \begin{picture}(100,55)(-5,10)
             
        \put(50,39.8){\thicklines\circle*{1.2}}
        \put(50,40){\line(-0.5,0.5){21}}
        \put(50,40){\line(-0.5,-0.5){21}}
        \put(50,40){\line(0.5,0.5){21}}
        \put(50,40){\line(0.5,-0.5){21}}
        \put(50,40){\line(0,1){25}}
        \put(50,40){\line(0,-1){25}}
        
        \put(50.3,35){$0$}
        \put(71,62){$\Gamma_6$}        
        \put(49,66){$\Gamma_1$}        
        \put(25.8,62.3){$\Gamma_2$}        
        \put(24.5,17.5){$\Gamma_3$}
        \put(49,11.5){$\Gamma_4$}
        \put(71,17){$\Gamma_5$}        
        
        \put(32,58){\thicklines\vector(-0.5,0.5){.0001}}
        \put(35,25){\thicklines\vector(0.5,0.5){.0001}}
        \put(68,58){\thicklines\vector(0.5,0.5){.0001}}
        \put(65,25){\thicklines\vector(-0.5,0.5){.0001}}
        \put(50,60){\thicklines\vector(0,1){.0001}}
        \put(50,25){\thicklines\vector(0,1){.0001}}
        \put(40,57){$I$}
        \put(30,40){$II$}
        \put(40,23){$III$}
        \put(56,23){$IV$}
        \put(66,40){$V$}
        \put(55,57){$VI$}
    \end{picture}
    \caption{\label{Fig:HG}The jump contour $\Sigma_{\mathrm{HG}}$ for $\Phi_{\mathrm{HG}}$. The ray $\Gamma_{k}$ is oriented from $0$ to $\infty$, and forms an angle with $\mathbb{R}^{+}$ which is a multiple of $\frac{\pi}{4}$.}
\end{center}
\end{figure}
This model RH problem was solved explicitly in \cite{ItsKrasovsky}. Define
\begin{equation}\label{phi_HG}
\widehat{\Phi}_{\mathrm{HG}}(z) = \begin{pmatrix}
\Gamma(1 -\beta)G(\beta; z) & -\frac{\Gamma(1 -\beta)}{\Gamma(\beta)}H(1-\beta;ze^{-i\pi }) \\
\Gamma(1 +\beta)G(1+\beta;z) & H(-\beta;ze^{-i\pi })
\end{pmatrix},
\end{equation}
where $G$ and $H$ are related to the Whittaker functions:
\begin{equation}\label{relation between G and H and Whittaker}
G(a;z) = \frac{M_{\kappa,\mu}(z)}{\sqrt{z}}, \quad H(a;z) = \frac{W_{\kappa,\mu}(z)}{\sqrt{z}}, \quad \mu = 0, \quad \kappa = \frac{1}{2}-a.
\end{equation}
The solution $\Phi_{\mathrm{HG}}$ is given by
\begin{equation}\label{model RHP HG in different sector}
\Phi_{\mathrm{HG}}(z) = \left\{ \begin{array}{l l}
\widehat{\Phi}_{\mathrm{HG}}(z)J_{2}^{-1}, & \mbox{ for } z \in I, \\
\widehat{\Phi}_{\mathrm{HG}}(z), & \mbox{ for } z \in II, \\
\widehat{\Phi}_{\mathrm{HG}}(z)J_{3}^{-1}, & \mbox{ for } z \in III, \\
\widehat{\Phi}_{\mathrm{HG}}(z)J_{2}^{-1}J_{1}^{-1}J_{6}^{-1}J_{5}, & \mbox{ for } z \in IV, \\
\widehat{\Phi}_{\mathrm{HG}}(z)J_{2}^{-1}J_{1}^{-1}J_{6}^{-1}, & \mbox{ for } z \in V, \\
\widehat{\Phi}_{\mathrm{HG}}(z)J_{2}^{-1}J_{1}^{-1}, & \mbox{ for } z \in VI. \\
\end{array} \right.
\end{equation}
The asymptotics of $M_{\kappa,\mu}(z)$ and $W_{\kappa,\mu}(z)$ as $z \to 0$ given by \cite[Subsection 13.14 (iii)]{NIST} allow to obtain a more precise version of \eqref{lol 35}. Using also $\Gamma(z)\Gamma(1-z) = \frac{\pi}{\sin (\pi z)} = -\Gamma(-z)\Gamma(1+z)$, we get
\begin{equation}\label{precise asymptotics of Phi HG near 0}
\Phi_{\mathrm{HG}}(z) = \widehat{\Phi}_{\mathrm{HG}}(z) = \begin{pmatrix}
\Psi_{11} & \Psi_{12} \\ \Psi_{21} & \Psi_{22}
\end{pmatrix} (I + \bigO(z)) \begin{pmatrix}
1 & \frac{\sin (\pi \beta)}{\pi} \log z \\
0 & 1
\end{pmatrix}, \qquad \mbox{as } z \to 0, \, z \in II,
\end{equation}
where
\begin{equation}
\log z = \log |z| + i \arg z, \qquad \arg z \in \Big(-\frac{\pi}{2},\frac{3\pi}{2}\Big),
\end{equation}
and
\begin{align*}
& \Psi_{11} = \Gamma(1-\beta), \qquad \Psi_{12} = \frac{1}{\Gamma(\beta)} \left( \frac{\Gamma^{\prime}(1-\beta)}{\Gamma(1-\beta)}+2\gamma_{\mathrm{E}} - i \pi \right), \\
& \Psi_{21} = \Gamma(1+\beta), \qquad \Psi_{22} = \frac{-1}{\Gamma(-\beta)} \left( \frac{\Gamma^{\prime}(-\beta)}{\Gamma(-\beta)} + 2\gamma_{\mathrm{E}} - i \pi \right),
\end{align*}
and where $\gamma_{\mathrm{E}}$ is Euler's gamma constant.
\section*{Acknowledgements}
This work was supported by the Swedish Research Council, Grant No. 2015-05430 and the European Research Council, Grant Agreement No. 682537. The author is grateful to the anonymous referees for helpful remarks and for their careful reading of the manuscript. The author also wishes to thank Jonatan Lenells for a careful reading of the introduction.


\begin{thebibliography}{99}
\bibitem{AGZ2010} G.W. Anderson, A. Guionnet and O. Zeitouni, {\em An introduction to random matrices}, Cambridge Studies in Advanced Mathematics \textbf{118}, Cambridge University Press, Cambridge, 2010.

\bibitem{BW1983} E. Basor and H. Widom, Toeplitz and Wiener-Hopf determinants with piecewise continuous symbols, \textit{ J. Funct. Anal.} \textbf{50} (1983), 387--413.

\bibitem{BleherIts} P. Bleher and A. Its, Semiclassical asymptotics of orthogonal polynomials, Riemann-Hilbert problem, and universality in the matrix model, \textit{Ann. of Math.} \textbf{150} (1999), 185--266.

\bibitem{BohigasPato1} O. Bohigas and M.P. Pato, Missing levels in correlated spectra, {\em Phys. Lett. B} {\bf 595} (2004), 171--176.

\bibitem{BohigasPato2} O. Bohigas and M.P. Pato, Randomly incomplete spectra and intermediate statistics,  \textit{Phys. Rev. E} (3) {\bf 74} (2006).

\bibitem{Bornemann} F. Bornemann, On the numerical evaluation of Fredholm determinants, \textit{Math. Comp.} \textbf{79} (2010), 871--915.

\bibitem{Borodin} A. Borodin, \textit{Determinantal point processes}, The Oxford handbook of random matrix theory, 231--249, Oxford Univ. Press, Oxford, 2011.

\bibitem{BDIK2015} T. Bothner, P. Deift, A. Its, I. Krasovsky, On the asymptotic behavior of a log gas in the bulk scaling limit in the presence of a varying external potential I, \textit{Comm. Math. Phys.} \textbf{337} (2015), 1397--1463.

\bibitem{BDIK2017} T. Bothner, P. Deift, A. Its, I. Krasovsky, On the asymptotic behavior of a log gas in the bulk scaling limit in the presence of a varying external potential II, \textit{Oper. Theory Adv. Appl.} \textbf{259} (2017).

\bibitem{BDIK2019} T. Bothner, P. Deift, A. Its, I. Krasovsky, The sine process under the influence of a varying potential, \textit{J. Math. Phys.} \textbf{59}, 091414 (2018), doi: 10.1063/1.505039.

\bibitem{BIP} T. Bothner, A. Its and A. Prokhorov, On the analysis of incomplete spectra in random matrix theory through an extension of the Jimbo-Miwa-Ueno differential, \textit{Adv. Math.} \textbf{345} (2019), 483--551.

\bibitem{BB1995} A.M. Budylin and V.S. Buslaev, Quasiclassical asymptotics of the resolvent of an integral convolution operator with a sine kernel on a finite interval, \textit{Algebra i Analiz} \textbf{7} (1995), 79--103.

\bibitem{BufetovConditional} A.I. Bufetov, Conditional measures of determinantal point processes, arXiv:1605.01400.

\bibitem{Charlier} C. Charlier, Asymptotics of Hankel determinants with a one-cut regular potential and Fisher-Hartwig singularities, \textit{Int. Math. Res. Not.} \textbf{2019} (2019), 7515--7576.

\bibitem{Ch2018} C. Charlier, Exponential moments and piecewise thinning for the Bessel point process, \textit{Int. Math. Res. Not.} (2020), rnaa054, https://doi.org/10.1093/imrn/rnaa054.

\bibitem{ChCl2} C. Charlier and T. Claeys, Thinning and conditioning of the Circular Unitary Ensemble, \textit{Random Matrices Theory Appl.} \textbf{6} (2017), 51 pp.

\bibitem{ChCl3} C. Charlier and T. Claeys, Large gap asymptotics for Airy kernel determinants with discontinuities, \textit{Comm. Math. Phys.} \textbf{375} (2020), 1299--1339.

\bibitem{ChDoe} C. Charlier and A. Doeraene, The generating function for the Bessel point process and a system of coupled Painlev\'{e} V equations, \textit{Random Matrices Theory Appl.} \textbf{8} (2019), 31pp.

\bibitem{ClaeysDoeraene} T. Claeys and A. Doeraene, The generating function for the Airy point process and a system of coupled Painlev\'e II equations, \textit{Stud. Appl. Math.} \textbf{140} (2018), 403--437.

\bibitem{CostinLebowitz} O. Costin and J.L. Lebowitz, Gaussian fluctuation in random matrices, \textit{Phys. Rev. Lett.} \textbf{75} (1995), 69--72.

\bibitem{Deift} P. Deift,
{\em Orthogonal Polynomials and Random Matrices: A Riemann-Hilbert Approach}, Amer. Math. Soc. \textbf{3} (2000).

\bibitem{DIKZ2007} P. Deift, A. Its, I. Krasovsky and X. Zhou, The Widom-Dyson constant for the gap probability in random matrix theory, \textit{J. Comput. Appl. Math.} \textbf{202} (2007), 26--47.
%


\bibitem{DeiftItsZhou} P. Deift, A. Its, and X. Zhou, A Riemann-Hilbert approach to asymptotic problems arising in the theory of random matrix models, and also in the theory of integrable statistical mechanics, \emph{Ann. Math.} \textbf{278} (1997), 149--235.
%
\bibitem{DKMVZ2}
P. Deift, T. Kriecherbauer, K.T-R McLaughlin, S. Venakides and X. Zhou, Uniform asymptotics for polynomials orthogonal with respect to varying exponential weights and applications to universality questions in random matrix theory, {\em Comm. Pure Appl. Math.} {\bf 52} (1999), 1335--1425.
%
\bibitem{DKMVZ1}
P. Deift, T. Kriecherbauer, K.T-R McLaughlin, S. Venakides and X. Zhou, Strong asymptotics of orthogonal polynomials with respect to exponential weights, {\em Comm. Pure Appl. Math.} {\bf 52} (1999), 1491--1552.
%
\bibitem{DeiftZhou1992} P. Deift and X. Zhou, A steepest descent method for oscillatory Riemann-Hilbert problems, \textit{ Bull. Amer. Math. Soc. (N.S.)} \textbf{26} (1992), 119--123.

\bibitem{DysonConjecture} F.J. Dyson, Statistical theory of energy levels of complex systems I, \textit{J. Math.Phys.} \textbf{3}, 140--156 (1962).

\bibitem{Dyson} F.J. Dyson, Fredholm determinants and inverse scattering problems, \textit{Comm. Math. Phys.} \textbf{47} (1976), 171--183.

\bibitem{Ehr sine} T. Ehrhardt, Dyson's constant in the asymptotics of the Fredholm determinant of the sine kernel, \textit{Comm. Math. Phys.} \textbf{262} (2006), 317--341.

\bibitem{Erdos} L. Erd\"{o}s, Universality of Wigner random matrices: a survey of recent results, \textit{Russian Math. Surveys} \textbf{66} (2011), 507--626.

\bibitem{EPRSY2010}  L. Erd\"{o}s, S. P\'{e}ch\'{e}, J.A. Ram\'{i}rez, B. Schlein and H.-T. Yau, Bulk universality for Wigner matrices, \textit{Comm. Pure Appl. Math.} \textbf{63} (2010), 895--925.

\bibitem{FahsKrasovsky} B. Fahs and I. Krasovsky, Splitting of a gap in the bulk of the spectrum of random matrices, \textit{Duke Math J.} \textbf{168}, 3529--3590.

\bibitem{FahsKrasovsky2} B. Fahs and I. Krasovsky, Sine-kernel determinant on two large intervals, arXiv:2003.08136.

%
%
%
%
\bibitem{FoulquieMartinezSousa} A. Foulquie Moreno, A. Martinez-Finkelshtein, and V. L. Sousa, Asymptotics of orthogonal polynomials for a weight with a jump on [-1,1], {\em Constr. Approx.} {\bf 33} (2011), 219--263.

\bibitem{Ghosh} S. Ghosh, Determinantal processes and completeness of random exponentials: the critical case, \textit{Probab. Theory Related Fields} \textbf{163} (2015),643--665.

\bibitem{HolcombPaquette} D. Holcomb and E. Paquette, The maximum deviation of the $\mathrm{Sine}_{\beta}$ counting process, \textit{ Electron. Commun. Probab.} \textbf{23} (2018), 13 pp.


\bibitem{IIKS} A. Its, A.G. Izergin, V.E. Korepin and N.A. Slavnov, Differential equations for quantum correlation functions, \emph{In proceedings of the Conference on Yang-Baxter Equations, Conformal Invariance and Integrability in Statistical Mechanics and Field Theory}, Volume \textbf{4}, (1990) 1003--1037.

\bibitem{ItsKrasovsky} A. Its and I. Krasovsky, Hankel determinant and orthogonal polynomials for the Gaussian weight with a jump, \textit{Contemporary Mathematics} \textbf{458} (2008), 215--248.

\bibitem{JMMS1980} M. Jimbo, T. Miwa, Y. M\^{o}ri and M. Sato, Density matrix of an impenetrable Bose gas and the fifth Painlevé transcendent, \textit{Phys. D} \textbf{1} (1980), 80--158.

\bibitem{JohanssonUniversality} K. Johansson, Universality of the local spacing distribution in certain ensembles of Hermitian Wigner matrices, \textit{Comm. Math. Phys.} \textbf{215} (2001), 683--705.

\bibitem{Johansson} K. Johansson. {\em Random matrices and determinantal processes}, Mathematical statistical physics, 1--55, Elsevier B.V., Amsterdam, 2006.


\bibitem{Krasovsky} I. Krasovsky, Gap probability in the spectrum of random matrices and asymptotics of polynomials orthogonal on an arc of the unit circle, \textit{Int. Math. Res. Not.} \textbf{2004} (2004), 1249--1272.

\bibitem{KuijlaarsUniversality} A.B.J. Kuijlaars, Universality, In \textit{The Oxford handbook of random matrix theory} (2011), pages 103--134, Oxford Univ. Press, Oxford.


\bibitem{KMcLVAV} A.B.J. Kuijlaars, K.T.--R. McLaughlin, W. Van Assche and M. Vanlessen, The Riemann-Hilbert approach to strong asymptotics for orthogonal polynomials on $[-1,1]$, {\it Adv. Math.} \textbf{188} (2004), 337--398.

\bibitem{KuijlaarsDiaz} A.B.J. Kuijlaars and E. Mi\~{n}a-D\'{i}az.  Universality for conditional measures of the sine point process, \textit{J. Approx. Theory}, \textbf{243} (2019), 1--24. 
%
%
\bibitem{LMR}F. Lavancier, J. Moller and E. Rubak, Determinantal point process models and statistical inference: Extended version, {\em J. Royal Stat. Soc.: Series B} {\bf 77} (2015), no. 4, 853--877.
%
%

\bibitem{NIST} F.W.J. Olver, A.B. Olde Daalhuis, D.W. Lozier, B.I. Schneider, R.F. Boisvert, C.W. Clark, B.R. Miller and B.V. Saunders, NIST Digital Library of Mathematical Functions. http://dlmf.nist.gov/, Release 1.0.13 of 2016-09-16.

\bibitem{Lub2016} D.S. Lubinsky, An update on local universality limits for correlation functions generated by unitary ensembles, \textit{SIGMA Symmetry Integrability Geom. Methods Appl.} \textbf{12} (2016), 36 pp.


\bibitem{PS1997} L. Pastur and M. Shcherbina, Universality of the local eigenvalue statistics for a class of unitary invariant random matrix ensembles, \textit{J. Statist. Phys.} \textbf{86} (1997),109--147.

\bibitem{PS2008} L. Pastur and M. Shcherbina, Bulk universality and related properties of Hermitian matrix models, \textit{J. Stat. Phys.} \textbf{130} (2008), 205--250.
%

\bibitem{Soshnikov} A. Soshnikov, Determinantal random point fields, \textit{ Russian Math. Surveys} \textbf{55} (2000), no. 5, 923--975.
%
\bibitem{TaoVu} T. Tao, and V. Vu, Random matrices:  Universality of the local eigenvalue statistics, \textit{Acta Math.}, \textbf{206} (2011), 127--204.

\bibitem{TaoVu2} T. Tao T and V. Vu, Random matrices:  the universality phenomenon for Wigner ensembles, in Modern Aspects of Random Matrix Theory, Proc. Sympos. Appl. Math. \textbf{72}, Amer. Math. Soc., Providence, RI, 2014, 121--172.


\bibitem{Widom1995} H. Widom, Asymptotics for the Fredholm determinant of the sine kernel on a union of intervals, \textit{ Comm. Math. Phys.} \textbf{171} (1995), 159--180.

\end{thebibliography}
\end{document}